\documentclass[aoas,preprint]{imsart}
\usepackage{xr}
\RequirePackage[OT1]{fontenc}
\RequirePackage[numbers]{natbib}
\RequirePackage[colorlinks,citecolor=blue,urlcolor=blue]{hyperref}
\usepackage{amsthm,amsmath,enumerate,amssymb,bbm,graphicx,mathrsfs,mathtools,enumitem,rotating,bm,txfonts,units,algorithmicx,caption,subfiles,subcaption,booktabs,psfrag,epsf}

\newenvironment{changemargin}[2]{%
\begin{list}{}{%
\setlength{\topsep}{0pt}%
\setlength{\leftmargin}{#1}%
\setlength{\rightmargin}{#2}%
\setlength{\listparindent}{\parindent}%
\setlength{\itemindent}{\parindent}%
\setlength{\parsep}{\parskip}%
}%
\item[]}{\end{list}}

\makeatletter
\newcommand*\bigcdot{\mathpalette\bigcdot@{1.0}}
\newcommand*\bigcdot@[2]{\mathbin{\vcenter{\hbox{\scalebox{#2}{$\m@th#1*$}}}}}
\makeatother

\startlocaldefs
\numberwithin{equation}{section}
\theoremstyle{plain}
\newtheorem{theorem}{Theorem}[section]
\newtheorem{lemma}{Lemma}[section]
\newtheorem{corollary}{Corollary}[section]
\newtheorem{algorithm}{Algorithm}[section]
\newtheorem{remark}{Remark}[section]

\newtheorem{assumption}{Assumption}[section]
\endlocaldefs

\makeatletter
\newcommand*{\centernot}{%
  \mathpalette\@centernot
}
\def\@centernot#1#2{%
  \mathrel{%
    \rlap{%
      \settowidth\dimen@{$\m@th#1{#2}$}%
      \kern.5\dimen@
      \settowidth\dimen@{$\m@th#1=$}%
      \kern-.5\dimen@
      $\m@th#1\not$%
    }%
    {#2}%
  }%
}
\makeatother
\newcommand{\independent}{\perp\mkern-9.5mu\perp}

\def\T{{ \mathrm{\scriptscriptstyle T} }}

\DeclareMathOperator{\Tr}{Tr}

\DeclareMathOperator{\Prob}{\mathbb{P}}
\DeclareMathOperator{\Bprob}{pr}
\DeclareMathOperator{\diag}{diag}

\DeclareMathOperator{\qvalue}{\mathit{q}}
\DeclareMathOperator{\sraw}{sR_{AW}}
\DeclareMathOperator{\interest}{int}
\DeclareMathOperator{\nuisance}{nuis}
\DeclareMathOperator{\Observed}{\mathcal{S}}
\DeclareMathOperator{\Missing}{\mathcal{M}}
\DeclareMathOperator{\Missingsub}{\mathcal{M}_1}
\DeclareMathOperator{\normal}{\mathcal{N}}
\DeclareMathOperator{\miss}{miss}
\DeclareMathOperator{\GMM}{(GMM)}

\DeclareMathOperator*{\argmin}{arg\,min}
\DeclareMathOperator*{\argmax}{arg\,max}
\newcommand{\indep}{\rotatebox[origin=c]{90}{$\models$}}
\DeclareMathOperator{\E}{\mathbb{E}}
\DeclareMathOperator{\V}{\mathbb{V}}
\DeclareMathOperator{\Data}{\mathcal{D}}
\DeclareMathOperator{\tdist}{\stackrel{d}{\to}}
\DeclareMathOperator{\C}{Cov}
\DeclareMathOperator{\Corr}{Corr}
\DeclareMathOperator{\vecM}{vec}
\DeclareMathOperator{\asim}{\stackrel{\cdot}{\sim}}

\DeclareMathOperator{\edist}{\stackrel{\text{\scriptsize \textit{d}}}{=}}
\DeclareMathOperator{\im}{Im}
\DeclarePairedDelimiter\abs{\lvert}{\rvert}%
\DeclarePairedDelimiter\norm{\lVert}{\rVert}%
\newcommand{\matbm}[1]{\bm{#1}}
\newcommand{\vecbm}[1]{\bm{#1}}

\makeatletter
\let\oldabs\abs
\def\abs{\@ifstar{\oldabs}{\oldabs*}}
\let\oldnorm\norm
\def\norm{\@ifstar{\oldnorm}{\oldnorm*}}
\makeatother

\newcommand{\LtabMetab}{
 \begin{minipage}{\linewidth}
   \centering
   \footnotesize
   \vspace{0.25cm}
   \captionof{table}[Parameters used to simulate latent confound effects.]{The $\pi_k$ and $\tau_k$ values used to simulate $\bm{\ell}_{1},\ldots,\bm{\ell}_{p}$ ($k=1,\ldots,10$).}\label{Table:LMetab}
	\begin{tabular}{c | c | c | c | c | c | c | c | c | c | c}
 Factor number ($k$) & $1$ & $2$ & $3$ & $4$ & $5$ & $6$ & $7$ & $8$ & $9$ & $10$\\\hline
  $\pi_k$ & 0 & 0 & 0.76 & 0.56 & 0.48 & 0.32 & 0.28 & 0.20 & 0.20 & 0.20\\\hline
  $\tau_k$ & 0.78 & 0.57 & 0.5 & 0.5 & 0.5 & 0.5 & 0.5 & 0.5 & 0.5 & 0.5
  \end{tabular}
  \vspace{0.4cm}
  \end{minipage}
}

\newcommand{\Fractab}{
 \begin{minipage}{\linewidth}
   \centering
   \vspace{0.25cm}
   \captionof{table}{The expected number of metabolites in each missing data bin for data simulated according to \eqref{equation:MetabMiss:Simulation} with $\Psi(x)=\exp(x)/\left\lbrace 1 + \exp(x) \right\rbrace$, where $f$ is the frequency of missing data.}\label{Table:SimFrac}
	\begin{tabular}{c | c | c | c }
$f = 0$ & $0 < f \leq 0.05$ & $0.05 < f \leq 0.5$ & $0.5 < f$\\\hline
251.6 & 233.6 & 298.3 & 416.4
  \end{tabular}
  \vspace{0.4cm}
  \end{minipage}
}

\newcommand{\COPSACMiss}{
 \begin{minipage}{\linewidth}
   \centering
   \vspace{0.25cm}
   \captionof{table}{The number of metabolites in each missing data bin in the blood plasma metabolomic data, where $f$ is as defined in Table \ref{Table:SimFrac}.}\label{Table:COPSACMiss}
	\begin{tabular}{c | c | c | c }
$f = 0$ & $0 < f \leq 0.05$ & $0.05 < f \leq 0.5$ & $0.5 < f$\\\hline
400 & 256 & 300 & 182
  \end{tabular}
  \vspace{0.4cm}
  \end{minipage}
}

\begin{document}

\begin{frontmatter}
\title{Estimation and inference in metabolomics with non-random missing data and latent factors}
\runtitle{MetabMiss}

\begin{aug}
\author{\fnms{Chris} \snm{McKennan$^{1,}$}\thanksref{t2}\ead[label=e1]{chm195@pitt.edu}},
\author{\fnms{Carole} \snm{Ober$^2$}
\ead[label=e3]{c-ober@bsd.uchicago.edu}}
\and
\author{\fnms{Dan} \snm{Nicolae$^2$}\ead[label=e2]{nicolae@galton.uchicago.edu}}

\thankstext{t2}{Supported in part by NIH grant R01 HL129735.}
\runauthor{C. McKennan et al.}

\affiliation{University of Pittsburgh$^1$ and University of Chicago$^2$}

\address{Department of Statistics\\
University of Pittsburgh\\
Pittsburgh, PA 15260\\
\printead{e1}}

\address{Department of Human Genetics\\
University of Chicago\\
Chicago, IL 60637\\
\printead{e3}}

\address{Department of Statistics\\
University of Chicago\\
Chicago, IL 60637\\
\printead{e2}}
\end{aug}

\begin{abstract}
High throughput metabolomics data are fraught with both non-ignorable missing observations and unobserved factors that influence a metabolite's measured concentration, and it is well known that ignoring either of these complications can compromise estimators. However, current methods to analyze these data can only account for the missing data or unobserved factors, but not both. We therefore developed MetabMiss, a statistically rigorous method to account for both non-random missing data and latent factors in high throughput metabolomics data. Our methodology does not require the practitioner specify a probability model for the missing data, and makes investigating the relationship between the metabolome and tens, or even hundreds, of phenotypes computationally tractable. We demonstrate the fidelity of MetabMiss's estimates using both simulated and real metabolomics data.
\end{abstract}


\end{frontmatter}


\section{Introduction}
\label{section:Introduction}
Metabolomics is the study of tissue- or body fluid-specific small molecule metabolites, and has the potential to lead to new insights into the origin of human disease \citep{MetabDisease,PiperineAsthma3,PiperineAsthma2} and drug metabolism \citep{DrugMetab1,DrugMetab2}. Recent advances in both liquid chromatography (LC) and untargeted mass spectrometry (MS) have made it possible to identify and quantify hundreds to thousands of metabolites per sample \citep{MSMetab}. Similar to high throughput gene expression, proteomic and DNA methylation data, these data contain systematic technical and biological variation whose sources are not observed by the practitioner \citep{Metab_RRmix}. However, what makes untargeted LC-MS metabolomic data particularly challenging is the vast amount of missing data, nearly all of which is missing not at random due to an unknown, metabolite-specific missingness mechanism in which more abundant and ionizable analytes are more likely to be observed \citep{LODMetabolomics}. For instance, 22\% of all $\text{\#metabolites} \times \text{\#samples} = 1138 \times 533$ observations were missing from our data example in Section \ref{section:DataAnlysis}, where Figure \ref{Figure:TechRep} shows that only analytes with the strongest signals were likely to be quantified in all technical replicates.\par 
\indent There are several methods that attempt to account for either latent covariates \citep{RUVMetab1,RUVMetab2,Metab_RRmix} or non-random missing data \citep{MNAR_SILAC,SILAC_MNAR2,MNARProteomics} when trying to infer the relationship between the metabolome and a variable of interest. However, these are not amenable to real, untargeted metabolomic data because the former set of methods cannot accommodate non-ignorable missing data, and the latter set ignores latent covariates that can bias estimators. Surprisingly, to the best of our knowledge, \cite{ImputeWithC} is the only work to even acknowledge the challenge of accounting for both. However, they propose imputing missing data with an arbitrary limit of detection, and require prior knowledge of a set of control metabolites whose concentrations are unrelated to the variable of interest to estimate latent factors.\par 
\indent Given the paucity of methods to analyze untargeted metabolomic data, we developed MetabMiss, the first method to account for both latent covariates and non-random missing data that does not rely on control metabolites, internal standards or erringly imputing missing data. Our method also offers the following advantages: 
\begin{enumerate}[label=(\alph*)]
    \item We do not require knowledge of the underlying probability distribution of the missing data.
    \item We modularize our method so that the metabolite-dependent missingness mechanisms are estimated only once per dataset, which makes computation on the order of a phenome wide association study tractable.\label{item:advantage:speed}
\end{enumerate}
And while we assume the functional form of the missing data mechanism is known, we provide a method to access the veracity of said function for each metabolite.\par 
\indent As far as we are aware, our estimators for the missingness mechanism are also the first estimators, among those designed for mass spectrometry data, that satisfy Property \ref{item:advantage:speed} and do not depend on the covariate(s) of interest. This makes analyzing modern metabolomic data tractable, as practitioners are often interested in understanding the relationship between metabolite concentration and many different covariates of interest due to the wealth of information available for each for sample. We discuss this further in Section \ref{subsection:RoadMap:IVGMM}.\par 
\indent The remainder of the manuscript is organized as follows: we give a mathematical description of the data in Section \ref{section:SetUp} and give an overview of our method in Section \ref{section:RoadMap}. We describe how we estimate the metabolite-dependent missingness mechanisms, estimate the coefficients of interest in a linear model and recover latent factors in Sections \ref{section:HBGMM}, \ref{section:IPW} and \ref{section:CMNAR}. We conclude by illustrating how our method performs in simulated and real metabolomic data in Sections \ref{section:simulations} and \ref{section:DataAnlysis}. An R package that implements MetabMiss can be installed from \href{https://github.com/chrismckennan/MetabMiss}{github.com/chrismckennan/MetabMiss}.

\begin{figure}[!ht]
    \centering
    \includegraphics[scale=0.4]{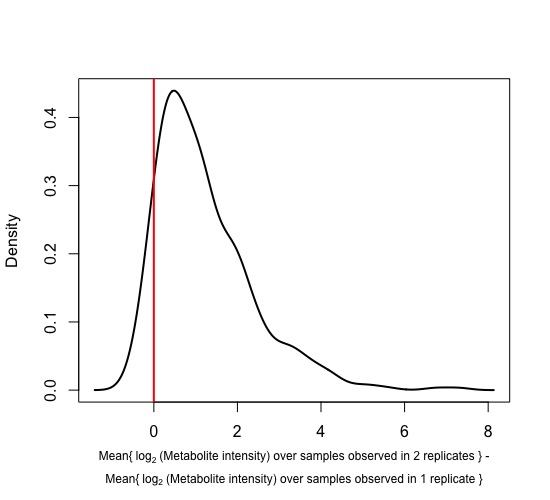}\caption{A density plot of the differences in mean observed metabolite $\log_2$-intensity between samples with observations in both technical replicates and those from samples with only 1 observation among the two replicates. Replicate pairs were obtained by running 20 biological samples from our motivating data example twice on the same mass spectrometer.}
    \label{Figure:TechRep}
\end{figure}

\section{Notation and problem set-up}
\label{section:SetUp}
\subsection{Notation}
\label{subsection:Notation}
Let $n > 0$ be an integer. We let $\vecbm{1}_n, \vecbm{0}_n \in \mathbb{R}^n$ be the vectors of all ones and zeros,  $I_n \in \mathbb{R}^{n \times n}$ to be the identity matrix, $[n]=\left\lbrace 1,\ldots,n \right\rbrace$ and $\vecbm{x}_i$ to be the $i$th element of $\vecbm{x} \in \mathbb{R}^n$. For $\matbm{M} \in \mathbb{R}^{n \times m}$, we let $\matbm{M}_{ij}$ be the $(i,j)$th element of $\bm{M}$, and define $P_{M}$ and $P_{M}^{\perp}$ to be the orthogonal projection matrices onto $\im\left( \bm{M}\right) = \left\lbrace \bm{M}\bm{v} : \bm{v} \in \mathbb{R}^m \right\rbrace$ and the null space of $\bm{M}^{\T}$. We also let $\matbm{X} \asim \left( \matbm{\mu}, \matbm{G}\right)$ if $\E\left( \matbm{X}\right) = \matbm{\mu}$ and $\V\left( \matbm{X}\right)=\matbm{G}$, $\matbm{X} \sim MN_{m \times n}\left(\bm{\mu},\bm{V},\bm{U}\right)$ if $\matbm{X}\in\mathbb{R}^{m\times n}$ and $\vecM\left(\matbm{X}\right)\sim N_{mn}\left(\vecM\left(\matbm{\mu}\right),\bm{U}\otimes \bm{V}\right)$ and lastly define $F_{\nu}(x)$ to be the cumulative distribution function of the t-distribution with $\nu > 0$ degrees of freedom.

\subsection{A description of and model for the data}
\label{subsection:ModelDescription}
Let $y_{gi}$ be the observed or unobserved log-transformed metabolite integrated intensity for metabolite $g \in [p]$ in sample $i \in [n]$, where the mass spectrometer intensity, integrated over time and mass-to-charge ratio, is proportional to a metabolite's concentration \citep{MassSpec_Proteomics}. Let $\bm{X}=\left( \bm{x}_1 \cdots \bm{x}_n \right)^{\T} \in \mathbb{R}^{n \times d}$ and $\bm{C} = \left( \bm{c}_1 \cdots \bm{c}_n \right)^{\T} \in \mathbb{R}^{n \times K}$ be observed and unobserved covariates (i.e. latent factors), where the former may contain biological factors like disease status, as well as technical factors like observed batch variables. We assume
\begin{align}
\label{equation:ModelForData}
y_{gi} = \vecbm{x}_i^{\T} \vecbm{\beta}_g + \vecbm{c}_i^{\T} \vecbm{\ell}_g + e_{gi}, \quad e_{gi} \asim \left( 0, \sigma_g^2 \right), \quad g \in [p]; i \in [n],
\end{align}
where our goal is to estimate $\vecbm{\beta}_g$. We assume the residuals $\left\lbrace e_{gi} \right\rbrace_{g\in[p], i\in[n]}$ are independent and $\left\lbrace e_{gi} \right\rbrace_{i\in[n]}$ are identically distributed for each $g \in [p]$. The unobserved covariates $\bm{c}_i$ can confound the relationship between $\bm{x}_i$ and $y_{gi}$, and also induce dependencies between different metabolites. We assume that $\bm{c}_1,\ldots,\bm{c}_n$ are independent and are independent of $\left\lbrace e_{gi} \right\rbrace_{g\in[p], i\in[n]}$. We do not assume an explicit probability model for $y_{gi}$ in order to avoid assuming a distribution for the missing data.\par
\indent We next define the indicator variable $r_{gi} = I\left( \text{$y_{gi}$ is observed}\right)$ and assume that for some known increasing cumulative distribution function $\Psi(x)$,
\begin{align}
\label{equation:MissingMech}
    \Prob\left( r_{gi} = 1 \mid y_{gi} \right) = \Psi\left\lbrace \alpha_g \left(y_{gi} - \delta_g\right) \right\rbrace, \quad g \in [p]; i \in [n].
\end{align}
The metabolite-dependent scale and location parameters $\alpha_g > 0$ and $\delta_g \in \mathbb{R}$ are such that $\alpha_g \searrow 0$ implies the mechanism is missing completely at random (MCAR) and $\alpha_g \nearrow \infty$ implies $y_{gi}$ is left-censored at $\delta_g$. Model \eqref{equation:MissingMech} is consistent with Figure \ref{Figure:TechRep} and previous observations that metabolites with smaller intensities are less likely to be observed \citep{LODMetabolomics,MassSpec_Proteomics}, and is a classic model for missing data in untargeted mass spectrometry experiments \citep{MNAR_SILAC,MNARProteomics,SILAC_MNAR2}. Typical values for $\Psi$ include the logistic function, an exponential probabilistic model \citep{MNAR_SILAC,SILAC_MNAR2} and the probit function \citep{MNARProteomics}. However, we observed in simulations that $\Psi(x)=F_{4}(x)$ is a more robust option, since its heavy tails make it less sensitive to outliers. This has previously been used as a robust alternative to logistic and probit functions \citep{WhyTdist1}.\par
\indent Implicit in \eqref{equation:MissingMech} is the assumption that conditional on $\matbm{Y} = \left(y_{gi}\right)_{g\in[p], i\in[n]} \in \mathbb{R}^{p \times n}$, $\left\lbrace r_{gi}\right\rbrace_{g\in[p], i\in[n]}$ are independent. This is likely only approximately true, since other intense analytes can preclude MS/MS fragmentation in data dependent mass spectrometry experiments. However, properly tuning the dynamic exclusion time can substantially mitigate any dependence \citep{DynamicExclusion}.

\section{A road map of our methodology}
\label{section:RoadMap}
Here we provide a compendious description of our method to estimate the metabolite-dependent missingness mechanisms, recover $\bm{C}$ and estimate $\bm{\beta}_1,\ldots,\bm{\beta}_p$. We delineate these steps in more detail in Sections \ref{section:HBGMM}, \ref{section:IPW} and \ref{section:CMNAR}.

\subsection{IV-GMM to estimate $\alpha_{g}$ and $\delta_{g}$}
\label{subsection:RoadMap:IVGMM}
We first estimate $\alpha_g$ and $\delta_g$ for metabolites $g$ with missing data. Unlike existing methods designed for untargeted mass spectrometry data whose estimates for the missingness mechanism depend on the user-specified $\matbm{X}$ \citep{MNAR_SILAC,MNARProteomics,SILAC_MNAR2}, our estimates only depend on $\bm{Y}$, and therefore only need to be estimated once per data matrix $\bm{Y}$. This makes analyzing modern datasets tractable, as practitioners typically collect a wealth of covariate information for each sample $i$, and will therefore need to infer the relationship between $\matbm{Y}$ and $\matbm{X}$ for many different covariate matrices $\matbm{X}$.\par
\indent Since the probability model for $\matbm{Y}$ is unknown, we build upon \cite{GMM_MNAR} and use instrumental variable generalized method of moments (IV-GMM) to estimate $\alpha_g$ and $\delta_g$. Fix a $g \in [p]$ and let $\vecbm{A}_1,\ldots,\vecbm{A}_n \in \mathbb{R}^s$ be random vectors such that $r_{gi} \indep \vecbm{A}_i \mid y_{gi}$ for all $i \in [n]$. Then \cite{GMM_MNAR} considers the following observable $s+1$ dimensional function for metabolite $g$:
\begin{align}
\label{equation:MomentFunction}
    &\bm{h}\left\lbrace \left(y_{gi},r_{gi},\bm{A}_i\right), \left(\alpha,\delta\right)\right\rbrace = \left(1 \, \vecbm{A}_i^{\T}\right)^{\T}\left( 1-r_{gi}\left[\Psi\left\lbrace \alpha\left(y_{gi}-\delta\right) \right\rbrace\right]^{-1}\right), \quad i\in [n]\\
    &\E\left[ \bm{h}\left\lbrace \left(y_{gi},r_{gi},\vecbm{A}_i\right), \left(\alpha_g,\delta_g\right)\right\rbrace \right]=\bm{0}_{s+1}, \quad i\in [n],\nonumber
\end{align}
where the second line follows from \eqref{equation:MissingMech}. The resulting generalized method of moments estimator for $\alpha_g$ and $\delta_g$ using \eqref{equation:MomentFunction} also requires the distribution of $y_{gi}$ to depend on $\vecbm{A}_i$ \citep{GMM_MNAR}. That is, $\vecbm{A}_i$ must be an instrumental variable for $r_{gi}$.\par
\indent Unfortunately, as is the case with nearly all biological data, $\bm{Y}$ is typically only weakly dependent on the observed covariates $\bm{X}$, meaning viable instruments $\vecbm{A}_i$ are almost never observed in metabolomic data. Instead, we leverage the fact that the majority of the variation in high throughput metabolomic data, like nearly all high throughput biological data, can be be explained by a relatively small number of potentially latent factors \citep{SVA,RUVMetab2,bcv}. For example, applying principal components analysis to the metabolites with complete data from our motivating data example revealed that only 10 components were necessary to explain nearly 50\% of the variation in those fully observed data. This fact forms the basis of our method to estimate each metabolite's missingness mechanism, which we briefly describe in Algorithm \ref{algorithm:MissMechOverview}.

\begin{algorithm}
\label{algorithm:MissMechOverview}
Fix $\epsilon_{\miss} \in [0,1)$, $K_{\miss} \geq 2$ and let $\Observed = \lbrace g \in [p] : n^{-1}\sum\limits_{i=1}^n \left(1-r_{gi}\right) \leq \epsilon_{\miss} \rbrace$ and $\Missing = \lbrace g \in [p] : n^{-1}\sum\limits_{i=1}^n \left(1-r_{gi}\right) > \epsilon_{\miss} \rbrace$ be the metabolites with (nearly) complete and missing data, respectively.
\begin{enumerate}[label=(\arabic*)]
    \item Use $\bm{Y}_{\Observed} = \left(y_{gi}\right)_{g\in\Observed,i\in[n]}$ to generate $K_{\miss}$ $n$-dimensional factors that explain most of the variation $\bm{Y}_{\Observed}$.\label{item:MissMechOverview:Cmiss}
    \item For each $g \in \Missing$, select two out of the $K_{\miss}$ factors estimated in Step \ref{item:MissMechOverview:Cmiss} to act as instruments for the missingness indicators $r_{g1},\ldots,r_{gn}$.\label{item:MissMechOverview:InstSelection}
    \item For each $g \in \Missing$, use IV-GMM with \eqref{equation:MomentFunction} and instruments obtained from Step \ref{item:MissMechOverview:InstSelection} to compute estimates for $\alpha_g$ and $\delta_g$, $\hat{\alpha}_g^{\GMM}$ and $\hat{\delta}_g^{\GMM}$.\label{item:MissMechOverview:GMM}
    \item Identify metabolites $g \in \Missing$ whose missing data patterns may not follow Model \eqref{equation:MissingMech} using $\hat{\alpha}_g^{\GMM},\hat{\delta}_g^{\GMM}$ and the Sargan-Hansen $J$ statistic.\label{item:MissMechOverview:Jtest}
    \item Obtain estimates for $\alpha_g,\delta_g$ and the weights $w_{gi}=r_{gi}/\Psi\left\lbrace \alpha_g\left(y_{gi}-\delta_g\right)\right\rbrace$ for $g\in\Missing$ and $i\in[n]$ using \textbf{H}ierarchical \textbf{B}ayesian \textbf{G}eneralized \textbf{M}ethod of \textbf{M}oments (HB-GMM).\label{item:MissMechOverview:HBGMM}
\end{enumerate}
\end{algorithm}

\indent We set $\epsilon_{\miss}=0.05$ in practice because simulations show that trace amounts of missing data have negligible effects on the bias in our downstream estimators for $\bm{\beta}_g$. We explain how we choose $K_{\miss}$ in Section \ref{subsection:SelectingInstruments}. Algorithm \ref{algorithm:MissMechOverview} tends to perform well because the estimated factors from Step \ref{item:MissMechOverview:Cmiss} will be approximately the columns of $\left(\bm{X}\, \bm{C}\right)$ from Model \eqref{equation:ModelForData} that explain much of the variance in $\bm{Y}$. And since they are not estimated using metabolites with missing data, they will be approximately independent of $r_{g1},\ldots,r_{gn}$ conditional on $y_{g1},\ldots,y_{gn}$, and therefore auspicious instruments for $r_{g1},\ldots,r_{gn}$  for $g \in \Missing$. We detail and provide concise, intuitive explanations of Steps \ref{item:MissMechOverview:Cmiss} - \ref{item:MissMechOverview:HBGMM} in Sections \ref{subsection:EstimatingInstruments} - \ref{subsection:HBGMM} below. We also justify Algorithm \ref{algorithm:MissMechOverview} in Sections \ref{section:Supp:SelectionProof} - \ref{section:Supp:EstGMMProof} of the Supplement, where we study the asymptotic properties of the estimators from in each step when $\epsilon_{\miss}=0$ and $n,p \to \infty$.

\subsection{Recovering latent factors and estimating coefficients of interest}
\label{subsection:RoadMap:CandBeta}
Our strategy is to use the estimates from Algorithm \ref{algorithm:MissMechOverview} to first obtain $\hat{\bm{C}}$, an estimate for $\bm{C}$, and then plug-in $\hat{\bm{C}}$ for $\bm{C}$ when estimating $\bm{\beta}_1,\ldots,\bm{\beta}_p$. An important feature of our method is once we run Algorithm \ref{algorithm:MissMechOverview}, computing $\hat{\bm{C}}$ and our software's default estimates for $\bm{\beta}_1,\ldots,\bm{\beta}_p$ is fast, which makes analyzing the relationship between $\bm{Y}$ and tens, or even hundreds of different $\bm{X}$'s computationally tractable.

\section{Estimating the missingness mechanisms using Algorithm \ref{algorithm:MissMechOverview}}
\label{section:HBGMM}
\subsection{Estimating the instruments in Step \ref{item:MissMechOverview:Cmiss}}
\label{subsection:EstimatingInstruments}
We define the factors from Step \ref{item:MissMechOverview:Cmiss} of Algorithm \ref{algorithm:MissMechOverview} to be $\hat{\bm{C}}_{\miss} \in \mathbb{R}^{n \times K_{\miss}}$, where $\hat{\bm{C}}_{\miss}$ is the maximum likelihood estimator for $\tilde{\bm{C}} \in \mathbb{R}^{n \times K_{\miss}}$ in the model
\begin{align}
    \label{equation:YObservedLike}
    \bm{Y}_{\Observed} \sim MN_{p_s \times n}\left( \tilde{\bm{\mu}}\bm{1}_n^{\T} + \tilde{\bm{L}}\tilde{\bm{C}}, \tilde{\sigma}^2 I_{p_s}, I_n \right),
\end{align}
where $p_s=\abs{\Observed}$, $\tilde{\bm{C}}^{\T}\bm{1}_n = \bm{0}_{K_{\miss}}$, $n^{-1}\tilde{\bm{C}}^{\T} \tilde{\bm{C}} = I_{K_{\miss}}$, $\tilde{\bm{L}}^{\T}\tilde{\bm{L}}$ is diagonal with non-increasing elements and any missing data are MCAR. If $\epsilon_{\miss} = 0$, $\hat{\bm{C}}_{\miss}$ is a scalar multiple of the first $K_{\miss}$ right singular vectors of $\bm{Y}_{\Observed}P_{1_n}^{\perp}$. When $\epsilon_{\miss} > 0$, the columns of $\hat{\bm{C}}_{\miss}$ are still ordered by decreasing average effect on the log-intensities of metabolites with nearly complete data. Further, by McDiarmid’s Inequality, $\Prob\left(g \in \Observed\right) \leq e^{-2\eta^2 n}$ if $n^{-1}\sum\limits_{i=1}^n\E\left( 1-r_{gi} \right) \geq \epsilon_{\miss}+\eta$ for $\eta > 0$, meaning it suffices to assume $\hat{\bm{C}}_{\miss} \indep r_{gi} \mid y_{gi}$ if $g \in \Missing$ for sufficiently large $n$. That is, for $\bm{h}$ defined in \eqref{equation:MomentFunction} and $\hat{\bm{c}}_i$ the $i$th row of $\hat{\bm{C}}_{\miss}$, we assume $\E\left[ \bm{h}\left\lbrace \left(y_{gi},r_{gi},\hat{\bm{c}}_i\right), \left(\alpha_g,\delta_g\right) \right\rbrace \right] = \bm{0}_{\left(K_{\miss}+1\right)}$ for $g \in \Missing$ and $i \in [n]$.\par 
\indent The columns of $\hat{\bm{C}}_{\miss}$ are the factors that explain the most variation in $\bm{Y}_{\Observed}$. While we expect most of them to derive from $\bm{C}$, some may be related to $\bm{X}$ if $\left\lbrace \bm{\beta}_g \right\rbrace_{g\in\Observed}$ are large enough. Note that $\hat{\bm{C}}_{\miss}$ is invariant of the choice of $\bm{X}$.

\subsection{Instrument selection in Step \ref{item:MissMechOverview:InstSelection}}
\label{subsection:SelectingInstruments}
It is critical that $y_{gi}$ be dependent on the instruments chosen in Step \ref{item:MissMechOverview:InstSelection} of Algorithm \ref{algorithm:MissMechOverview}. Otherwise, the moment condition in \eqref{equation:MomentFunction} will not identify the parameters $\alpha_g$ and $\delta_g$. We therefore use Algorithm \ref{algorithm:SelectU} to only select the instruments $\hat{\bm{U}}_g \in \mathbb{R}^{n \times 2}$ that influence metabolite $g$'s intensity.

\begin{algorithm}
\label{algorithm:SelectU}
Let $\hat{\bm{C}}_{\miss} = \left( \hat{\bm{C}}_1 \cdots \hat{\bm{C}}_{K_{\miss}} \right)$ and $\bm{y}_g = \left(y_{gi}\right)_{i\in[n]}$.
\begin{enumerate}[label=(\arabic*)]
    \item For each $g \in \Missing$ and $k \in \left[K_{\miss}\right]$, use ordinary least squares (OLS) to regress $\bm{y}_g$ onto $\left(\bm{1}_n\, \hat{\bm{C}}_k\right)\in\mathbb{R}^{n \times 2}$, where missing data are treated as MCAR. Let $p_{g,k}$ be the OLS \textit{P} value for the null hypothesis that $\hat{\bm{C}}_k$ is independent of $\bm{y}_g$.\label{item:InstrumentSelection:Pvalue}
    \item For each $k \in \left[K_{\miss}\right]$, use $\left\lbrace p_{g,k} \right\rbrace_{g \in \Missing}$ to determine the corresponding q-values $\left\lbrace q_{g,k} \right\rbrace_{g \in \Missing}$.\label{item:InstrumentSelection:Qvalue}
    \item For each $g \in \Missing$, let $q_{g,g_1} \leq \cdots \leq q_{g,g_{K_{\miss}}}$ be the $K_{\miss}$ ordered q-values. Define $\hat{\bm{U}}_g = \left(\hat{\bm{u}}_{g1} \cdots \hat{\bm{u}}_{gn}\right)^{\T} =  \left(\hat{\bm{C}}_{g_1} \, \hat{\bm{C}}_{g_2} \right)$.\label{item:InstrumentSelection:Indices}
\end{enumerate}
\end{algorithm}

We justify the regression in Step \ref{item:InstrumentSelection:Pvalue} using Theorem \ref{theorem:MetabMiss:OLS_MNAR} in Section \ref{section:Supp:SelectionProof}, which states that under technical assumptions on the distributions of $\bm{Y}_{\Observed}$ and $\bm{y}_g$ for $g \in \Missing$, $p_{g,k}$ is asymptotically uniform under the null hypothesis that $\hat{\bm{C}}_k$ is independent of $\bm{y}_g$. We also use Algorithm \ref{algorithm:SelectU} to choose $K_{\miss}$. If $f\left(k\right)$ is the fraction of metabolites $g \in \Missing$ such that $q_{g,g_2} \leq 0.05$ assuming $K_{\miss} = k$, we set $K_{\miss} = \min \left\lbrace k \in \left\lbrace 2,\ldots,K_{PA} \right\rbrace : f\left(k\right) \geq 0.9 \right\rbrace$, where $K_{PA}$ is parallel analysis' \citep{BujaFA} estimate for $K$ under Model \eqref{equation:YObservedLike} with $\epsilon_{\miss} = 0$. The estimate $K_{\miss}$ is typically much smaller than $K_{PA}$ in practice. For example, $K_{\miss}=10$ and $K_{PA} = 20$ in our motivating data example. We show that our results are robust to the choice of $K_{\miss}$ in Section \ref{section:simulations}.\par
\indent Evidently, this selection step implies $\hat{\bm{u}}_{gi}$ is not strictly independent of $r_{gi}$ conditional on $y_{gi}$. However, we show in Section \ref{section:Supp:SelectionProof} of the Supplement that this dependence is asymptotically negligible under weak assumptions. We therefore assume that the indices $g_1,g_2$ are known and $\hat{\bm{u}}_{gi} \independent r_{gi} \mid y_{gi}$ for the remainder of Section \ref{section:HBGMM}.


\subsection{IV-GMM in Step \ref{item:MissMechOverview:GMM}}
\label{subsection:IVGMM}
Fix a $g \in \Missing$ and define
\begin{align}
\label{equation:SampleMoments}
    \bm{h}_{gi}\left(\tilde{\alpha},\tilde{\delta}\right) = \bm{h}\left\lbrace \left(y_{gi},r_{gi},\hat{\bm{u}}_{gi}\right),\left(\tilde{\alpha},\tilde{\delta}\right) \right\rbrace \in \mathbb{R}^3, \quad \bar{\bm{h}}_g\left(\tilde{\alpha},\tilde{\delta}\right) = n^{-1}\sum\limits_{i=1}^n \bm{h}_{gi}\left(\tilde{\alpha},\tilde{\delta}\right).
\end{align}
We let $\hat{\alpha}_g^{\GMM}$ and $\hat{\delta}_g^{\GMM}$ be the two-step generalized method of moments estimators, defined as
\begin{align}
\label{equation:TwoStepGMM}
    \left\lbrace \hat{\alpha}_g^{\GMM}, \hat{\delta}_g^{\GMM} \right\rbrace = \argmin_{\tilde{\alpha}>0, \tilde{\delta} \in \mathbb{R}}\left\lbrace \bar{\bm{h}}_g\left(\tilde{\alpha},\tilde{\delta}\right)^{\T}\bm{W}_g \bar{\bm{h}}_g\left(\tilde{\alpha},\tilde{\delta}\right)\right\rbrace,
\end{align}
where for $\left\lbrace \hat{\alpha}_g^{(1)}, \hat{\delta}_g^{(1)} \right\rbrace = \mathop{\argmin}\limits_{\tilde{\alpha}>0, \tilde{\delta} \in \mathbb{R}}\left\lbrace \bar{\bm{h}}_g\left(\tilde{\alpha},\tilde{\delta}\right)^{\T} \bar{\bm{h}}_g\left(\tilde{\alpha},\tilde{\delta}\right)\right\rbrace$, the weight matrix $\bm{W}_g$ is
\begin{align}
\label{equation:TwoStepW}
    \bm{W}_g = \bm{W}_g\left\lbrace \hat{\alpha}_g^{(1)}, \hat{\delta}_g^{(1)} \right\rbrace = \left[ n^{-1}\sum\limits_{i=1}^n  \bm{h}_{gi}\left\lbrace\hat{\alpha}^{(1)},\hat{\delta}^{(1)}\right\rbrace \bm{h}_{gi}\left\lbrace\hat{\alpha}^{(1)},\hat{\delta}^{(1)}\right\rbrace^{\T} \right]^{-1}.
\end{align}
The properties of this estimator when $\hat{\bm{U}}_g$ is observed and not estimated and the triplets $\left\lbrace \left(r_{gi},y_{gi},\hat{\bm{u}}_{gi} \right) \right\rbrace_{i \in [n]}$ are independent are well understood \citep{Hansen_2step,GMM_MNAR}. We extend these results in Theorem \ref{theorem:MetabMiss:GMMAsy} in the Supplement to account for the uncertainty in $\hat{\bm{U}}_g$ and prove that under similar regularity conditions as those considered in \cite{GMM_MNAR}, $\abs{\hat{\alpha}_g^{\GMM}-\alpha_g},\abs{\hat{\delta}_g^{\GMM}-\delta_g}= O_P\left(n^{-1/2}\right)$ and for $\bm{\Gamma}_g\left(\tilde{\alpha},\tilde{\delta}\right) = \nabla_{\tilde{\alpha},\tilde{\delta}}\bar{\bm{h}}_g\left(\tilde{\alpha},\tilde{\delta}\right) \in \mathbb{R}^{3 \times 2}$,
\begin{subequations}
\label{equation:AsymptoticDistn}
\begin{align}
    &n^{1/2}\hat{\bm{V}}_g^{-1/2}\left[ \left\lbrace \hat{\alpha}_g^{\GMM},\hat{\delta}_g^{\GMM} \right\rbrace - \left(\alpha_g,\delta_g\right) \right] \tdist N_2\left( \bm{0}, I_2\right)\\
    & \hat{\bm{V}}_g = \left[ \bm{\Gamma}_g\left\lbrace \hat{\alpha}_g^{\GMM},\hat{\delta}_g^{\GMM} \right\rbrace^{\T}\bm{W}_g \bm{\Gamma}_g\left\lbrace \hat{\alpha}_g^{\GMM},\hat{\delta}_g^{\GMM} \right\rbrace \right]^{-1}
\end{align}
\end{subequations}
as $n,p \to \infty$. This result is analogous to Theorem 2 in \cite{GMM_MNAR}, and we use this asymptotic distribution in Section \ref{subsection:HBGMM} to refine our estimates for $\alpha_g$ and $\delta_g$.

\subsection{The Sargan-Hansen $J$ statistic in Step \ref{item:MissMechOverview:Jtest}}
\label{subsection:Jtest}
The accuracy of downstream estimates for $\bm{\beta}_g$ is contingent on the missing data model being approximately correct. Therefore, we leverage the fact that we use three moment conditions to estimate two parameters and use the Sargan-Hansen $J$ statistic, which is routinely used to test moment restrictions in generalized method of moment estimators \citep{Hansen_2step,Baum_ModelSpecification,IdentificationCriterion_book}, to flag metabolites whose missingness mechanisms may not follow Model \eqref{equation:MissingMech}.\par
\indent A consequence of \eqref{equation:AsymptoticDistn} is that under the null hypothesis $H_{0,g}$ that Model \eqref{equation:MissingMech} is correct for metabolite $g \in \Missing$ and the assumptions necessary to prove \eqref{equation:AsymptoticDistn} hold, the statistic $J_g= n\bar{\bm{h}}_g\left\lbrace \hat{\alpha}_g^{\GMM},\hat{\delta}_g^{\GMM} \right\rbrace^{\T}\bm{W}_g \bar{\bm{h}}_g\left\lbrace \hat{\alpha}_g^{\GMM},\hat{\delta}_g^{\GMM} \right\rbrace$ is asymptotically $\chi^2_1$ as $n,p \to \infty$, which is analogous to Lemma 4.2 in \cite{Hansen_2step}. One could then use $J_g$ to test $H_{0,g}$. However, it has been repeatedly observed that using said asymptotic distribution to do inference with $J_g$ is anti-conservative in data with moderate, and even large sample sizes \citep{AntiCons_Jtest2,AntiCons_Jtest1,EmpLikeBoot}. To circumvent this, we follow \cite{EmpLikeBoot} and develop an empirical likelihood-derived bootstrap null distribution for $J_g$, and subsequently estimate $lfdr_g = \Prob\left( H_{0,g} \mid J_g \right)$ using \cite{qvalueSoftware}. We then flag any metabolites with an $lfdr_g$ smaller than a user-specified value, which defaults to 0.8 in our software. Section \ref{section:Supp:Jstat} in the Supplement describes the details of the bootstrap procedure.


\subsection{HB-GMM in Step \ref{item:MissMechOverview:HBGMM}}
\label{subsection:HBGMM}
So far we have estimated each metabolite-specific missingness mechanism independently for each metabolite $g \in \Missing$. While the mechanisms are almost certainly not identical, one might expect them to be relatively similar, and that one should be able to design a better estimator by pooling information across metabolites. Further, constructing an informative prior on the missingness mechanisms allows one to better explore the objective function in \eqref{equation:TwoStepGMM}, which could be multimodal \citep{MNAR_Surface}. We therefore developed Hierarchical Bayesian Generalized Method of Moments (HB-GMM), a Bayesian method to estimate $\alpha_g$, $\delta_g$ and the weights $w_{gi} = r_{gi}/\Psi\left\lbrace \alpha_g\left(y_{gi}-\delta_g\right) \right\rbrace$ for each $g \in \Missing$ and $i \in [n]$. The weights play an important role in estimating $\bm{C}$ in Section \ref{section:CMNAR}.\par
\indent Our method extends Bayesian generalized method of moments \cite{BayesGMM_Orig,BayesGMM,BGMM_Theory} by both incorporating estimated instruments and estimating an informative prior from the data. Define $\Data= \left\lbrace \left(y_{gi},r_{gi},\hat{\bm{u}}_{gi}\right) \right\rbrace_{g\in\Missing,i\in[n]}$. By Bayes' rule and assuming $\left\lbrace \left(\alpha_g,\delta_g\right) \right\rbrace_{g \in \Missing}$ are independent and drawn from some prior distribution,
\begin{align}
\label{equation:BayesRule}
\begin{aligned}
    \Bprob\left[ \left\lbrace \left(\alpha_g,\delta_g\right) \right\rbrace_{g \in \Missing} \mid \Data \right] \propto \Bprob \left[\Data \mid \left\lbrace \left(\alpha_g,\delta_g\right) \right\rbrace_{g \in \Missing}\right] \prod\limits_{g \in \Missing}\Bprob\left(\alpha_g,\delta_g\right).
\end{aligned}
\end{align}
However, the likelihood $\Bprob\left[\Data \mid \left\lbrace \left(\alpha_g,\delta_g\right) \right\rbrace_{g \in \Missing}\right]$ is unknown because the distribution of $y_{gi}$ is unknown. Nevertheless, we do know that under Model \eqref{equation:MissingMech} and assuming $\hat{\bm{u}}_{gi} \independent r_{gi} \mid y_{gi}$ for all $g \in \Missing$ and $i \in [n]$, $\bar{\bm{h}}_g\left( \alpha_g,\delta_g\right)$ and $\bar{\bm{h}}_s\left( \alpha_s,\delta_s\right)$ are uncorrelated for $g \neq s \in \Missing$. Further, since $\bar{\bm{h}}_g\left(\alpha_g,\delta_g\right)$ is an average of $n$ approximately independent random variables, $\bar{\bm{h}}_g\left(\alpha_g,\delta_g\right)$ is asymptotically normal with mean zero and variance given by \eqref{equation:hbarAsy} under the same assumptions used to prove \eqref{equation:AsymptoticDistn}.
\begin{align}
\label{equation:hbarAsy}
\begin{aligned}
    &n^{1/2}\left\lbrace \hat{\bm{\Sigma}}_g\left( \alpha_g,\delta_g\right) \right\rbrace^{-1/2}\bar{\bm{h}}_g\left(\alpha_g,\delta_g\right) \tdist N_3\left(\bm{0}_3,I_3 \right) \text{ as $n,p \to \infty$}, \quad g \in \Missing\\
    &\hat{\bm{\Sigma}}_g\left( \alpha_g,\delta_g\right) = n^{-1}\sum\limits_{i=1}^n \left\lbrace \bm{h}_{gi}\left( \alpha_g,\delta_g\right) - \bar{\bm{h}}_g\left(\alpha_g,\delta_g\right) \right\rbrace \left\lbrace \bm{h}_{gi}\left( \alpha_g,\delta_g\right) - \bar{\bm{h}}_g\left(\alpha_g,\delta_g\right) \right\rbrace^{\T}
\end{aligned}
\end{align}
These facts help justify replacing the likelihood in \eqref{equation:BayesRule} with the pseudo-likelihood
\begin{align*}
    q\left[ \Data \mid \left\lbrace \left(\alpha_g,\delta_g\right) \right\rbrace_{g \in \Missing} \right] = \prod\limits_{g \in \Missing} \normal\left\lbrace \bar{\bm{h}}_g\left(\alpha_g,\delta_g\right)\mid \bm{0}_3, n^{-1}\hat{\bm{\Sigma}}_g\left( \alpha_g,\delta_g\right) \right\rbrace,
\end{align*}
where $\normal\left( \cdot \mid \bm{a}, \bm{b}\right)$ is the likelihood of a normal distribution with mean $\bm{a}$ and variance $\bm{b}$. The form that the pseudo-likelihood takes is computationally convenient because it implies we can sample from the pseudo-posterior
\begin{align*}
\begin{aligned}
   q\left[ \left\lbrace \left(\alpha_g,\delta_g\right) \right\rbrace_{g \in \Missing} \mid \Data \right] \propto \prod\limits_{g \in \Missing} \left[ \normal\left\lbrace \bar{\bm{h}}_g\left(\alpha_g,\delta_g\right)\mid \bm{0}_3, n^{-1}\hat{\bm{\Sigma}}_g\left( \alpha_g,\delta_g\right) \right\rbrace \Bprob\left(\alpha_g,\delta_g\right) \right]
\end{aligned}
\end{align*}
with Markov chain Monte Carlo using $\abs{\Missing}$ parallel chains, which we use to obtain
\begin{subequations}
\label{equation:HBGMM_estimates}
\begin{align}
    \label{equation:HBGMM_estimates:alphadelta}
    \hat{\alpha}_g &=\E\left( \alpha_{g} \mid \Data \right),\quad \hat{\delta}_g=\E\left( \delta_{g} \mid \Data \right), \quad g \in \Missing\\
    \label{equation:HBGMM_estimates:weights}
    \hat{w}_{gi} &=\E\left( w_{gi} \mid \Data \right) = r_{gi}\E\left[1/\Psi\left\lbrace \alpha_g\left(y_{gi}-\delta_g\right) \right\rbrace \mid \Data\right], \quad g \in \Missing; i \in [n]\\
    \label{equation:HBGMM_estimates:Varweights}
    \hat{v}_{gi} &=\E\left( w_{gi}^2 \mid \Data \right) = r_{gi}\E\left[1/\Psi\left\lbrace \alpha_g\left(y_{gi}-\delta_g\right) \right\rbrace^2 \mid \Data\right], \quad g \in \Missing; i \in [n].
\end{align}
\end{subequations}
This technique of replacing the likelihood with the pseudo-likelihood in \eqref{equation:BayesRule} is standard in Bayesian GMM when $\abs{\Missing} = 1$ and $\hat{\bm{U}}_{g}$ is observed \citep{BayesGMM_Orig,BayesGMM,BGMM_Theory}.\par
\indent It remains to specify the prior for $\left(\alpha_g,\delta_g\right)$. We assume that $\left( \log\left(\alpha_g\right), \delta_g \right)^{\T} \mid \left(  \bm{\mu}, \bm{U}\right) \sim N_2\left( \bm{\mu}, \bm{U}\right)$ for all $g \in \Missing$, where we log-transform $\alpha_g$ to make in amenable to a normal prior. We first estimate $\bm{\mu}$ as $\hat{\bm{\mu}} = \abs{\Missing}^{-1}\sum\limits_{g\in\Missing}\left( \log\left\lbrace \hat{\alpha}_g^{\GMM}\right\rbrace, \hat{\delta}_g^{\GMM} \right)^{\T}$. Assuming \eqref{equation:AsymptoticDistn} is approximately correct, we then use empirical Bayes and define our estimate for $\bm{U}$, $\hat{\bm{U}}$, as the maximizer of the following objective over $\bm{U} \succ \bm{0}$:
\begin{align*}
    \prod\limits_{g \in \Missing} \int &\normal\left[ \left( \log\left\lbrace \hat{\alpha}_g^{\GMM} \right\rbrace, \hat{\delta}_g^{\GMM} \right)^{\T}  \mid \left( \eta_g, \delta_g \right)^{\T}, \hat{\bm{R}}_g \right]\normal\left\lbrace\left( \eta_g, \delta_g \right)^{\T} \mid \hat{\bm{\mu}}, \bm{U}\right\rbrace d\eta_g d\delta_g,
\end{align*}
where $\hat{\bm{R}}_g = \diag\left\lbrace 1/\hat{\alpha}_g^{\GMM},1 \right\rbrace \hat{\bm{V}}_g \diag\left\lbrace 1/\hat{\alpha}_g^{\GMM},1 \right\rbrace$ for $\hat{\bm{V}}_g$ defined in \eqref{equation:AsymptoticDistn}. We estimate $\bm{U}$ using the product of marginal likelihoods because under the assumptions used to prove \eqref{equation:AsymptoticDistn}, the estimates $\left( \hat{\alpha}_g^{\GMM}, \hat{\delta}_g^{\GMM} \right)$ and $\left( \hat{\alpha}_s^{\GMM}, \hat{\delta}_s^{\GMM} \right)$ are asymptotically independent for $g \neq s \in \Missing$. See Section \ref{section:Supp:EstGMMProof} in the Supplement for more details.

\section{Estimating coefficients when $C$ is known}
\label{section:IPW}
Here we describe our method for estimating $\bm{\beta}_g$ and $\bm{\ell}_g$ in Model \eqref{equation:ModelForData} when $\bm{C}$ is known, which is based on inverse probability weighting \citep{PairwisePseudoLike}. This methodology is used in Section \ref{section:CMNAR} to recover $\bm{C}$, and is also our default method to perform inference on the coefficients of interest because estimates are consistent, it obviates specifying a probability model for the missing data and computation is fast enough to perform a metabolite phenome wide association study. For notational simplicity, we rewrite Model \eqref{equation:ModelForData} as
\begin{align*}
    y_{gi} = \bm{z}_i^{\T}\bm{\eta}_g + e_{gi}, \quad e_{gi} \asim \left(0,\sigma_g^2\right), \quad g\in[p]; i\in [n]
\end{align*}
for the remainder of Section \ref{section:IPW}. Since estimation is trivial when there is little missing data, our goal is to estimate $\bm{\eta}_g$ for all $g \in \Missing$ when $\bm{Z}=\left(\bm{z}_1\cdots \bm{z}_n\right)^{\T}$ is observed. 

\subsection{Point estimates}
\label{subsection:IPW:point}
Fix a $g \in \Missing$ and for all $i \in [n]$, define the score function $\bm{s}_{gi}\left( \bm{\eta}\right) = \bm{z}_i \left(y_{gi} - \bm{z}_i^{\T} \bm{\eta} \right)$, $\gamma_{gi} = \Prob\left(r_{gi} = 1 \mid \bm{Z}\right)$ and the inverse probability weighted estimating equation $\bm{f}_{g}\left(\bm{\eta}\right) = \sum\limits_{i=1}^n \hat{\gamma}_{gi}\hat{w}_{gi} \bm{s}_{gi}\left( \bm{\eta}\right)$, where $\hat{w}_{gi}$ is defined in \eqref{equation:HBGMM_estimates:weights} and $\hat{\gamma}_{gi}$ is an estimate of $\gamma_{gi}$. If $\hat{w}_{gi} = w_{gi}$ and $\hat{\gamma}_{gi}=\gamma_{gi}$ for all $i \in [n]$, then
\begin{align*}
    \E\left\lbrace \bm{f}_{g}\left(\bm{\eta}_{g}\right) \mid \bm{Z}\right\rbrace &= \sum\limits_{i=1}^n \gamma_{gi}\E\left\lbrace \E\left( w_{gi} \mid y_{gi} \right)\bm{s}_{gi}\left( \bm{\eta}_{g}\right) \mid \bm{Z} \right\rbrace= \sum\limits_{i=1}^n \gamma_{gi}\E\left\lbrace \bm{s}_{gi}\left( \bm{\eta}_{g}\right) \mid \bm{Z} \right\rbrace = \bm{0}.
\end{align*}
The above equality can be shown to hold in the more general case when $\gamma_{gi} \indep \bm{y}_g \mid \bm{Z}$ for all $i\in[n]$, meaning the root of $\bm{f}_{g}$ will be an accurate estimate of $\bm{\eta}_g$ if $\hat{w}_{gi}$ is consistent for $w_{gi}$ and $\hat{\gamma}_{gi}$ is only weakly dependent on $\bm{y}_g$. We include $\hat{\gamma}_{gi}$ to stabilize potentially large weights $\hat{w}_{gi}$ and thereby reduce the variance of our estimates, since $\hat{\gamma}_{gi}$ will tend to be small if $\hat{w}_{gi}$ is large. This method of stabilized inverse probability weighting has been successfully applied to data that are missing at random \citep{StabilizedWeights_MAR}, and we estimate $\gamma_{gi}$ using a logistic regression with the estimated instruments $\hat{\bm{U}}_g$. We then define our estimate for $\bm{\eta}_{g}$ as the root of $\bm{f}_g$:
\begin{align}
\label{equation:IPW:beta.hat}
    \hat{\bm{\eta}}_g = \left(\bm{Z}^{\T} \hat{\bm{W}}_g\bm{Z} \right)^{-1} \bm{Z}^{\T} \hat{\bm{W}}_g\bm{y}_g, \quad \hat{\bm{W}}_g = \diag\left( \hat{w}_{g1}\hat{\gamma}_{g1}, \ldots, \hat{w}_{gn}\hat{\gamma}_{gn} \right).
\end{align}
Note $\bm{s}_{gi}$ in $\bm{f}_g$ can be redefined to be any M-estimator, like Huber's or Tukey's robust estimators, provided $\E\left\lbrace \bm{s}_{gi}\left(\bm{\eta}_g\right) \mid \bm{Z} \right\rbrace = \bm{0}$.

\subsection{Quantifying uncertainty}
\label{subsection:IPW:variance}
Fix a $g \in \Missing$. Here we describe our estimator for $\V\left( \hat{\bm{\eta}}_g\right)$, which we use to recover $\bm{C}$ in Section \ref{section:CMNAR} and perform inference on $\bm{\eta}_g$. Our estimator is a novel finite sample-corrected sandwich variance estimator that also accounts for the uncertainty in the estimated weights $\hat{w}_{gi}$.\par
\indent Suppose for simplicity that $\hat{w}_{gi}=w_{gi}$ and $\hat{\gamma}_{gi}=\gamma_{gi}$. Then
\begin{align*}
    n^{1/2}\left(\bm{\eta}_g - \hat{\bm{\eta}}_g\right) = \left(n^{-1}\bm{Z}^{\T}\hat{\bm{W}}_g \bm{Z}\right)^{-1}\left(n^{-1/2}\sum\limits_{i=1}^n \gamma_{gi}w_{gi}e_{gi}\bm{z}_i \right).
\end{align*}
Since $\left\lbrace \left(y_{gi},r_{gi}\right) \right\rbrace_{i \in [n]}$ are mutually independent and $\E\left( \gamma_{gi}w_{gi}e_{gi} \mid \bm{Z}\right) = 0$,
\begin{align*}
    n\V\left(\hat{\bm{\eta}}_g \mid \bm{Z}\right) \approx \left(n^{-1}\bm{Z}^{\T}\hat{\bm{W}}_g \bm{Z}\right)^{-1} \left( n^{-1}\sum\limits_{i=1}^n \gamma_{gi}^2 w_{gi}^2 e_{gi}^2 \bm{z}_i \bm{z}_i^{\T}\right) \left(n^{-1}\bm{Z}^{\T}\hat{\bm{W}}_g \bm{Z}\right)^{-1}.
\end{align*}
Therefore, we need only approximate the middle term to estimate $\V\left(\hat{\bm{\eta}}_g \mid \bm{Z}\right)$. Simply plugging in $\hat{w}_{gi}^2$ for $w_{gi}^2$ will tend to underestimate $\V\left(\hat{\bm{\eta}}_g \mid \bm{Z}\right)$, since the uncertainty in $\hat{w}_{gi}$ increases as $w_{gi}$ increases. Further, plugging in $\hat{e}_{gi} = y_{gi} - \bm{z}_i^{\T}\hat{\bm{\eta}}_g$ for $e_{gi}$ will also underestimate $\V\left(\hat{\bm{\eta}}_g \mid \bm{Z}\right)$, since this ignores the uncertainty in $\hat{\bm{\eta}}_g$. We circumvent the former by replacing $w_{gi}^2$ with $\hat{v}_{gi}$ defined in \eqref{equation:HBGMM_estimates:Varweights}, where $\hat{v}_{gi} \geq \hat{w}_{gi}^2$ such that $\hat{v}_{gi} = \hat{w}_{gi}^2$ if and only if $\V\left(w_{gi} \mid \Data\right) = 0$. That is, $\hat{v}_{gi}$ helps account for the uncertainty in our estimate for $w_{gi}$. We lastly show how we estimate $e_{gi}^2$ in Section \ref{section:MetabMiss:Variance} of the Supplement, which leads to the following estimate for $\V\left( \hat{\bm{\eta}}_g \mid \bm{Z}\right)$:
\begin{align}
\label{equation:IPW:Variance}
\hat{\V}\left( \hat{\bm{\eta}}_g \mid \bm{Z}\right) = \left(\bm{Z}^{\T} \hat{\bm{W}}_g \bm{Z} \right)^{-1} \left\lbrace \sum\limits_{i=1}^n \left( 1-\hat{h}_{gi}\right)^{-2}\hat{\gamma}_{gi}^2\hat{v}_{gi}\hat{e}_{gi}^2 \bm{z}_i\bm{z}_i^{\T} \right\rbrace \left(\bm{Z}^{\T} \hat{\bm{W}}_g \bm{Z} \right)^{-1}.
\end{align}
The term $\left( 1-\hat{h}_{gi}\right)^{-2}$ is a finite sample correction, where $\hat{h}_{gi}$ is the $i$th leverage score of $\hat{\bm{W}}_g^{1/2}\bm{Z}$ for $i \in [n]$. This resembles the $\left( 1-\hat{h}_{gi}\right)^{-1}$ inflation term commonly used to correct the sandwich variance estimator \citep{Sandwich_GEE}. The difference arises because the residuals $e_{g1},\ldots,e_{gn}$ are dependent on the design matrix $\hat{\bm{W}}_g^{1/2}\bm{Z}$ when data are missing not at random (MNAR). As far as we are aware, this is the first such finite sample variance correction for inverse probability weighted estimators with data that are MNAR.\par

\section{Recovering $C$ when data are MNAR}
\label{section:CMNAR}
Here we describe our method for estimating the latent covariates $\bm{C}$, where we now return to using the notation of Model \eqref{equation:ModelForData}. Define $\bm{X}_{\interest}$ and $\bm{X}_{\nuisance}$ such that $\bm{X}=\left( \bm{X}_{\interest}\, \bm{X}_{\nuisance}\right)$, where $\bm{X}_{\interest}$ contains the covariates of interest like disease status and $\bm{X}_{\nuisance}$ contains observed nuisance covariates like the intercept and technical factors. We assume for simplicity of presentation that $\bm{X} = \bm{X}_{\interest}$, and we describe the simple extension when $\bm{X}=\left( \bm{X}_{\interest}\, \bm{X}_{\nuisance}\right)$ in Section \ref{section:Supp:Nuisance} in the Supplement.\par 
\indent Let $\bm{y}_g = \left(y_{gi}\right)_{i\in[n]}$ and $\bm{e}_g=\left(e_{gi}\right)_{i\in[n]}$ for each $g \in [p]$. Then
\begin{subequations}
\label{equation:EstC}
\begin{align}
    \label{equation:EstC:yg}
    &\bm{y}_{g} = \bm{X}\tilde{\bm{\beta}}_{g} + \bm{C}_2\bm{\ell}_{g} + \bm{e}_g, \quad \tilde{\bm{\beta}}_{g} = \bm{\beta}_{g} + \bm{\Omega}\bm{\ell}_{g}, \quad \bm{e}_g \asim \left(\bm{0}_n,\sigma_g^2 I_n\right),\quad g \in [p]\\
    \label{equation:EstC:OmegaC2}
    & \bm{\Omega} = \left(\bm{X}^{\T} \bm{X}\right)^{-1}\bm{X}^{\T}\bm{C}, \quad \bm{C}_2 = P_{X}^{\perp}\bm{C}.
\end{align}
\end{subequations}
\noindent It is easy to show that typical estimates for $\bm{\beta}_g$ using the design matrices $\left(\bm{X}\, \hat{\bm{C}}\right)$ and $\left(\bm{X}\, \bm{C}\right)$, like OLS and that in \eqref{equation:IPW:beta.hat}, will be identical if $\im\left(\hat{\bm{C}}\right)=\im\left(\bm{C}\right)$. Consequently, we need only estimate $\im\left(\bm{C}\right)$, which is quite auspicious because even though $\bm{C}$ is not identifiable in \eqref{equation:ModelForData}, $\im\left(\bm{C}\right)$ is identifiable under assumptions on the sparsity of $\left(\bm{\beta}_1 \cdots \bm{\beta}_p\right)$ \citep{CorrConf}. We therefore assume without loss of generality that $\bm{C}$ and our estimator for $\bm{C}$, $\hat{\bm{C}}$, satisfy $ = n^{-1}\bm{C}^{\T}P_{X}^{\perp}\bm{C}= n^{-1}\hat{\bm{C}}^{\T}P_{X}^{\perp}\hat{\bm{C}} = I_K$. We estimate $\bm{C}_2$ and $\bm{\Omega}$ in Sections \ref{subsection:Cperp} and \ref{subsection:Omega} below, and define
\begin{align}
\label{equation:Chat}
    \hat{\bm{C}} = \bm{X}\hat{\bm{\Omega}} + \hat{\bm{C}}_2.
\end{align}
For $\bm{Z}=\left(\bm{X}\, \hat{\bm{C}}\right)$ and $\bm{R}_g=\diag\left(r_{g1},\ldots,r_{gn}\right)$, our estimates for $\bm{\eta}_g=\left(\bm{\beta}_g^{\T},\bm{\ell}_g^{\T}\right)^{\T}$ are
\begin{subequations}
\label{equation:FullEstimates}
\begin{align}
\label{equation:FullBetaHat}
    \hat{\bm{\eta}}_g = &\begin{cases}
    \left(\bm{Z}^{\T}\bm{R}_g\bm{Z}\right)^{-1}\bm{Z}^{\T}\bm{R}_g\bm{y}_g & \text{if $g \in \Observed$}\\
    \eqref{equation:IPW:beta.hat} & \text{if $g \in \Missing$}
    \end{cases}\\
    \label{equation:FullVar}
    \hat{\V}\left(\hat{\bm{\eta}}_g\right)=&\begin{cases}
    \lbrace\Tr(\bm{R}_g)-d-K\rbrace^{-1}\norm{\bm{R}_g\left(\bm{y}_g-\bm{Z}\hat{\bm{\eta}}_g\right)}_2^2\left(\bm{Z}^{\T}\bm{R}_g\bm{Z}\right)^{-1} & \text{if $g \in \Observed$}\\
    \eqref{equation:IPW:Variance} & \text{if $g \in \Missing$.}
    \end{cases}
\end{align}
\end{subequations}


\subsection{Estimating latent factors that are orthogonal to the design}
\label{subsection:Cperp}
We first describe our estimators for $\tilde{\bm{\beta}}_g$, $\bm{\ell}_g$ and $\bm{C}_2$, which we also use in Section \ref{subsection:Omega} to estimate $\bm{\Omega}$. Let $\Missingsub = \left\lbrace g \in \Missing : lfdr_g \geq 0.8 \right\rbrace$ be the set of metabolites with missing data whose missingness mechanisms appear to follow \eqref{equation:MissingMech}, where $lfdr_g$ was defined in Section \ref{subsection:Jtest}. We estimate $\tilde{\bm{\beta}}_{g}, \bm{\ell}_{g}$ and $\bm{C}_2$ using metabolites with nearly complete data or missing data whose missingness mechanisms appear to follow \eqref{equation:MissingMech} with the following scaled quasi-likelihood obsjective function, under the restriction that $\bm{C}_2^{\T}\bm{X} = \bm{0}$ and $n^{-1}\bm{C}_2^{\T}\bm{C}_2 = I_K$:
\begin{align*}
\begin{aligned}
    \left\lbrace \left\lbrace \hat{\tilde{\bm{\beta}}}_g \right\rbrace_{g \in \Observed \cup \Missingsub}, \left\lbrace \hat{\bm{\ell}}_g \right\rbrace_{g \in \Observed \cup \Missingsub},  \hat{\bm{C}}_2\right\rbrace =& \argmax_{\substack{ \tilde{\bm{\beta}}_g \in \mathbb{R}^d,\, \bm{\ell}_g \in \mathbb{R}^K \\\bm{C}_2 \in \mathbb{R}^{n \times K}  }} \left[ -\sum\limits_{g \in \Observed} \norm{\bm{R}_g \left\lbrace \bm{y}_g - \left( \bm{X}\tilde{\bm{\beta}}_g + \bm{C}_2 \bm{\ell}_g \right) \right\rbrace}_2^2 \right.\\
    &\left. - \sum\limits_{g \in \Missingsub} \norm{\hat{\bm{W}}_g^{1/2}\left\lbrace \bm{y}_g - \left( \bm{X}\tilde{\bm{\beta}}_g + \bm{C}_2 \bm{\ell}_g \right) \right\rbrace}_2^2 \right].   
\end{aligned}
\end{align*}
This optimization treats missing data from metabolites with little to no missing data as MCAR. For fixed $\bm{C}_2$, the updates for $\tilde{\bm{\beta}}_g$ and $\bm{\ell}_g$ are given by \eqref{equation:FullBetaHat} with $\bm{Z}=\left(\bm{X}\, \bm{C}_2\right)$.

\subsection{Estimating latent factors in the image of the design}
\label{subsection:Omega}
We now describe how we estimate $\bm{\Omega}$. The estimates $\left(\hat{\tilde{\bm{\beta}}}_g^{\T},\hat{\bm{\ell}}_g^{\T}\right)^{\T}$ can be expressed as \eqref{equation:FullBetaHat} using the design matrix $\bm{Z}=\left(\bm{X}\, \hat{\bm{C}}_2\right)$. By \eqref{equation:EstC:yg}, this suggests an appropriate model for $\hat{\tilde{\bm{\beta}}}_g$ is $\hat{\tilde{\bm{\beta}}}_{g} \asim \left( \bm{\beta}_{g} + \bm{\Omega}\bm{\ell}_g, \hat{\V}\left(\hat{\tilde{\bm{\beta}}}_{g}\right)\right)$, where $\hat{\V}\left(\hat{\tilde{\bm{\beta}}}_{g}\right) \in \mathbb{R}^{d \times d}$ is the upper left $d \times d$ submatrix of $\hat{\V}\left(\hat{\bm{\eta}}_g\right)$ defined in \eqref{equation:FullVar} for $\bm{Z}=\left(\bm{X}\, \hat{\bm{C}}_2\right)$. If $\left(\bm{\beta}_1 \cdots \bm{\beta}_p\right)$ is sparse, the expression for $\tilde{\bm{\beta}}_g$ in \eqref{equation:EstC:yg} suggests we can regress the estimates for $\tilde{\bm{\beta}}_g$ onto those for $\bm{\ell}_g$ to estimate $\bm{\Omega}$. This is outlined in Algorithm \ref{algorithm:Omega}.
\begin{algorithm}[Estimating $\bm{\Omega}$]
\label{algorithm:Omega}
Let $\epsilon_{\qvalue} \in [0,1]$, $R \geq 0$ be an integer and $\hat{\tau}_{g,j}$ be the $j$th diagonal element of $\hat{\V}\left(\hat{\tilde{\bm{\beta}}}_g\right)\in\mathbb{R}^{d \times d}$ for all $g \in \Observed \cup \Missingsub$ and $j \in [d]$.
\begin{enumerate}[label=(\arabic*)]
\setcounter{enumi}{-1}
    \item For $j\in[d]$, let $\hat{\bm{\Omega}}^{(0)}_{j} = \left( \sum\limits_{g \in \Observed \cup \Missingsub} \hat{\tau}_{g,j}^{-1} \hat{\bm{\ell}}_g \hat{\bm{\ell}}_g^{\T} \right)^{-1} \left(\sum\limits_{g \in \Observed \cup \Missingsub} \hat{\tau}_{g,j}^{-1}\hat{\tilde{\bm{\beta}}}_{g_j} \hat{\bm{\ell}}_g \right)$ and $\hat{\bm{\Omega}}^{(0)} = \left( \hat{\bm{\Omega}}^{(0)}_{1} \cdots \hat{\bm{\Omega}}^{(0)}_{d} \right)^{\T}$. Define $\hat{\bm{C}}^{(0)} = \bm{X}\hat{\bm{\Omega}}^{(0)} + \hat{\bm{C}}_2$. If $R = 0$, return $\hat{\bm{C}} = \hat{\bm{C}}^{(0)}$.\label{item:Omega:Intial}
    \item Let $\hat{\bm{C}}^{(r)}$ be given. Define $\hat{\bm{\beta}}_g^{(r)}$ and $\hat{\V}\left\lbrace \hat{\bm{\beta}}_g^{(r)} \right\rbrace$ to be the first $d$ coordinates and upper left $d \times d$ block of $\hat{\bm{\eta}}_g$ and $\hat{\V}\left(\hat{\bm{\eta}}_g\right)$ defined in \eqref{equation:FullBetaHat} and \eqref{equation:FullVar}, respectively, for $\bm{Z}=\left(\bm{X}\, \hat{\bm{C}}^{(r)}\right)$. For $z^2 \sim \chi^2_1$, let $p_{g,j}=\Prob\left[ z^2\geq \left\lbrace\hat{\bm{\beta}}_g^{(r)}\right\rbrace^2/\hat{\V}\left\lbrace \hat{\bm{\beta}}_g^{(r)} \right\rbrace_{jj} \right]$ for all $g \in \Observed \cup \Missingsub$ and $j \in[d]$.\label{item:Omega:pvalue}
    \item Obtain the q-values $\left\lbrace q_{g,j} \right\rbrace_{g \in \Observed \cup \Missingsub}$ using the \textit{P} values $\left\lbrace p_{g,j} \right\rbrace_{g \in \Observed \cup \Missingsub}$ for each $j \in [d]$. Define $\hat{\bm{\Omega}}^{(r+1)} = \left( \hat{\bm{\Omega}}^{(r+1)}_{1} \cdots \hat{\bm{\Omega}}^{(r+1)}_{d} \right)^{\T}$ to be
    \begin{align*}
         \hat{\bm{\Omega}}^{(r+1)}_{j} =& \left\lbrace \sum\limits_{g \in \Observed \cup \Missingsub} I\left(q_{g,j} > \epsilon_{\qvalue}\right)\hat{\tau}_{g,j}^{-1} \hat{\bm{\ell}}_g \hat{\bm{\ell}}_g^{\T} \right\rbrace^{-1}\left\lbrace\sum\limits_{g \in \Observed \cup \Missingsub} I\left(q_{g,j} > \epsilon_{\qvalue}\right)\hat{\tau}_{g,j}^{-1} \hat{\tilde{\bm{\beta}}}_{g_j} \hat{\bm{\ell}}_{g} \right\rbrace.
    \end{align*}
    Update $r \leftarrow r+1$ and define $\hat{\bm{C}}^{(r)} = \bm{X}\hat{\bm{\Omega}}^{(r)} + \hat{\bm{C}}_2$.\label{item:Omega:Refine}
    \item Repeat Steps \ref{item:Omega:pvalue} and \ref{item:Omega:Refine} for $r=0,1,\ldots,R-1$ and return $\hat{\bm{\Omega}} = \hat{\bm{\Omega}}^{(R)}$.
\end{enumerate}
\end{algorithm}
Our software's default is $\epsilon_{\qvalue} = 0.1$ and $R=3$. While $\hat{\bm{\Omega}}^{(0)}$ is a suitable estimate for $\bm{\Omega}$ when $\left(\bm{\beta}_1\cdots\bm{\beta}_p\right)$ is very sparse, Step \ref{item:Omega:Refine} identifies and removes metabolites with non-zero coefficients of interest $\bm{\beta}_g$ and helps alleviate the impact of outliers in the regression estimate for $\bm{\Omega}$ when $\left(\bm{\beta}_1 \cdots \bm{\beta}_p\right)$ is only approximately sparse.

\section{A simulation study}
\label{section:simulations}
\subsection{Simulation setup}
\label{subsection:SimSetup}
Here we analyze simulated metabolomic data to compare the performance of our method with other existing methods. We simulated the log-intensities of $p=1200$ metabolites in $n=600$ individuals, 300 of which were cases and the remaining 300 were controls. The observed design matrix was $\bm{X} = \left( \bm{X}_{\interest}\, \bm{1}_n \right)$, where $\bm{X}_{\interest} = \left( \bm{1}_{n/2}^{\T}, \bm{0}_{n/2}^{\T}\right)^{\T} \in \mathbb{R}^n$. The parameters $p$ and $n$ were chosen to match those from our real data example in Section \ref{section:DataAnlysis}, and we include additional results when $n=100$ and $n=300$ in Section \ref{subsection:Supp:Smalln} of the supplement. We set $K=10$, and for some constant $a$ and appropriate $\Psi(x)$, simulated data as
\begin{subequations}
\label{equation:MetabMiss:Simulation}
\begin{align}
    \label{equation:Sim:adelta}
    &\log\left( \alpha_{g}\right) \sim N_1\left( \mu_{\alpha},0.4^2\right), \quad \delta_{g} \sim N_1\left( 16, 1.2^2\right), \quad g \in [p]\\
    \label{equation:Sim:C}
    &\bm{C}=\left(\bm{c}_1 \cdots \bm{c}_n\right)^{\T} \sim MN_{n \times K}\left( \left( a\bm{X}_{\interest} \, \bm{0}_n \cdots \bm{0}_n\right), I_n, I_K \right)\\
    \label{equation:Sim:ell}
    &\bm{\ell}_{g_k} \sim \pi_k \delta_0 + \left( 1-\pi_k\right)N_1\left( 0,\tau_k^2\right), \quad g\in[p]; k\in[K]\\
    \label{equation:Sim:meanvar}
    &\mu_{g} \sim N_1\left(18, 5^2 \right), \quad \sigma_{g}^2 \sim \text{Gamma}\left( 0.2^{-2}, 0.2^{-2}\right), \quad g\in[p]\\
    \label{equation:Sim:effect}
    &\beta_{g} \sim 0.8\delta_0 + 0.2N_1\left( 0, 0.4^2\right), \quad g\in[p]\\
    \label{equation:Sim:y}
    &y_{gi} \sim N_1\left( \mu_{g} + \bm{X}_{\interest_i}\beta_{g} + \bm{c}_i^{\T} \bm{\ell}_{g}, \sigma_{g}^2 \right), \quad g\in[p]; i\in[n]\\
    \label{equation:Sim:r}
    &r_{gi} = \text{Bernoulli}\left[ \Psi\left\lbrace \alpha_{g}\left(y_{gi}-\delta_{g} \right) \right\rbrace \right], \quad g\in[p]; i\in[n]
\end{align}
\end{subequations}
where $\delta_0$ is the point mass at 0 and $\mu_{\alpha}$ in \eqref{equation:Sim:adelta} was set so that if $Z$ has cumulative distribution function $\Psi\left\lbrace \exp\left( \mu_{\alpha}\right)x \right\rbrace$, $\V(Z)=1$. The constant $a$ in \eqref{equation:Sim:C} was chosen so that $\bm{C}$ explained 7.5\% of the variance in $\bm{X}_{\interest}$ on average across all simulations, and Table \ref{Table:LMetab} contains the values of $\pi_k$ and $\tau_k^2$. These were chosen so that the non-zero eigenvalues $\lambda_1,\ldots ,\lambda_K$ of $\mathcal{I}=(n-1)^{-1}P_{1_n}^{\perp}\bm{C}\left( p^{-1}\sum\limits_{g=1}^p \sigma_{g}^{-2}\bm{\ell}_{g}\bm{\ell}_{g}^{\T} \right)\bm{C}^{\T} P_{1_n}^{\perp}$ were 0.61, 0.33, 0.19, 0.14, 0.12, 0.08, 0.07, 0.05, 0.05 and 0.05 on average across all simulated datasets, since these were the first 10 eigenvalues of the estimated $\mathcal{I}$ in our data example in Section \ref{section:DataAnlysis}. Similarly, the prior variances for the missingness mechanism parameters in \eqref{equation:Sim:adelta}, as well as the mean and variance for the global mean $\mu_g$ in \eqref{equation:Sim:meanvar}, were set to their estimated equivalents from our data example in Section \ref{section:DataAnlysis}. Since we typically do not know the exact functional form of $\Psi(x)$ in practice, we set $\Psi(x) = \exp(x)/\left\lbrace 1+\exp(x) \right\rbrace$ and analyzed each simulated dataset assuming $\Psi(x)=F_4(x)$. The distribution of missing data is given in Table \ref{Table:SimFrac}, which closely matched that in our real data example.

\LtabMetab

\Fractab

\indent We simulated 60 datasets and in each simulation, removed all metabolites that were missing in more than 50\% of the samples, since we find that these metabolites tend to have large $J$ statistics in real data. We used Algorithm \ref{algorithm:MissMechOverview} with $\epsilon_{\miss} = 0.05$ and $K_{\miss} = 5$ to estimate the metabolite-dependent missingness mechanism parameters $\alpha_{g}$ and $\delta_{g}$, and subsequently estimated $\bm{C}$ as \eqref{equation:Chat} with $\epsilon_{\qvalue} = 0.1$, assuming $K=10$ was known. We lastly estimated $\beta_g$ and said estimator's variance using \eqref{equation:FullEstimates}, and formed 95\% confidence intervals and computed \textit{P} values assuming $\hat{\beta}_{g} \sim N_1\left( \beta_{g}, \hat{\V}\left( \hat{\beta}_{g}\right) \right)$. We refer to this procedure as ``MetabMiss".\par
\indent Similar to our real data example, $K_{\miss}$ was such that $q_{g,g_2}$, defined in Algorithm \ref{algorithm:SelectU}, was less than 0.05 in at least 90\% of all metabolites $g \in \Missing$ in each simulated dataset. Our results were identical when we let $K_{\miss}$ be as small as 3 and as large as 10. The fact that $K$ was assumed to be known when estimating $\bm{C}$ was inconsequential, since parallel analysis \citep{SVAinR} applied to metabolites with only complete data consistently estimated $K=10$. We demonstrate the fidelity of MetabMiss' estimates for the coefficients of interest $\beta_1,\ldots,\beta_p$ in Section \ref{subsection:SimResults}. We also illustrate the accuracy of Algorithm \ref{algorithm:MissMechOverview}'s estimates for $\alpha_g$ and $\delta_g$, as well as the uniformity of the bootstrapped Sargan-Hansen $J$ statistic \textit{P} values, in Section \ref{subsection:Supp:SimMisParams}.

\subsection{Simulation results}
\label{subsection:SimResults}
Given the paucity of methods available to analyze metabolomics data, we could only compare our method to those that account for non-random missing data or $\bm{C}$, but not both. It is not interesting to compare MetabMiss to methods that only account for the former, like those proposed in \cite{MNAR_SILAC,SILAC_MNAR2,MNARProteomics}, because ignoring $\bm{C}$ will dramatically inflate type I error. Instead, we compared MetabMiss to existing methods that have been used to recover $\bm{C}$ in metabolomic data but do not account for the non-random missing data, which included IRW-SVA \citep{SVA2008}, dSVA \citep{BiometrikaConfounding}, RUV-2 \citep{RUV} and RUV-4 \citep{RUV4}. We do not report results from the method proposed in \cite{CATE}, as it performed nearly identically to dSVA in all simulation scenarios. Since none of the of the aforementioned methods can accommodate missing data, we estimated $\bm{C}$ with each using only metabolites with complete observations, and computed confidence intervals and \textit{P} values using OLS with the estimated design matrix $\left( \bm{X}\, \hat{\bm{C}}\right)$, assuming missing data were MCAR. We remark that we could not analyze these simulated data with the methods proposed in \cite{RUVMetab1} or \cite{Metab_RRmix} because both methods rely on a random effects model whose implemented estimators are not amenable to any missing data.\par 
\indent We first evaluated each method's ability to identify metabolites with non-zero effect of interest $\beta_{g}$ while controlling the false discovery rate at a nominal level. The results are given in Figure \ref{Figure:HypothesisTesting}, where the only method to suitably control for false discoveries is MetabMiss. The fact that MetabMiss is slightly underpowered compared to the other methods is to be expected, as anti-conservative inference is typically more powerful. We also evaluated the confidence interval coverage for the effects of interest $\beta_{g}$ for each method in Figure \ref{Figure:CI}. These results illustrate the consequences of performing inference on estimators that do not properly account for the missing data, and also highlight the fidelity of our finite sample-corrected estimator for the variance defined in  \eqref{equation:IPW:Variance}. 

\begin{figure}[t!]
   \centering
   \includegraphics[scale=0.25]{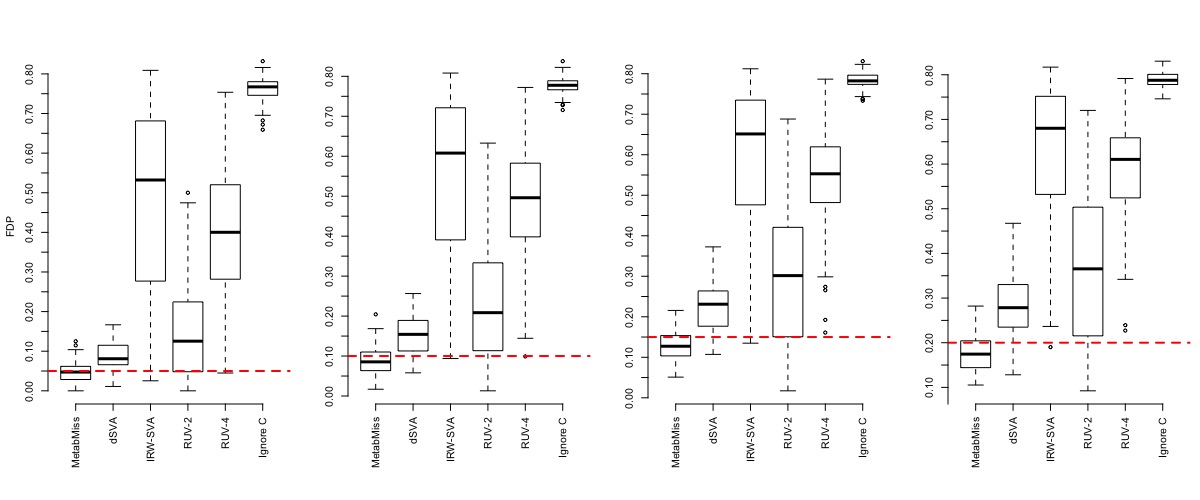}\\
   \includegraphics[scale=0.25]{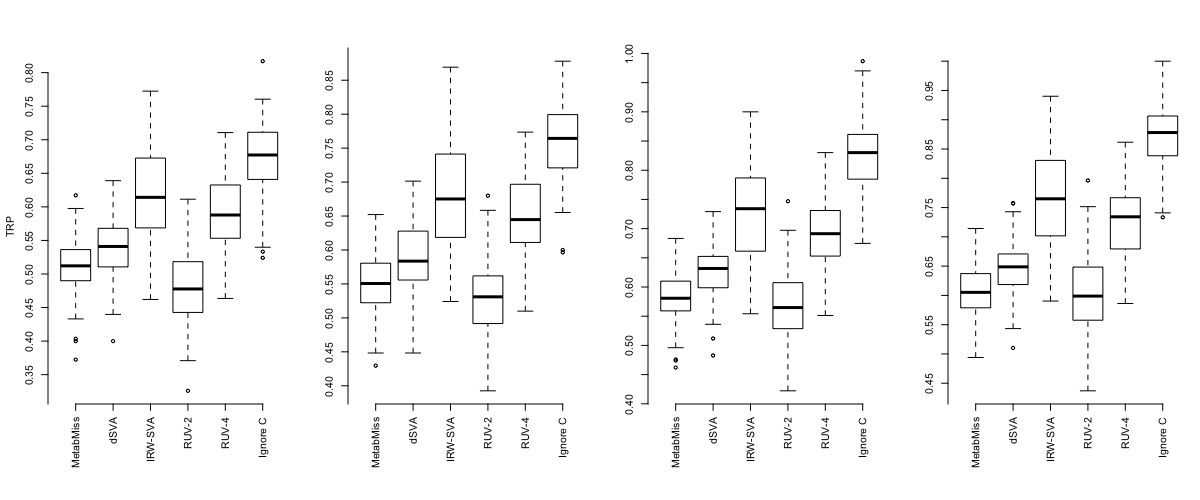}
   \caption[Inference using MetabMiss to simultaneously account for latent covariates and non-random missing data.]{From left to right: the false discovery proportion, FDP (a), and true recovery proportion, TRP (b), for metabolites with q-values $\leq 0.05$, 0.1, 0.15 and 0.2. The TRP is the fraction of metabolites with non-zero $\beta_g$ identified at a given q-value threshold. q-values were determined using the qvalue package in R \citep{qvalueSoftware}.}\label{Figure:HypothesisTesting}
\end{figure}

\begin{figure}[t!]
    \centering
    \includegraphics[scale=0.4]{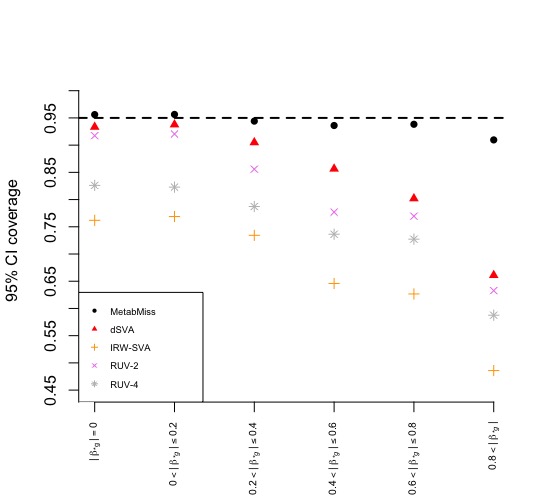}
    \caption{The fraction of effects of interest $\left\lbrace \beta_{g} \right\rbrace_{g \in \Missing}$ in all 60 simulated datasets that lie in their respective 95\% confidence intervals $\hat{\beta}_g \pm 1.96\left\lbrace \hat{\V}\left( \hat{\beta}_g\right) \right\rbrace^{1/2}$, stratified by $\lvert \beta_g \rvert$. The coverage when $\bm{C}$ was ignored was uniformly less than IRW-SVA's.}\label{Figure:CI}
\end{figure}

\section{Data analysis}
\label{section:DataAnlysis}
We used blood plasma metabolomic data measured in $n=533$ six year old Danish children enrolled in the Copenhagen Prospective Studies of Asthma in Children cohort \citep{COPSAC} to demonstrate the importance of accounting for both missing data and unobserved covariates in untargeted metabolomic data. Table \ref{Table:COPSACMiss} provides an overview of the extent of the missing data in each of the $p=1138$ measured metabolites. We excluded metabolites that were missing in more than 50\% of the samples and set $\epsilon_{\miss} = 0.05$ and $K_{\miss}=K_{PA}/2 = 10$ when estimating the missingness mechanisms with HB-GMM, where $K_{PA}$ was parallel analysis' \citep{BujaFA} estimates for $K$. $K_{\miss}$ was chosen using the procedure outlined in Section \ref{subsection:SelectingInstruments}. 

\COPSACMiss

\indent Once we estimated the missingness mechanisms, we could easily assess the relationships between the quantified metabolome and the many recorded phenotypes using MetabMiss. We were particularly interested in phenotypes related to asthma, and present the results for specific airway resistance ($\sraw$), which measures airway resistance to flow \citep{sRAW}. Using the design matrix $\bm{X}=\left(\bm{X}_{\interest} \, \bm{1}_n\right)$, where $\bm{X}_{\interest} \in \mathbb{R}^n$ was each individual's measured $\sraw$ value, we estimated $K$ with \citep{SVAinR} using the metabolites with complete data and regressed the quantified metabolites onto $\bm{X}$ using MetabMiss. We present the Q-Q plots of \textit{P} values in Figure \ref{Figure:RealData:sRAW}.\par
\indent Figure \ref{Figure:RealData:sRAW} shows that MetabMiss not only corrects the minor \textit{P} value inflation, but also empowers the analysis by reducing the residual variance. While the analysis with dSVA only identified a single metabolite, MetabMiss identified six additional metabolites at a a q-value threshold of 0.2: two sphingolipids, a benzoate derivative, pyruvate and three derivatives of piperine, which is an alkaloid found in black pepper. A reduction in sphingolipid synthesis was associated with increased airway hyperractivity in children \citep{Sphingo}, which is congruent with the estimated sign of the $\sraw$ effect on the intensity of the two sphingolipid metabolites. Benzoate preservatives have been linked to lung function-related phenotypes \citep{BenzoateAsthma,BenzoateRhinitis}, and pyruvate and lactate (q-value = 0.23) levels have previously been associated with asthma \citep{Pyruvate1,Pyruvate2}.\par 
\indent The three derivatives of piperine were particularly interesting in the context of our methodology because all three had between 12\% and 48\% missing data with J-test \textit{P} values between 0.77 and 0.99 (see Section \ref{subsection:Jtest}), suggesting that Model \eqref{equation:MissingMech} is a reasonable model for their missingness mechanisms. We found that higher concentrations of these three metabolites were associated with increased airway resistance. This corroborates the known biological impact of piperine, as it has been shown that piperine has a strong affinity for and activates TRPV1 cation channels on the ends of somatic and visceral parasympathetic nervous system sensory fibers \citep{TRPV1Piperine}. This triggers mast cells, bronchial epithelial cells and immune cells to release proinflammatory cytokines \citep{TRPV1ProInflammatory}, which ultimately causes bronchoconstriction \citep{TRPV1Lung,TRPV1Blocking}.\par 
\indent An interesting feature of Figure \ref{Figure:RealData:sRAW} is that the ordering of the \textit{P} values changes when one accounts for the missing data. In fact, we would not have identified the association between $\sraw$ and piperine concentration if one used existing latent factor correction methods. This is in part because metabolites with missing data had a slightly different latent factor signature than those with complete data, which is why MetabMiss is able to better account for the latent variation than existing methods.

\begin{figure}[!ht]
\centering
\includegraphics[scale=0.45]{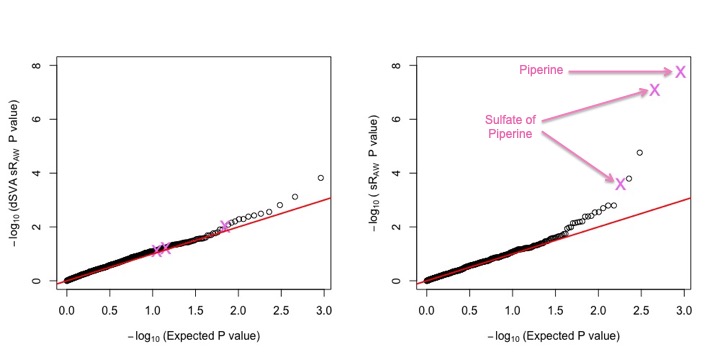}
\caption{A Q-Q plot of \textit{P} values for the null hypotheses $H_{0,g}:\beta_g = 0$ when $\bm{C}$ is estimated with dSVA using metabolites with complete data and the missing data are treated as MCAR (left), and using MetabMiss (right). The x-axis is the expected ordered \textit{P} value under the null hypothesis, assuming all tests are independent. The three quantified derivatives of piperine are each labeled with a violet ``$\times$".}\label{Figure:RealData:sRAW}
\end{figure}


\section{Discussion}
\label{section:Discussion}
We have presented, to the best of our knowledge, the first method to simultaneously account for latent factors and non-ignorable missing data in untargeted metabolomic data. Our method simplifies this complex problem by modularizing the estimation of each metabolite-dependent missingness mechanism and latent factors, and does so without assuming a specific probability model for the missing data. This modularization also makes modern metabolomic data analysis tractable, since our estimators for the missingness mechanism only depend on $\bm{Y}$ and are invariant to the choice of model matrix $\bm{X}$.\par
\indent An important tuning parameter in Algorithm \ref{algorithm:SelectU} is the number of estimated instruments to use to estimate $\alpha_g$ and $\delta_g$. We use two instruments ($\hat{\bm{U}}_g \in \mathbb{R}^{n \times 2}$) so that we have an extra degree of freedom to use the Sargan-Hansen $J$ statistic to identify metabolites whose missingness mechanisms may not follow \eqref{equation:MissingMech}. Including additional instruments may improve the efficiency of our estimator, which could be an interesting area of future research.

\section*{Acknowledgements}
We thank Hans Bisgaard, Klaus B{\o}nnelykke and the rest of the researchers at COPSAC for providing the data to make this research possible. We also thank Morten Arendt Rasmussen, Daniela Rago and Donata Vercelli for comments and suggestions that have substantially improved this work. 

\bibliographystyle{acmtrans-ims}
\bibliography{Bibliography}

\pagebreak
\renewcommand{\thefigure}{S\arabic{figure}}
\renewcommand{\thesection}{S\arabic{section}}
\setcounter{figure}{0}
\setcounter{table}{0}
\setcounter{section}{0}

\begin{changemargin}{-1cm}{-1cm}
\begin{center}
    {\bf \Large Supplemental material for ``Estimation and inference in metabolomics with non-random missing data and latent factors"}
\end{center}

\section{Notation}
\label{section:Supp:notation}
Besides the notation introduced in Section \ref{subsection:Notation} of the main text, we use the following notation throughout the supplement. Let $\bm{M} \in \mathbb{R}^{n \times m}$. We define $\bm{M}_{i\bigcdot} \in \mathbb{R}^{m}$ and $\bm{M}_{\bigcdot j} \in \mathbb{R}^{n}$ to be the $i$th row and $j$th column of $\bm{M}$, respectively. If $d > 0$ is the dimension of the null space of $\bm{M}^{\T}$, we define $\bm{Q}_M \in \mathbb{R}^{n \times d}$ be a matrix whose columns form an orthonormal basis for the null space of $\bm{M}^{\T}$. If $\bm{X}_1,\bm{X}_2,\ldots \in \mathbb{R}^{r \times s}$ is a sequence of random matrices (or vectors if $s=1$), we let $\bm{X}_n = O_P\left(a_n\right)$ and $\bm{X}_n = o_P\left(a_n\right)$ if $\norm{\bm{X}_n}_2/a_n = O_P(1)$ and $\norm{\bm{X}_n}_2/a_n = o_P(1)$ as $n \to \infty$, respectively.

\section{A bootstrap null distribution for the J statistics}
\label{section:Supp:Jstat}
Fix a $g \in \Missing$, define the null hypothesis $H_{0_g}$ to be that Model \eqref{equation:MissingMech} is correct and let
\begin{align*}
    J_g = n\bar{\bm{h}}_g\left\lbrace \hat{\alpha}_g^{\GMM}, \hat{\delta}_g^{\GMM} \right\rbrace^{\T} \bm{W}_g \bar{\bm{h}}_g\left\lbrace \hat{\alpha}_g^{\GMM}, \hat{\delta}_g^{\GMM} \right\rbrace.
\end{align*}
We show in Corollary \ref{corollary:MetabMiss:Jtest} in Section \ref{section:Supp:EstGMMProof} that when $H_{0_g}$ and additional technical assumptions hold, $J_g \tdist \chi^2_1$ as $n,p \to \infty$. However, we find that we over-reject using this asymptotic distribution, which is consistent with previous practitioners' observations \citep{AntiCons_Jtest2,AntiCons_Jtest1,EmpLikeBoot}. We therefore use \cite{EmpLikeBoot} to develop a bootstrapped null distribution for $J_g$, the details of which are given below.\par
\indent The bootstrap population is the observed sample $\left\lbrace \left( y_{gi},r_{gi},\hat{\bm{u}}_{gi} \right) \right\rbrace_{i \in [n]}$. Let
\begin{align*}
    d\hat{F}_{g}^{\text{sample}}(\bm{x}) = \sum\limits_{i=1}^{n}n^{-1}I\left\lbrace \bm{x}=\left( y_{gi},r_{gi},\hat{\bm{u}}_{gi} \right) \right\rbrace
\end{align*}
be the empirical density that puts weights $1/n$ on each sample point. If $\hat{\bm{u}}_{gi}$ only had two elements, then
\begin{align*}
    \E_{\hat{F}_{g}^{\text{sample}}}\left[\bar{\bm{h}}_g\left\lbrace \hat{\alpha}_g^{\GMM}, \hat{\delta}_g^{\GMM}\right\rbrace \right] = \bm{0}_2.
\end{align*}
However, since $\hat{\bm{u}}_{gi}$ has three elements, the above expectation will not be zero, which would violate the requirement that the population moment be zero under $H_{0_g}$. We therefore use the empirical likelihood to derive an alternative sampling distribution, $\hat{F}_g$, that ensures the bootstrapped population mean of $\bar{\bm{h}}_g\left\lbrace \hat{\alpha}_g^{\GMM}, \hat{\delta}_g^{\GMM}\right\rbrace$ is zero. That is, we define
\begin{align*}
    d\hat{F}_{g}\left(\bm{x}\right) = \sum\limits_{i=1}^n \eta_{gi} I\left\lbrace \bm{x}=\left( y_{gi},r_{gi},\hat{\bm{u}}_{gi} \right) \right\rbrace,
\end{align*}
where
\begin{align*}
    \left\lbrace \eta_{g1},\ldots,\eta_{gn}\right\rbrace = \argmax_{\Delta_1,\ldots,\Delta_n \in [0,1]} \prod_{i=1}^n \Delta_i, \quad \sum\limits_{i=1}^n \Delta_i = 1, \quad \sum\limits_{i=1}^n \Delta_i \bm{h}_{gi}\left\lbrace \hat{\alpha}_g^{\GMM}, \hat{\delta}_g^{\GMM} \right\rbrace = \bm{0}_3.
\end{align*}
We subsequently recompute $J_g$ for each bootstrapped sample drawn from $\hat{F}_{g}$ to derive the bootstrap null distribution.

\section{Estimating latent covariates in models with nuisance covariates}
\label{section:Supp:Nuisance}
Here we extend our estimator for $\bm{C}$ defined in \eqref{equation:Chat} of Section \ref{section:CMNAR} when $\bm{X} = \left(\bm{X}_{\interest}\, \bm{X}_{\nuisance} \right)$, where $\bm{X}_{\interest} \in \mathbb{R}^{n \times d_{\interest}}$ and $\bm{X}_{\nuisance} \in \mathbb{R}^{n \times d_{\nuisance}}$. To do so, let
\begin{align*}
    \bm{\beta}_{g} = \begin{pmatrix}
    \bm{\beta}_{g}^{(\interest)}\\
    \bm{\beta}_{g}^{(\nuisance)}
    \end{pmatrix}, \quad g \in [p].
\end{align*}
We can then re-write \eqref{equation:EstC} as
\begin{align*}
    &\bm{y}_g = P_{X_{\nuisance}}^{\perp}\bm{X}_{\interest}\tilde{\bm{\beta}}_{g}^{(\interest)} + \bm{X}_{\nuisance}\tilde{\bm{\beta}}_{g}^{(\nuisance)} + \bm{C}_2 \bm{\ell}_{g} + \bm{e}_g, \quad \tilde{\bm{\beta}}_{g}^{(\interest)} = \bm{\beta}_{g}^{(\interest)} + \bm{\Omega}\bm{\ell}_{g}, \quad g\in[p]\\
    &\bm{\Omega} = \left( \bm{X}_{\interest}^{\T} P_{X_{\nuisance}}^{\perp} \bm{X}_{\interest}\right)^{-1}\bm{X}_{\interest}^{\T} P_{X_{\nuisance}}^{\perp}\bm{C}, \quad \bm{C}_2 = P_{X}^{\perp}\bm{C}, \quad g\in[p]
\end{align*}
where $\tilde{\bm{\beta}}_{g}^{(\nuisance)}$ is a nuisance parameter. We estimate $\bm{C}_2$ as we did in Section \ref{subsection:Cperp}, where now
\begin{align*}
    \hat{\tilde{\bm{\beta}}}_g = \begin{pmatrix}
    \hat{\tilde{\bm{\beta}}}_g^{(\interest)}\\
    \hat{\tilde{\bm{\beta}}}_g^{(\nuisance)}
    \end{pmatrix}
\end{align*}
and $\hat{\tilde{\bm{\beta}}}_g^{(\interest)} \in \mathbb{R}^{d_{\interest}}$, $\hat{\tilde{\bm{\beta}}}_g^{(\nuisance)} \in \mathbb{R}^{d_{\nuisance}}$. Using the same reasoning as when $\bm{X} = \bm{X}_{\interest}$ in Section \ref{section:CMNAR}, we model $\hat{\tilde{\bm{\beta}}}_g^{(\interest)} \asim \left( \bm{\beta}_{g}^{(\interest)} + \bm{\Omega}\bm{\ell}_{g}, \hat{\bm{v}}_g \right)$, where $\hat{\bm{v}}_g$ is a submatrix of the estimate for the variance defined in \eqref{equation:IPW:Variance} when $g \in \Missing$ or the ordinary least squares estimator for $\V\left( \hat{\tilde{\bm{\beta}}}_g^{(\interest)} \right)$ when $g \in \Observed$ using the design matrix $\left( \bm{X}_{\interest} \, \bm{X}_{\nuisance} \, \hat{\bm{C}}_2\right)$. We subsequently estimate $\bm{\Omega}$ with Algorithm \ref{algorithm:Omega}.

\section{Justification of the estimate for $\V\left( \hat{\bm{\eta}}_g\right)$ defined in \eqref{equation:IPW:Variance}}
\label{section:MetabMiss:Variance}
Here were perform an error analysis to justify using \eqref{equation:IPW:Variance} to estimate $\V\left( \hat{\bm{\eta}}_g\right)$, where $\bm{\eta}_g$ is as defined in Section \ref{section:IPW}. We also define
\begin{align*}
   \bm{s}_{gi} = \bm{z}_i\left(y_{gi} - \bm{z}_i^{\T} \bm{\eta}_{g} \right), \quad \hat{\bm{s}}_{gi} = \bm{z}_i\left(y_{gi} - \bm{z}_i^{\T} \hat{\bm{\eta}}_{g} \right), \quad g \in \Missing; i \in [n]
\end{align*}
where $\hat{\bm{\eta}}_{g}$ is defined in \eqref{equation:IPW:beta.hat}. We assume throughout this section that $\alpha_g$ and $\delta_g$ are known. Unless otherwise stated, all expectations and variances are taken conditional on $\bm{Z}$. By the derivation in Section \ref{subsection:IPW:variance}, our goal is to approximate $\gamma_{gi}^2\E\left( w_{gi}^2\bm{s}_{gi}\bm{s}_{gi}^{\T} \mid \bm{Z} \right)$ for all $g \in \Missing$ and $i \in [n]$.\par
\indent Fix a $g \in \Missing$ and define $\xi_{gi} = \gamma_{gi}w_{gi}$ and $\bm{W}_g = \diag\left( \xi_{g1}, \ldots, \xi_{gn}\right)$. When $\gamma_{gi}, \alpha_g$ and $\delta_g$ are known, the estimator $\hat{\bm{\eta}}_g$ is such that
\begin{align*}
    \sum\limits_{i=1}^n \xi_{gi}\bm{s}_{gi} &= \bm{Z}^{\T} \bm{W}_g\left( \bm{y}_g - \bm{Z}\bm{\eta}_{g}\right) = \bm{Z}^{\T} \bm{W}_g\left( \bm{y}_g - \bm{Z}\hat{\bm{\eta}}_g\right) + \bm{Z}^{\T} \bm{W}_g\bm{Z} \left(\hat{\bm{\eta}}_g - \bm{\eta}_{g} \right)\\
    &= \bm{0} + \bm{Z}^{\T} \bm{W}_g\bm{Z} \left(\hat{\bm{\eta}}_g - \bm{\eta}_{g} \right),
\end{align*}
where the last equality follows by the definition of $\hat{\bm{\eta}}_g$. Therefore, for all $i \in [n]$,
\begin{align*}
    \hat{\bm{s}}_{gi} &= \bm{z}_i \left( y_{gi} - \bm{z}_i^{\T} \bm{\eta}_{g}\right) + \bm{z}_i\bm{z}_i^{\T} \left(\bm{\eta}_{g} -\hat{\bm{\eta}}_g\right) = \bm{s}_{gi} - \bm{z}_i\bm{z}_i^{\T} \left( \bm{Z}^{\T} \bm{W}_g \bm{Z}\right)^{-1}\sum\limits_{j=1}^n \xi_{gj} \bm{s}_{gj}\\
    &= \left\lbrace I_d - \xi_{gi}\bm{z}_i \bm{z}_i^{\T} \left( \bm{Z}^{\T} \bm{W}_g \bm{Z}\right)^{-1}  \right\rbrace \bm{s}_{gi} - \bm{z}_i \bm{z}_i^{\T} \left( \bm{Z}^{\T} \bm{W}_g \bm{Z}\right)^{-1} \sum\limits_{j \neq i} \xi_{gj} \bm{s}_{gj}\\
    &= \bm{z}_i\left\lbrace 1-\xi_{gi}\bm{z}_i^{\T} \left( \bm{Z}^{\T} \bm{W}_g \bm{Z}\right)^{-1}\bm{z}_i \right\rbrace e_{gi} - \bm{z}_i \bm{z}_i^{\T} \left( \bm{Z}^{\T} \bm{W}_g \bm{Z}\right)^{-1} \sum\limits_{j \neq i} \xi_{gj} \bm{s}_{gj}\\
    &= \left( 1-h_{gi}\right)\bm{s}_{gi} - \bm{z}_i \bm{z}_i^{\T} \left( \bm{Z}^{\T} \bm{W}_g \bm{Z}\right)^{-1} \sum\limits_{j \neq i} \xi_{gj} \bm{s}_{gj}
\end{align*}
where $h_{gi}$ is the $i$th leverage score of $\bm{W}_g^{1/2}\bm{Z}$. For $\bm{A}_i = \sum\limits_{j \neq i} \xi_{gj} \bm{s}_{gj}$, the naive plug-in estimator can be expressed as
\begin{align}
    \gamma_{gi}^2 w_{gi}^2 \hat{\bm{s}}_{gi}\hat{\bm{s}}_{gi}^{\T} =& \xi_{gi}^2 \left( 1-h_{gi}\right)^2 \bm{s}_{gi}\bm{s}_{gi}^{\T} - \left(1-h_{gi} \right)\bm{z}_i \bm{z}_i^{\T} \left( \bm{Z}^{\T} \bm{W}_g \bm{Z}\right)^{-1}  \bm{A}_i\xi_{gi}^2\bm{s}_{gi}^{\T}\nonumber\\
    &-\left\lbrace \left(1-h_{gi} \right)\bm{z}_i \bm{z}_i^{\T} \left( \bm{Z}^{\T} \bm{W}_g \bm{Z}\right)^{-1} \bm{A}_i \xi_{gi}^2\bm{s}_{gi}^{\T} \right\rbrace^{\T}\nonumber\\
    \label{equation:Var:Naive}
    &+\xi_{gi}^2\left(1-h_{gi} \right)^2 \bm{z}_i \bm{z}_i^{\T} \left( \bm{Z}^{\T} \bm{W}_g \bm{Z}\right)^{-1} \bm{A}_i \bm{A}_i^{\T} \left( \bm{Z}^{\T} \bm{W}_g \bm{Z}\right)^{-1}\bm{z}_i \bm{z}_i^{\T} ,
\end{align}
and the corrected estimator as
\begin{align}
  \left(1-h_{gi} \right)^{-2}\gamma_{gi}^2 w_{gi}^2 \hat{\bm{s}}_{gi}\hat{\bm{s}}_{gi}^{\T} = & \xi_{gi}^2 \bm{s}_{gi}\bm{s}_{gi}^{\T} - \left(1-h_{gi} \right)^{-1}\bm{z}_i \bm{z}_i^{\T} \left( \bm{Z}^{\T} \bm{W}_g \bm{Z}\right)^{-1}  \bm{A}_i\xi_{gi}^2\bm{s}_{gi}^{\T}\nonumber\\
  &- \left\lbrace \left(1-h_{gi} \right)^{-1}\bm{z}_i \bm{z}_i^{\T} \left( \bm{Z}^{\T} \bm{W}_g \bm{Z}\right)^{-1}  \bm{A}_i\xi_{gi}^2\bm{s}_{gi}^{\T} \right\rbrace^{\T}\nonumber\\
  \label{equation:Var:Corrected}
  &+ \xi_{gi}^2 \bm{z}_i \bm{z}_i^{\T} \left( \bm{Z}^{\T} \bm{W}_g \bm{Z}\right)^{-1} \bm{A}_i \bm{A}_i^{\T} \left( \bm{Z}^{\T} \bm{W}_g \bm{Z}\right)^{-1}\bm{z}_i \bm{z}_i^{\T}.
\end{align}
\indent We first note that $\bm{A}_i$ is independent of $r_{gi}$ and $y_{gi}$ and $\E\left( \bm{A}_i\right) = \bm{0}$. Therefore, if we ignore the uncertainty in $n^{-1}\bm{Z}^{\T}\bm{W}_g \bm{Z}$, the second and third terms in \eqref{equation:Var:Naive} and \eqref{equation:Var:Corrected} will have expectation $\bm{0}$. The last terms are positive semi-definite, where since $\left(r_{g1},y_{g1} \right),\ldots,\left(r_{gn},y_{gn} \right)$ are independent, $\left(\bm{Z}^{\T} \bm{W}_g\bm{Z} \right)^{-1}\bm{A}_i\bm{A}_i \left(\bm{Z}^{\T} \bm{W}_g\bm{Z} \right)^{-1}$ will have eigenvalues that are $O_P\left( n^{-1}\right)$ under suitable regularity conditions. However, the first term in \eqref{equation:Var:Naive} will tend to be small and therefore underestimate $\gamma_{gi}^2\E\left(r_{gi}w_{gi}^2 \bm{s}_{gi}\bm{s}_{gi}^{\T} \right)$, especially when $d$ or the weights $w_{gi}$ are large, since this will increase leverage scores. Further, unlike data whose missing values are missing at random, the leverage scores are correlated with $\bm{s}_{gi}$, where large values of $\bm{s}_{gi}$ typically imply the weights will be large. These two facts cause the naive plug-in sandwich estimator to underestimate the variance. The corrected estimator circumvents these issues because the first term in \eqref{equation:Var:Corrected} is exactly the term whose expectation we are attempting to estimate.

\section{Additional simulation results}
\label{section:Supp:AddSimResults}

\subsection{Accuracy of our estimates for $\alpha_{g}$ and $\delta_{g}$}
\label{subsection:Supp:SimMisParams}
We simulated and analyzed 20 datasets according to \eqref{equation:MetabMiss:Simulation} with $\Psi(x) = F_4(x)$ to compare our method's (HB-GMM) estimates for $\alpha_{g}$ and $\delta_{g}$ defined in \eqref{equation:HBGMM_estimates} to the those proposed in \cite{GMM_MNAR}. The latter are simply $\hat{\alpha}_g^{\GMM}$ and $\hat{\delta}_g^{\GMM}$ defined in \eqref{equation:TwoStepGMM}. We could not compare it to the estimators proposed in \cite{MNAR_SILAC} or \cite{SILAC_MNAR2}, because they assume the missingness mechanisms are the same for each analyte. We remark that the estimators for $\alpha_{g}$ and $\delta_{g}$ these authors proposed in the aforementioned articles rely on the assumption that $y_{gi}$ is normally distributed.\par
\indent  The results are given in Figure \ref{Figure:MissingParams}. Our Bayesian estimator (HB-GMM) outperforms the standard two-step generalized method of moments estimator proposed in \cite{GMM_MNAR} (GMM), which illustrates the advantages of pooling information across metabolites.

\begin{figure}[b!]
\centering
    \includegraphics[scale=0.4]{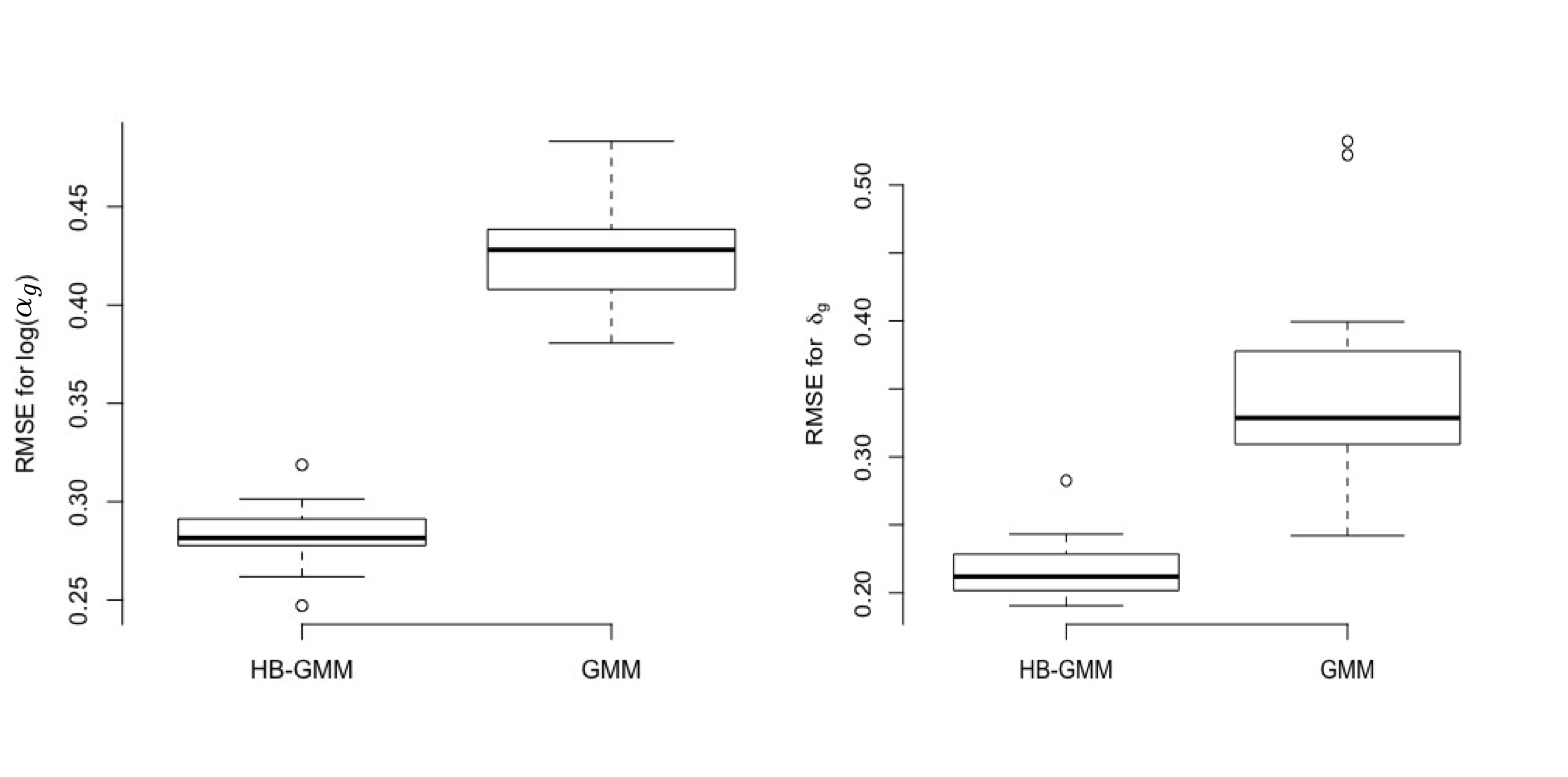}
    \caption{The root mean squared error (RMSE) $\left[ \lvert \Missing \rvert^{-1}\sum\limits_{g \in \Missing}\left\lbrace \log\left( \hat{\alpha}_g\right) - \log\left( \alpha_g\right) \right\rbrace^2 \right]^{1/2}$ (left) and $\left\lbrace \lvert \Missing \rvert^{-1}\sum\limits_{g \in \Missing}\left(  \hat{\delta}_g -  \delta_g \right)^2 \right\rbrace^{1/2}$ (right) over 20 simulations. The estimates $\log\left( \hat{\alpha}_g\right)$ and $\hat{\delta}_g$ were either our Bayesian estimators proposed in \eqref{equation:HBGMM_estimates:alphadelta} (HB-GMM), or the two-step estimator proposed in \cite{GMM_MNAR} (GMM).}\label{Figure:MissingParams}
\end{figure}

\subsection{$J$ statistics from simulated data}
\label{seuction:Supp:JstatSim}
Here we examined the uniformity of the $J$ statistics \textit{P} values from the data simulated in Section \ref{section:simulations} when $\Psi(x) = F_4(x)$. The results are given in Figure \ref{Figure:MetaMiss:Jtest}, where the \textit{P} values appear to be uniformly distributed.

\begin{figure}[t!]
   \centering
   \includegraphics[scale=0.3]{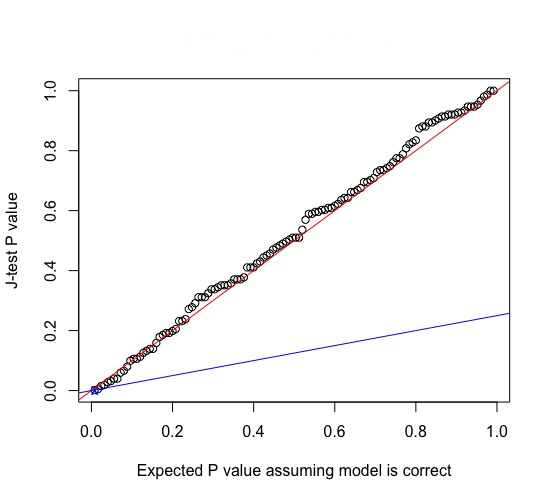}
   \includegraphics[scale=0.3]{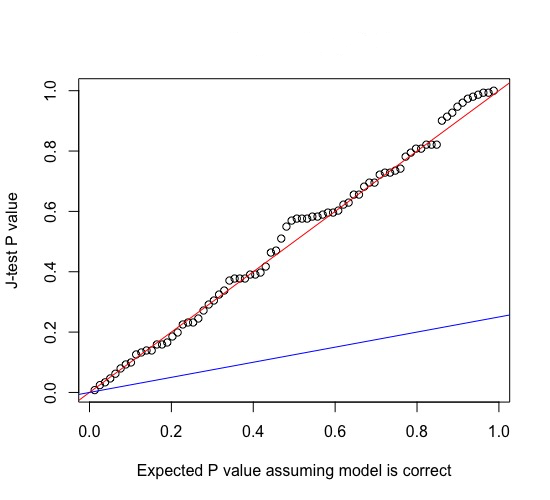}
   \caption{Bootstrapped $J$ statistic \textit{P} values from a single simulated dataset for metabolites with between 5\% and 35\% missing data (left) and 35\% and 50\% missing data (right). The red line is the line $y=x$ and the blue line is the line $y=x/4$. Points that lie below the blue line would be rejected by the Benjamini-Hochberg procedure at a level $\alpha = 0.25$, and are marked with a blue ``$\times$".}\label{Figure:MetaMiss:Jtest}
\end{figure}

\subsection{Simulations with smaller sample sizes}
\label{subsection:Supp:Smalln}
Here we analyze additional data simulated according to \eqref{equation:MetabMiss:Simulation} with $\Psi(x)=\exp(x)/\left\lbrace 1+\exp(x) \right\rbrace$ and $n=300$ or $n=100$ to demonstrate our method's performance on data with smaller sample sizes. Just like we did in Section \ref{subsection:SimResults}, we analyzed each of the 60 simulated datasets for each value of $n$ with MetabMiss by making the incorrect assumption that $\Psi(x)=F_4(x)$. For simplicity of presentation, we only report each method's potential to estimate and perform inference on $\beta_{g}$ for $g \in \Missing$. The results for each set of simulations are given in Figures \ref{Figure:MetaMiss:FDR_Supp} and \ref{Figure:MetaMiss:CI_Supp}.

\begin{figure}[b!]
   \centering
   \includegraphics[scale=0.5]{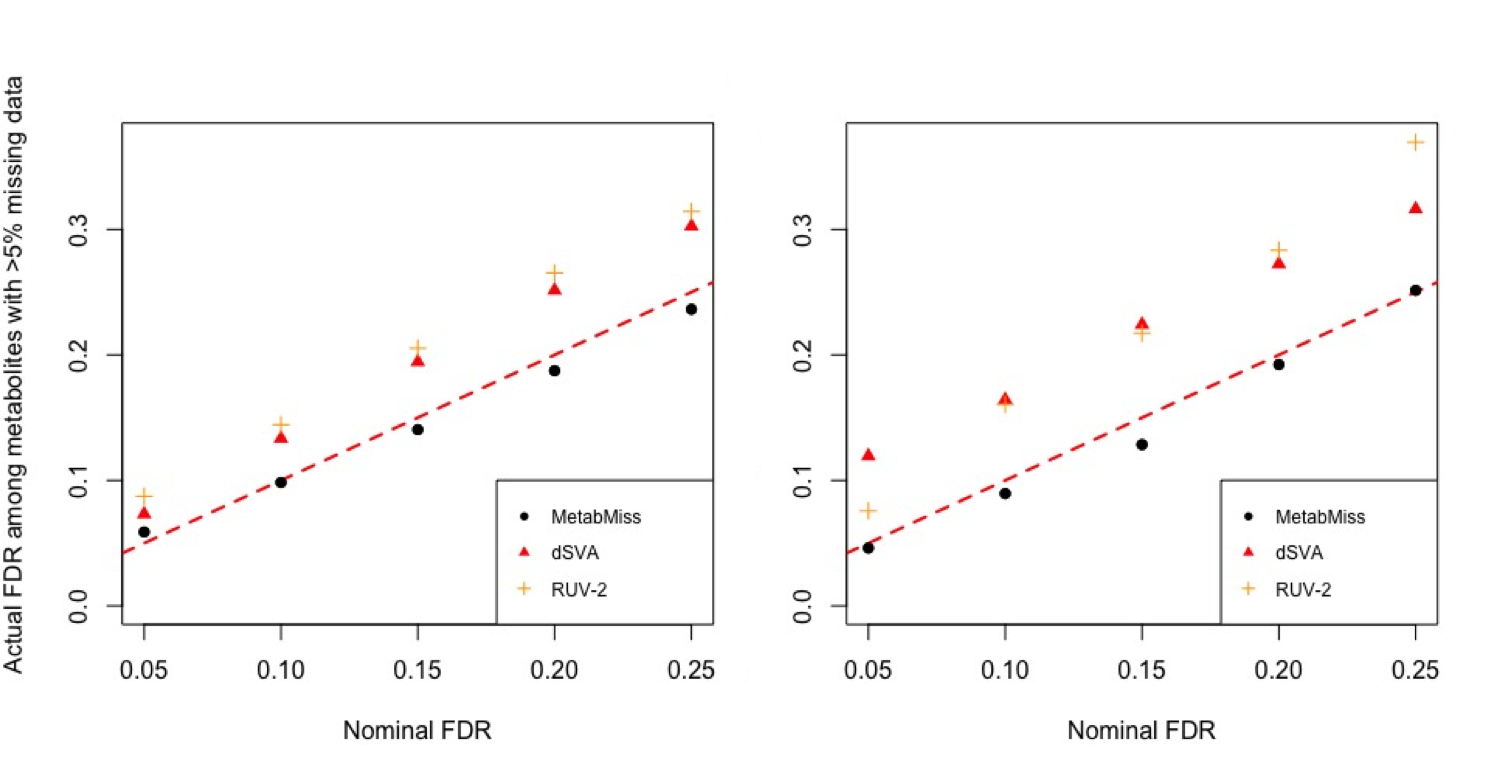}
   \caption{The false discovery rate (FDR) among rejected metabolites with $>5\%$ but $\leq 50\%$ missing data as a function of q-value threshold (Nominal FDR) when $n=300$ (left) and when $n=100$ (right). IRW-SVA's and RUV-4's performance was uniformly worse that dSVA's in both simulation settings. The dashed red line is the line $y=x$.}\label{Figure:MetaMiss:FDR_Supp}
\end{figure}

\begin{figure}[b!]
   \centering
   \includegraphics[scale=0.20]{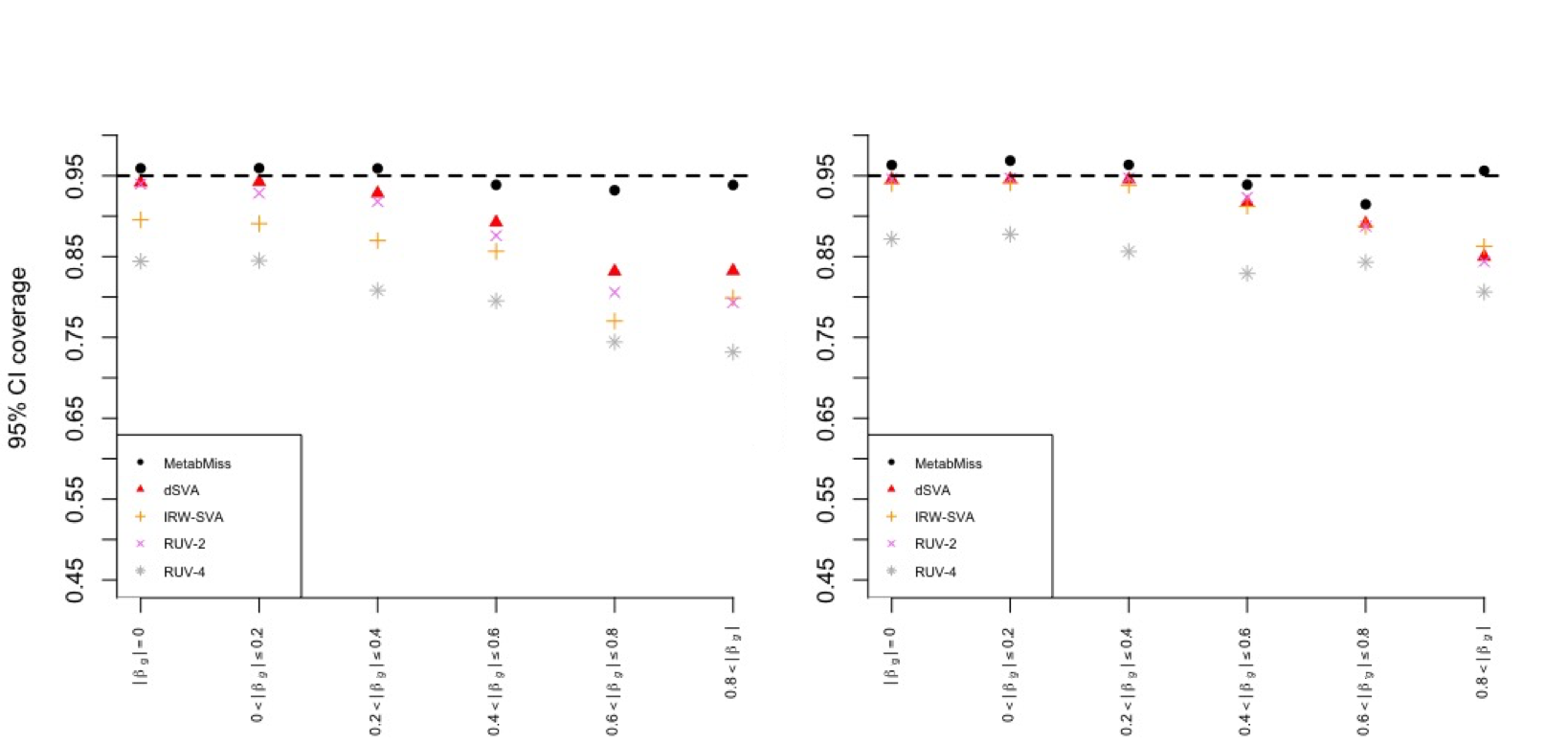}
   \caption{The 95\% confidence interval coverage when $n=300$ (left) and when $n=100$ (right). These plots are analogous to Figure \ref{Figure:CI}.}\label{Figure:MetaMiss:CI_Supp}
\end{figure}

\section{A mathematical justification of Algorithm \ref{algorithm:SelectU}}
\label{section:Supp:SelectionProof}
In this section, we justify the instrumental variable selection step Algorithm \ref{algorithm:SelectU}. We first prove the asymptotic properties of $\hat{\bm{C}}_{\miss}$ in Lemma \ref{lemma:MetabMiss:PCA} and then prove in Theorem \ref{theorem:MetabMiss:OLS_MNAR} that the \textit{P} value $p_{g,k}$ in Step \ref{item:InstrumentSelection:Pvalue} is asymptotically uniform under the null hypothesis that the $k$th column of $\hat{\bm{C}}_{\miss}$ is independent of $\bm{y}_g$. We lastly prove in Theorem \ref{theorem:Supp:OrderEffects} that under weak conditions, one is justified ignoring the uncertainty in the estimated indices $g_1$ and $g_2$ from Step \ref{item:InstrumentSelection:Indices} when deriving the asymptotic distribution in \eqref{equation:AsymptoticDistn}. We first state the assumptions that we will use throughout this section, as well as Sections \ref{section:Supp:NormalProof} are \ref{section:Supp:EstGMMProof}.
\begin{assumption}
\label{assumption:Basic}
Suppose $\bm{y}_g=\left(y_{g1}\cdots y_{gn}\right)^{\T} = \bm{Z}\bm{\xi}_{g} + \tilde{\bm{C}}\tilde{\bm{\ell}}_{g} + \bm{e}_g$ for all $g \in [p]$, where $\bm{e}_g=\left(e_{g1}\cdots e_{gn}\right)^{\T}$, $\bm{Z} \in \mathbb{R}^{n \times t}$ for $t \geq 1$ and $\tilde{\bm{C}} \in \mathbb{R}^{n \times \tilde{K}}$ for $\tilde{K} \geq 1 \wedge (3-t)$, and that Model \eqref{equation:MissingMech} holds for $g \in \Missing$. Define $\tilde{\bm{L}}_{\Observed}$ to be the sub-matrix of $\tilde{\bm{L}} = \left(\tilde{\bm{\ell}}_1 \cdots \tilde{\bm{\ell}}_p \right)^{\T}$, restricted to the rows $g \in \Observed$, and let $p_s = \abs{\Observed}$. Then following hold for some constant $c_1 > 1$:
\begin{enumerate}[label=(\roman*)]
    \item $\bm{Z} \in \mathbb{R}^{n \times t}$ is non-random, $\norm{\bm{Z}}_{\infty}\leq c_1$, $\bm{1}_n \in \im\left(\bm{Z}\right)$ and $\mathop{\lim}\limits_{n \to \infty} n^{-1}\bm{Z}^{\T} \bm{Z} = \bm{\Sigma}_Z \succ \bm{0}$.
    \item Both $p$ and $p_s$ are non-decreasing function of $n$ and $\mathop{\lim}\limits_{n \to \infty}p/n, \mathop{\lim}\limits_{n \to \infty}p_s/n \in \left( 0,\infty\right)$.
    \item $\tilde{\bm{L}}_{\Observed}^{\T} \tilde{\bm{L}}_{\Observed}$ is diagonal with non-increasing elements and $\mathop{\lim}\limits_{p_s \to \infty}p_s^{-1}\tilde{\bm{L}}_{\Observed}^{\T} \tilde{\bm{L}}_{\Observed} = \diag\left(\lambda_1,\ldots,\lambda_{\tilde{K}}\right)$. For $K_{\miss} \in \left[1,\tilde{K}\right]$ and $\lambda_{\tilde{K}+1}=0$, the eigenvalues satisfy $\lambda_1 > \cdots > \lambda_{K_{\miss}} > \lambda_{\left(K_{\miss}+1\right)} \geq 0$.\label{item:AssumptionBasics:L}
    \item $\tilde{\bm{C}} = \bm{Z}\bm{A} + \bm{\Xi}$ for some non-random matrix $\bm{A} \in \mathbb{R}^{t \times \tilde{K}}$. The random matrix $\bm{\Xi}$ is independent of $\bm{e}_1,\ldots,\bm{e}_p$, $ \bm{\Xi}_{1 \bigcdot},\ldots,\bm{\Xi}_{n \bigcdot}$ are independent and identically distributed and $\bm{\Xi}_{1 \bigcdot} \asim \left(\bm{0}_{\tilde{K}}, I_{\tilde{K}} \right)$. Further, the entries of $\bm{\Xi}$ have uniformly bounded eighth moments.
    \item $e_{g1},\ldots,e_{gn} \asim \left( 0,\sigma_g^2\right)$ are independent and identically distributed with uniformly bounded fourth moments for all $g \in [p]$. Further, $\bm{e}_1,\ldots,\bm{e}_p$ are independent.
    \item The function $\Psi(x):\mathbb{R} \to (0,1)$ defined in \eqref{equation:MissingMech} is continuous and strictly increasing, where $\mathop{\lim}\limits_{x \to -\infty}\Psi(x)=0$ and $\mathop{\lim}\limits_{x \to \infty}\Psi(x)=1$. Further, $\alpha_{g} > 0$ for all $g \in \Missing$.
    \item If $g \in \Observed$, $\Prob\left(r_{gi}=1\right)=1$ for all $i\in[n]$.\label{item:AssumptionBasics:MCAR}
\end{enumerate}
\end{assumption}

\begin{remark}
\label{remark:Z}
The covariates $\bm{Z}$ are a set of observed covariates that systematically effect $\bm{Y}$ and can be used as observed instruments when estimating the missingness mechanisms. Since such factors are rarely observed, we had assumed for simplicity that $\bm{Z}=\bm{1}_n$ in Sections \ref{section:RoadMap} and \ref{section:HBGMM}. We let $\bm{Z}$ include more than the intercept to allow for the possibility the practitioner observes additional covariates that can be used as instruments.
\end{remark}

\begin{remark}
\label{remark:C}
$\tilde{\bm{C}}$ in Assumption \ref{assumption:Basic} is not the same $\bm{C}$ defined in \eqref{equation:ModelForData}. It is instead the concatenation of $\bm{C}$ and $\bm{X}_{-\bm{Z}}$, the columns of $\bm{X}$ that are not a part of $\bm{Z}$.
\end{remark}

\begin{remark}
\label{remark:LIdentify}
Since $\tilde{\bm{L}}$ is identifiable up to multiplication by an invertible matrix on the right, the assumption that assumption that $\tilde{\bm{L}}^{\T}\tilde{\bm{L}}$ is diagonal with decreasing elements and $\C\left(\bm{\Xi}_{i \bigcdot}\right)=I_{\tilde{K}}$ for all $i\in[n]$ is without loss of generality.
\end{remark}

\begin{remark}
\label{remark:L}
Since $\tilde{\bm{C}} = \left(\bm{C}\, \bm{X}_{-\bm{Z}}\right)$ and $\bm{X}_{-\bm{Z}}$ will typically have a weak effect on $\bm{Y}$, the smallest eigenvalues of $p_s^{-1}\tilde{\bm{L}}^{\T}\tilde{\bm{L}}$ will tend to converge to 0 as $p_s$ gets large. Therefore, we only require a subset of $\lambda_1,\ldots,\lambda_{\tilde{K}}$ be bounded away from 0 in Item \ref{item:AssumptionBasics:L} in Assumption \ref{assumption:Basic}.
\end{remark}

\begin{remark}
\label{remark:AssumptionBasics:MCAR}
Item \ref{item:AssumptionBasics:MCAR} simplifies the problem by assuming that metabolites whose data are used to compute $\hat{\bm{C}}_{\miss}$ have no missing data. While this is typically not true in real data, $\epsilon_{\miss}$ in Algorithm \ref{algorithm:MissMechOverview} is set to be small enough that any bias incurred due to data that are missing not at random is negligible in practice.
\end{remark}

\begin{lemma}
\label{lemma:MetabMiss:Pz}
Fix a $g \in \Missing$. If Assumption \ref{assumption:Basic} holds, then $\Prob\left( r_{gi} = 1 \mid \bm{Z}_{i \bigcdot}\right) \geq c$ for all $i \in [n]$ for some constant $c > 0$ that does not depend on $n$.
\end{lemma}

\begin{proof}
Define $\tilde{\bm{\xi}}_{g} = \bm{\xi}_{g}+\bm{A}\tilde{\bm{\ell}}_{g}$. By Assumption \ref{assumption:Basic}, $\tilde{e}_{g1} = y_{g1}- \bm{Z}_{1 \bigcdot}^{\T} \tilde{\bm{\xi}}_{g} ,\ldots,\tilde{e}_{gn} =y_{gn}- \bm{Z}_{n \bigcdot}^{\T} \tilde{\bm{\xi}}_{g}$ are identically distributed. Let $\mu_n = \mathop{\min}\limits_{i \in [n]}\left(  \bm{Z}_{i \bigcdot}^{\T} \tilde{\bm{\xi}}_{g}\right)$. Note that $\mu_n \geq \mu$ for some constant $\mu > -\infty$ because the entries of $\bm{Z}$ are uniformly bounded. Then because $\Psi(x)$ is strictly increasing and non-zero, for any $n > 0$ and $i \in [n]$,
\begin{align*}
    \Prob\left( r_{gi}=1 \mid   \bm{Z}_{i \bigcdot}\right) &= \E\left[ \Psi\left\lbrace \alpha_{g}\left( \bm{Z}_{i \bigcdot}^{\T} \tilde{\bm{\xi}}_{g} + \tilde{e}_{gi}  - \delta_{g}\right) \right\rbrace \mid  \bm{Z}_{i \bigcdot} \right] \geq \E\left[ \Psi\left\lbrace \alpha_{g}\left( \mu + \tilde{e}_{gi}  - \delta_{g} \right) \right\rbrace \right]\\
     &= \E\left[ \Psi\left\lbrace \alpha_{g}\left( \mu + \tilde{e}_{g1}  - \delta_{g} \right) \right\rbrace \right] > 0.
\end{align*}
\end{proof}

We now prove a useful lemma that describes the properties of $\hat{\bm{C}}_{\miss}$, defined in Section \ref{subsection:EstimatingInstruments}, when $\epsilon_{\miss} = 0$. For notational convenience, we define $K \leftarrow \tilde{K}$, $\bm{C} \leftarrow \tilde{\bm{C}}$, $\bm{\ell}_g \leftarrow \tilde{\bm{\ell}}_g$, $\bm{L} \leftarrow \tilde{\bm{L}}$ and $\bm{L}_{\Observed} \leftarrow \tilde{\bm{L}}_{\Observed}$ for the remainder of the supplement.

\begin{lemma}
\label{lemma:MetabMiss:PCA}
Suppose Assumption \ref{assumption:Basic} holds, and let $n^{-1/2}\hat{\bm{C}}_2$ be the first $K$ right singular vectors of $\bm{Y}_{\Observed}P_{Z}^{\perp} = \bm{L}_{\Observed}\left( P_{Z}^{\perp}\bm{C}\right)^{\T} + \bm{E}_{\Observed}P_{Z}^{\perp}$. Define
\begin{align}
\label{equation:MetabMiss:Ctilde}
    \bm{C}_2 = P_{Z}^{\perp}\bm{C}\left( n^{-1}\bm{C}^{\T} P_{Z}^{\perp}\bm{C} \right)^{-1/2}\bm{U}, \quad \tilde{\bm{L}}_{\Observed} = \bm{L}_{\Observed}\left( n^{-1}\bm{C}^{\T} P_{Z}^{\perp}\bm{C} \right)^{1/2}\bm{U},
\end{align}
where $\bm{U}$ is a unitary matrix such that $\tilde{\bm{L}}_{\Observed}^{\T}\tilde{\bm{L}}_{\Observed}$ is diagonal with non-increasing elements. Then
\begin{align}
\label{equation:MetabMiss:C2hat}
    \hat{\bm{C}}_2 = \bm{C}_2\hat{\bm{v}} + n^{1/2}\bm{Q}_{P_{Z}^{\perp}C}\hat{\bm{w}},
\end{align}
where for $\bm{a}_k \in \mathbb{R}^{K}$ the $k$th standard basis vector,
\begin{subequations}
\label{equation:MetabMiss:wvResults}
\begin{align}
    \label{equation:MetabMiss:wvResults:bound}
    &\norm{\hat{\bm{v}}_{\bigcdot k} - \bm{a}_k}_2 = O_P\left\lbrace \left(np_s\right)^{-1/2} \right\rbrace, \quad \norm{\hat{\bm{w}}_{\bigcdot k}}_2 = O_P\left( p_s^{-1/2}\right), \quad k\in\left[K_{\miss}\right]\\
    \label{equation:MetabMiss:wvResults:w}
    & \norm{ \hat{\bm{w}}_{\bigcdot k} - \left(n\lambda_k p_s\right)^{-1} \bm{Q}_{P_{Z}^{\perp}C}^{\T} P_{Z}^{\perp}\bm{E}_{\Observed}^{\T} \left( n^{1/2}\tilde{\bm{L}}_{\Observed} + n^{-1/2}\bm{E}_{\Observed}P_{Z}^{\perp}\bm{C}_2 \right)\hat{\bm{v}}_{\bigcdot k} }_2 = O_P\left\lbrace \left(np_s\right)^{-1/2} \right\rbrace, \quad k\in \left[K_{\miss}\right].
\end{align}
\end{subequations}
\end{lemma}

\begin{proof}
For notational simplicity, we drop the subscript $\Observed$ and redefine $p \leftarrow p_s$, $\lambda_k \leftarrow n\lambda_k$ for all $k \in [K]$ and $\bm{C}_2 \leftarrow n^{-1/2}\bm{Q}_Z^{\T}\bm{C}_2$. Therefore, $\bm{Y}\bm{Q}_Z^{\T} = n^{1/2}\tilde{\bm{L}}\bm{C}_2^{\T} + \bm{E}\bm{Q}_Z$, $\bm{C}_2^{\T} \bm{C}_2 = I_K$ and $np^{-1}\tilde{\bm{L}}^{\T} \tilde{\bm{L}} = \diag\left(\gamma_1,\ldots,\gamma_K \right)$, where $\gamma_k = \lambda_k\left\lbrace 1+o_P(1) \right\rbrace$ for all $k \in \left[K_{\miss}\right]$ by Assumption \ref{assumption:Basic}. Define $\bm{E}_1 = \bm{E}\bm{Q}_Z\bm{C}_2$, $\bm{Q} = \bm{Q}_{\bm{C}_2}$, $\bm{E}_2 = \bm{E}\bm{Q}_Z\bm{Q}$ and $\bm{\Sigma} = \diag\left( \sigma_1^2,\ldots,\sigma_p^2\right)$. Define
\begin{align}
    \bm{S} &= \left(p \lambda_{K_{\miss}} \right)^{-1}\left(\bm{C}_2 \, \bm{Q} \right)^{\T} \bm{Y}^{\T}P_{Z}^{\perp} \bm{Y} \left(\bm{C}_2 \, \bm{Q} \right)\nonumber\\
    \label{equation:MetabMiss:S}
    & = \begin{pmatrix}
    \left( p \lambda_{K_{\miss}}\right)^{-1} \left(n^{1/2}\tilde{\bm{L}} + \bm{E}_1 \right)^{\T}\left(n^{1/2}\tilde{\bm{L}} + \bm{E}_1 \right) & \left( p \lambda_{K_{\miss}}\right)^{-1}\left(n^{1/2}\tilde{\bm{L}} + \bm{E}_1 \right)^{\T}\bm{E}_2 \\
    \left( p \lambda_{K_{\miss}}\right)^{-1}\bm{E}_2^{\T} \left(n^{1/2}\tilde{\bm{L}} + \bm{E}_1 \right) & \left( p \lambda_{K_{\miss}}\right)^{-1} \bm{E}_2^{\T}\bm{E}_2.
    \end{pmatrix}
\end{align}
By Theorem 5.37 of \cite{Vershynin}, $\norm{\left( p \lambda_{K_{\miss}}\right)^{-1} \bm{E}_2^{\T}\bm{E}_2}_2 = O_P\left(\lambda_{K_{\miss}} ^{-1}\right)$ under Assumption \ref{assumption:Basic}. Further, it is easy to see that conditional on $\bm{C}$,
\begin{align*}
    \left\lbrace n/\left( p\lambda_{K_{\miss}}\right) \right\rbrace^{1/2}\tilde{\bm{L}}^{\T}\bm{E}_{j_{\bigcdot r}} \asim \left( \bm{0}_K, n/\left( p\lambda_{K_{\miss}}\right)\tilde{\bm{L}}^{\T}\bm{\Sigma}\tilde{\bm{L}} \right),
\end{align*}
for $j=1,2$, meaning $\norm{\left\lbrace n/\left( p\lambda_{K_{\miss}}\right) \right\rbrace^{1/2}\tilde{\bm{L}}^{\T}\bm{E}_1}_2 = O_P(1)$ and $\norm{\left\lbrace n/\left( p\lambda_K\right) \right\rbrace^{1/2}\tilde{\bm{L}}^{\T}\bm{E}_2}_2 = O_P\left( n^{1/2}\right)$. Next, for $\rho = p^{-1}\sum\limits_{g=1}^p \sigma_g^2$, $\E\left( p^{-1}\bm{E}_1^{\T}\bm{E}_1 \mid \bm{C}\right) = \rho I_K$. Further, for $r,s \in [K]$,
\begin{align*}
    \V\left( p^{-1}\bm{E}_{1_{\bigcdot r}}^{\T}\bm{E}_{1_{\bigcdot s}} \mid \bm{C} \right) &= p^{-2}\sum\limits_{g=1}^p \V\left\lbrace \left(\bm{Q}_Z\bm{C}_{2_{\bigcdot r}}\right)^{\T} \bm{e}_g \left(\bm{Q}_Z\bm{C}_{2_{\bigcdot s}}\right)^{\T}\bm{e}_g \mid \bm{C} \right\rbrace\\
    & \leq p^{-2}\sum\limits_{g=1}^p \E\left[ \left\lbrace \left(\bm{Q}_Z\bm{C}_{2_{\bigcdot r}}\right)^{\T} \bm{e}_g\right\rbrace^4 \mid \bm{C} \right]^{1/2} \E\left[ \left\lbrace\left(\bm{Q}_Z\bm{C}_{2_{\bigcdot s}}\right)^{\T} \bm{e}_g\right\rbrace^4 \mid \bm{C} \right]^{1/2},
\end{align*}
where
\begin{align*}
    \E\left[ \left\lbrace \left(\bm{Q}_Z\bm{C}_{2_{\bigcdot r}}\right)^{\T} \bm{e}_g\right\rbrace^4 \mid \bm{C} \right] &= \sum\limits_{i,j=1}^n \left(\bm{Q}_Z\bm{C}_2 \right)_{i r}^2 \left(\bm{Q}_Z\bm{C}_2 \right)_{j r}^2\E\left(e_{gi}^2 e_{gj}^2 \mid \bm{C}\right) \leq c \left\lbrace \sum\limits_{i=1}^n \left(\bm{Q}_Z\bm{C}_2 \right)_{i r}^2\right\rbrace^2 = c
\end{align*}
for $c = \mathop{\max}\limits_{g \in [p]} \E\left( e_{g1}^4\right)$. This shows that $\norm{p^{-1}\bm{E}_1^{\T}\bm{E}_1 - \rho I_K}_2 = O_P\left( p^{-1/2}\right)$. Lastly, $\E\left( p^{-1}\bm{E}_2^{\T}\bm{E}_1 \mid \bm{C}\right) = \bm{0}$ and for $k \in [K], i \in [n-t-K]$,
\begin{align*}
    \E\left[ \left\lbrace p^{-1}\left( \bm{E}_2^{\T}\bm{E}_1\right)_{ik}\right\rbrace^2 \mid \bm{C}\right] &= \V\left\lbrace p^{-1}\left( \bm{E}_2^{\T}\bm{E}_1\right)_{ik}  \mid \bm{C}\right\rbrace= p^{-2}\sum\limits_{g=1}^p \V\left\lbrace \left( \bm{Q}_Z \bm{Q}\right)_{\bigcdot i}^{\T}\bm{e}_g \bm{e}_g^{\T} \left( \bm{Q}_Z \bm{C}_2\right)_{\bigcdot k} \mid \bm{C}\right\rbrace\\
    & \leq cp^{-1}
\end{align*}
for $c$ defined above. Therefore, $\norm{p^{-1}\bm{E}_2^{\T}\bm{E}_1}_2 = O_P\left\lbrace \left( np^{-1}\right)^{1/2}\right\rbrace$.\par
\indent Define $\mu_k = \gamma_k/\lambda_{K_{\miss}}$ and let $\hat{\mu}_k$ be the the $k$th eigenvalue of \eqref{equation:MetabMiss:S} for $k \in [K]$. By Weyl's inequality, the above work shows that $\hat{\mu}_k = \mu_k + \rho/\lambda_{K_{\miss}} + o_P(1)$ for all $k \in \left[K_{\miss}\right]$. Next, define
\begin{align*}
    \tilde{\bm{N}} &= \left\lbrace n/\left(p\lambda_{K_{\miss}} \right)\right\rbrace^{1/2}\tilde{\bm{L}} + \left(p\lambda_{K_{\miss}} \right)^{-1/2}\bm{E}_1,\quad\bm{B} = \left(p\lambda_{K_{\miss}} \right)^{-1}\left(n^{1/2}\tilde{\bm{L}} + \bm{E}_1 \right)^{\T}\bm{E}_2\\
    \bm{D} &= \left(p\lambda_{K_{\miss}} \right)^{-1} \bm{E}_2^{\T}\bm{E}_2
\end{align*}
and let $\hat{\bm{v}} \in \mathbb{R}^{K \times K}, \hat{\bm{w}} \in \mathbb{R}^{(n-t-K) \times K}$ be such that $\left( \hat{\bm{v}}^{\T}\,  \hat{\bm{w}}^{\T} \right)^{\T}$ are the first $K$ eigenvectors of \eqref{equation:MetabMiss:S}. Then
\begin{align*}
    \hat{\mu}_k \hat{\bm{v}}_{\bigcdot k} &= \left\lbrace \tilde{\bm{N}}^{\T} \tilde{\bm{N}} + \bm{B} \left( \hat{\mu}_k I_{n-t-K} - \bm{D} \right)^{-1}\bm{B}^{\T} \right\rbrace \hat{\bm{v}}_{\bigcdot k}, \quad k\in \left[K_{\miss}\right]\\
     \hat{\bm{w}}_{\bigcdot k} &= \left( \hat{\mu}_k I_{n-t-K} - \bm{D} \right)^{-1}\bm{B}^{\T} \hat{\bm{v}}_{\bigcdot k} \quad k\in \left[K_{\miss}\right].
\end{align*}
It is easy to see that \eqref{equation:MetabMiss:wvResults:bound} follows from the above work and the fact that $1-\lambda_{k+1}/\lambda_k \geq c_1^{-1}$ for all $k\in\left[K_{\miss}\right]$ in Item \ref{item:AssumptionBasics:L} of Assumption \ref{assumption:Basic}. \eqref{equation:MetabMiss:wvResults:w} follows because $\hat{\mu}_k/\mu_k = O_P\left\lbrace \left(np\right)^{-1/2} \right\rbrace = O_P\left(\lambda_{K_{\miss}}^{-1}\right)$ by Weyl's Theorem.
\end{proof}

\begin{lemma}
\label{lemma:MetabMiss:OLS_MNAR_observed}
Suppose the assumptions of Lemma \ref{lemma:MetabMiss:PCA} hold and let $g \in \Missing$ and $\bm{C}_2$ be as defined in \eqref{equation:MetabMiss:Ctilde}. For some $k \in \left[K_{\miss}\right]$, define $z_g$ to be the ordinary least squares z-score for the $[t+1]$st regressor from the regression $\bm{y}_g \sim \left( \bm{Z}\, \bm{C}_{2_{\bigcdot k}} \right)$, restricted to only the observed values of $\bm{y}_g$. Then if $\bm{C}_{2_{\bigcdot k}}$ is independent of $\bm{y}_g$, $z_g \tdist N_1\left( 0,1\right)$ as $n \to \infty$.
\end{lemma}

\begin{proof}
For notation purposes, we define $\bm{y}_g=\bm{y}=\left(y_1,\ldots,y_n\right)^{\T}$, $\bm{e}_g = \bm{e}=\left(e_1,\ldots,e_n\right)^{\T}$, $\bm{Z}=\left( \bm{z}_1 \cdots \bm{z}_n\right)^{\T}$, $\mu_i = \bm{z}_i^{\T}\left(\bm{\xi}_{g}+\bm{A}\bm{\ell}_{g}\right)$, $r_i=r_{gi}$ for all $i \in [n]$ and $\bm{c}_{\perp} =\bm{C}_{2_{\bigcdot k}}$. Define $\bm{R} = \diag\left( r_1,\ldots,r_n\right)$ and let $\hat{\ell}$ be the $t+1$st regression coefficient from the ordinary least squares regression of $\bm{y}$ onto $\left( \bm{Z}, \bm{c}_{\perp}\right)$ that ignores missing data. That is,
\begin{align}
\label{equation:OLS_MNAR}
    \left(\hat{\bm{\beta}}, \hat{\ell} \right)^{\T} = \left\lbrace \left( \bm{Z}, \bm{c}_{\perp}\right)^{\T} \bm{R} \left( \bm{Z}, \bm{c}_{\perp}\right) \right\rbrace^{-1}\begin{pmatrix}
    \bm{Z}^{\T} \bm{R}\bm{y}\\
    \bm{c}_{\perp}^{\T} \bm{R}\bm{y}
    \end{pmatrix}.
\end{align}
Note that
\begin{align*}
    \bm{c}_{\perp} = P_{Z}^{\perp}\bm{c}, \quad \bm{c} = \bm{\Xi}\left(n^{-1}\bm{\Xi}^{\T} P_{Z}^{\perp}\bm{\Xi} \right)^{-1/2}\bm{u}_k,
\end{align*}
where $\bm{u}_k \in \mathbb{R}^K$ is the $k$th column of $\bm{U}$ defined in the statement of Lemma \ref{lemma:MetabMiss:PCA}. By the assumptions of $\bm{\Xi}$ in Assumption \ref{assumption:Basic}, $n^{-1}\bm{\Xi}^{\T} P_{Z}^{\perp}\bm{\Xi} = I_K + O_P\left( n^{-1/2}\right)$ and $\norm{\bm{u}_k-\bm{a}_k}_2\ = O_P\left( n^{-1/2}\right)$, where $\bm{a}_k \in \mathbb{R}^K$ is the $k$th standard basis vector. Since $\bm{c}_{\perp}$ is independent of $\bm{y}$, it is also independent of $\bm{R}$ by Model \eqref{equation:MissingMech}, meaning it suffices to assume $\bm{\Xi}$ is independent of $\bm{y}$. Further, we can re-write \eqref{equation:OLS_MNAR} as
\begin{align*}
     \left(\hat{\bm{\beta}}, \hat{\ell} \right)^{\T} = \begin{pmatrix}
    I_t & \bm{s}\\
    \bm{0} & I_K
    \end{pmatrix} \left\lbrace \left( \bm{Z}, \bm{c}\right)^{\T} \bm{R} \left( \bm{Z}, \bm{c}\right) \right\rbrace^{-1}\begin{pmatrix}
    \bm{Z}^{\T} \bm{R}\bm{y}\\
    \bm{c}^{\T} \bm{R}\bm{y}
    \end{pmatrix}, \quad \bm{s} = \left( \bm{Z}^{\T} \bm{Z}\right)^{-1}\bm{Z}^{\T} \bm{c}.
\end{align*}
Therefore, to understand the distribution of $\hat{\ell}$, it suffices to replace $\bm{c}_{\perp}$ with $\bm{c}$ in \eqref{equation:OLS_MNAR}.\par
\indent Define the function
\begin{align*}
    \bm{h}\left(\bm{\beta}, \ell \right) = n^{-1}\left(\bm{Z} , \bm{c} \right)^{\T} \bm{R}\left( \bm{y} - \bm{Z}\bm{\beta} - \bm{c}\ell \right)
\end{align*}
and let
\begin{align*}
    \bm{\beta}_{r=1} = \left\lbrace \E\left( \bm{Z}^{\T}\bm{R}\bm{Z} \right) \right\rbrace^{-1} \E\left( \bm{Z}^{\T} \bm{R}\bm{y}\right).
\end{align*}
Note that $\mathop{\limsup}\limits_{n \to \infty} \norm{\bm{\beta}_{r=1}}_2 < \infty$. We start by understanding the asymptotic properties of $\bm{h}\left(\bm{\beta}_{r=1}, 0 \right)$, which we can do by analyzing the following:

\begin{enumerate}[label=(\roman*)]
    \item \underline{$n^{-1}\bm{Z}^{\T} \bm{R}\bm{y}$}
    \begin{align*}
        \V\left( n^{-1}\bm{Z}^{\T} \bm{R}\bm{y}\right) = n^{-2}\sum\limits_{i=1}^n \bm{z}_i \bm{z}_i^{\T} \V\left( r_i y_i\right) \preceq n^{-2}\sum\limits_{i=1}^n \bm{z}_i \bm{z}_i^{\T} \E\left( y_i^2\right) \preceq n^{-1}c \left(n^{-1}\bm{Z}^{\T}\bm{Z} \right)
    \end{align*}
    where $c > 0$ is a constant. Therefore, $n^{-1}\bm{Z}^{\T} \bm{R}\bm{y} = \E\left( n^{-1}\bm{Z}^{\T} \bm{R}\bm{y}\right) + O_P\left(n^{-1/2} \right)$.\label{item:MNAR:ZtRy}
    \item \underline{$n^{-1}\bm{c}^{\T} \bm{R}\left(\bm{y}-\bm{Z}\bm{\beta}_{r=1} \right)$}
    Since $\bm{\Xi}$ is independent of $y_1,\ldots,y_n$ (and therefore $r_1,\ldots,r_n$),
    \begin{align}
        n^{-1/2}\bm{\Xi}^{\T}\bm{R}\left(\bm{y}-\bm{Z}\bm{\beta}_{r=1}\right) = O_P\left(1\right).
    \end{align}
    Therefore,
    \begin{align*}
        n^{-1}\bm{c}^{\T} \bm{R}\left(\bm{y}-\bm{Z}\bm{\beta}_{r=1} \right) &= \bm{u}_k^{\T} \left(n^{-1}\bm{\Xi}^{\T} P_{Z}^{\perp}\bm{\Xi} \right)^{-1/2}n^{-1}\bm{\Xi}^{\T} \bm{R}\left(\bm{y}-\bm{Z}\bm{\beta}_{r=1} \right)\\
        &= n^{-1}\sum\limits_{i=1}^n  \bm{\Xi}_{ik} r_i \left( y_i - \bm{z}_i^{\T} \bm{\beta}_{r=1}\right) + O_P\left(n^{-1} \right)
    \end{align*},
    where $\E\left\lbrace \bm{\Xi}_{ij} r_i \left( y_i - \bm{x}_i^{\T} \bm{\beta}_{r=1}\right) \right\rbrace = 0$ and 
    \begin{align*}
        n^{-1}\sum\limits_{i=1}^n \bm{\Xi}_{ij}^2 r_i \left( y_i - \bm{z}_i^{\T} \bm{\beta}_{r=1}\right)^2 = n^{-1}\sum\limits_{i=1}^n \E\left\lbrace r_i \left( y_i - \bm{z}_i^{\T} \bm{\beta}_{r=1}\right)^2\right\rbrace + O_P\left(n^{-1/2} \right)
    \end{align*}
    by the bounded fourth moment assumptions. Let $\alpha=\alpha_{g} > 0, \delta = \delta_{g} \in \mathbb{R}$ and $\gamma_i = \mu_i - \bm{z}_i^{\T} \bm{\beta}_{r=1}$. Note that $\max_{i \in [n]}\abs{\gamma_i}, \max_{i \in [n]}\abs{\mu_i} \leq c$ for some constant $c > 0$ by Assumption \ref{assumption:Basic}. Let $M > 0$ be large enough so that 
    \begin{align*}
        \max_{i \in [n]}\left[ \E\left\lbrace \left(e_1 + \gamma_i \right)^2 \right\rbrace I\left( e_1 \geq -M\right) \right] \geq c_1^{-1}.
    \end{align*}
    Then because $e_1,\ldots,e_n$ are identically distributed,
    \begin{align*}
        \E\left\lbrace r_i \left( y_i - \bm{z}_i^{\T} \bm{\beta}_{r=1}\right)^2\right\rbrace &= \E\left[\Psi\left\lbrace \alpha\left(e_{i} + \mu_i-\delta \right) \right\rbrace \left(e_i + \gamma_i \right)^2 \right]\\
        & \geq \E\left[\Psi\left\lbrace \alpha\left(e_{i} + \mu_i-\delta \right) \right\rbrace \left(e_i + \gamma_i \right)^2 I\left( e_i \geq -M\right) \right]\\
        & \geq \Psi\left[ \alpha\left\lbrace-M + \min_{j \in [n]}\left(\mu_j\right)-\delta \right\rbrace \right]c_1^{-1},
    \end{align*}
    where 
    \begin{align*}
        \mathop{\liminf}\limits_{n \to \infty}\left\lbrace \min_{j \in [n]}\left(\mu_j\right)\right\rbrace > -\infty \Rightarrow \mathop{\liminf}\limits_{n \to \infty} \Psi\left[ \alpha\left\lbrace-M + \min_{j \in [n]}\left(\mu_j\right)-\delta \right\rbrace \right] > 0.
    \end{align*}
    By the Lindeberg-Feller Central Limit Theorem, we get that
    \begin{align*}
        &n^{-1/2}v_n^{-1/2}\bm{c}^{\T} \bm{R}\left( \bm{y}-\bm{Z}\bm{\beta}_{r=1}\right) \edist N\left( 0,1\right) + o_P(1)\\
        & v_n = n^{-1}\sum\limits_{i=1}^n \E\left\lbrace r_i \left( y_i - \bm{z}_i^{\T} \bm{\beta}_{r=1}\right)^2\right\rbrace,
    \end{align*}
    where $\mathop{\liminf}\limits_{n \to \infty}v_n > 0$ by the above work.\label{item:ctRy}
    \item We see that $n^{-1}\bm{Z}^{\T} \bm{R}\bm{Z} = \E\left( n^{-1}\bm{Z}^{\T} \bm{R}\bm{Z}\right) + O_P\left( n^{-1/2}\right)$, where $\E\left( n^{-1}\bm{Z}^{\T} \bm{R}\bm{Z}\right) \succeq c n^{-1}\bm{Z}^{\T} \bm{Z}$ for some constant $c > 0$ by Lemma \ref{lemma:MetabMiss:Pz}. An analysis identical to that in \ref{item:ctRy} shows $n^{-1}\bm{Z}^{\T} \bm{R}\bm{c} = O_P\left( n^{-1/2}\right)$.\label{item:MNAR:therest}
    \item \underline{$n^{-1}\bm{c}^{\T} \bm{R}\bm{c}$}
    Let $n_{r=1} = \sum\limits_{i=1}^n r_i$.
    \begin{align*}
        n^{-1}\bm{c}^{\T} \bm{R}\bm{c} &= \bm{u}_k^{\T} \left(n^{-1}\bm{\Xi}^{\T} P_{Z}^{\perp}\bm{\Xi} \right)^{-1/2}n^{-1}\bm{\Xi}^{\T} \bm{R}\bm{\Xi}\left(n^{-1}\bm{\Xi}^{\T} P_{Z}^{\perp}\bm{\Xi} \right)^{-1/2} \bm{u}_k\\
        & = n^{-1}\sum\limits_{i=1}^n \Prob\left( r_i=1 \mid \bm{Z}\right) + O_P\left( n^{-1/2}\right) = n^{-1}n_{r=1} + O_P\left( n^{-1/2}\right).
    \end{align*}\label{item:MNAR:ctRc}
\end{enumerate}

This shows that
\begin{align*}
    n^{1/2}\begin{pmatrix} \hat{\bm{\beta}} - \bm{\beta}_{r=1}\\ \hat{\ell} \end{pmatrix} = \begin{pmatrix} n^{-1}\bm{Z}^{\T} \bm{R}\bm{Z} & n^{-1}\bm{Z}^{\T} \bm{R}\bm{c}\\
    n^{-1}\bm{c}^{\T} \bm{R}\bm{Z} & n^{-1}\bm{c}^{\T} \bm{R}\bm{c}\end{pmatrix}^{-1}\left\lbrace \begin{matrix} n^{-1/2}\bm{Z}^{\T}\bm{R}\left( \bm{y}-\bm{Z}\bm{\beta}_{r=1} \right) \\ n^{-1/2}\bm{c}^{\T}\bm{R}\left( \bm{y}-\bm{Z}\bm{\beta}_{r=1} \right) \end{matrix} \right\rbrace.
\end{align*}
The results from \ref{item:MNAR:ZtRy}, \ref{item:ctRy}, \ref{item:MNAR:therest} and \ref{item:MNAR:ctRc} and a little algebra then show that as $n \to \infty$,
\begin{align*}
    &\left\lbrace  \bm{M}_{(t+1)(t+1)} \tilde{v}_n\right\rbrace^{-1/2}\hat{\ell} \edist N(0,1) + o_P(1), \quad \bm{M} = \begin{pmatrix} \bm{Z}^{\T} \bm{R}\bm{Z} & \bm{Z}^{\T} \bm{R}\bm{c}\\
    \bm{c}^{\T} \bm{R}\bm{Z} & \bm{c}^{\T} \bm{R}\bm{c}\end{pmatrix}^{-1}\\
    &\tilde{v}_n = \left(n_{r=1}-t-1 \right)^{-1}\sum\limits_{i=1}^n r_i \left( y_i - \bm{z}_i^{\T} \bm{\beta}_{r=1}\right)^2.
\end{align*}
Lastly, for $\bm{\delta} = \hat{\bm{\beta}}-\bm{\beta}_{r=1}$ and $\tilde{n}_{r=1} = n_{r=1} - t - 1$,
\begin{align}
    \hat{\tilde{v}}_n =& \tilde{n}_{r=1}^{-1}\sum\limits_{i=1}^n r_i\left( y_i - \bm{z}_i^{\T} \hat{\bm{\beta}} - \bm{c}_i \hat{\ell}\right)^2 = \tilde{v}_n+ 2\tilde{n}_{r=1}^{-1}\sum\limits_{i=1}^n r_i \left( y_i - \bm{z}_i^{\T} \bm{\beta}_{r=1}\right)\left( \bm{z}_i^{\T}  \bm{\delta} - \bm{c}_i \hat{\ell}\right)\nonumber\\
    \label{equation:MetabMiss:vtilde}
    &  + \tilde{n}_{r=1}^{-1}\sum\limits_{i=1}^n r_i \left( \bm{z}_i^{\T}  \bm{\delta} + \bm{c}_i \hat{\ell}\right)^2.
\end{align}
By \ref{item:MNAR:ZtRy}, \ref{item:ctRy}, \ref{item:MNAR:therest} and \ref{item:MNAR:ctRc}, $\norm{\bm{\delta}}_2, \norm{\hat{\ell}}_2 = O_P\left( n^{-1/2}\right)$, meaning the the last two terms are $o_P(1)$ as $n \to \infty$. This completes the proof.
\end{proof}

\begin{theorem}
\label{theorem:MetabMiss:OLS_MNAR}
Suppose Assumption \ref{assumption:Basic} holds and let $g \in \Missing$ and $\bm{C}_2$ be as defined in \eqref{equation:MetabMiss:Ctilde}. For some $k \in \left[K_{\miss}\right]$, define $\hat{z}_g$ to be the ordinary least squares z-score for the $[t+1]$st regressor from the regression $\bm{y}_g \sim \left( \bm{Z}\, \hat{\bm{C}}_{2_{\bigcdot k}} \right)$, restricted to only the observed values of $\bm{y}_g$. Then under the null hypothesis that $\bm{C}_{2_{\bigcdot k}}$ or $\hat{\bm{C}}_{2_{\bigcdot k}}$ is independent of $\bm{y}_g$, $\hat{z}_g \tdist N_1\left( 0,1\right)$ as $n,p \to \infty$.
\end{theorem}

\begin{proof}
Let $z_g$, $\bm{R}$, $r_i$, $\bm{y}$ and $y_i$ be as defined in Lemma \ref{lemma:MetabMiss:OLS_MNAR_observed}. Both $\hat{z}_g$ and $z_g$ can be written as $\frac{\text{``estimate"}}{\text{``standard error of estimate"}}$. Define
\begin{align*}
    \hat{\bm{M}} = \begin{pmatrix}
    \bm{Z}^{\T} \bm{R} \bm{Z} & \bm{Z}^{\T}\bm{R}\hat{\bm{C}}_{2_{\bigcdot k}}\\
    \hat{\bm{C}}_{2_{\bigcdot k}}^{\T} \bm{R} \bm{Z} & \hat{\bm{C}}_{2_{\bigcdot k}}^{\T} \bm{R} \hat{\bm{C}}_{2_{\bigcdot k}}
    \end{pmatrix}^{-1}, \quad \bm{M} = \begin{pmatrix}
    \bm{Z}^{\T} \bm{R} \bm{Z} & \bm{Z}^{\T}\bm{R}\bm{C}_{2_{\bigcdot k}}\\
    \bm{C}_{2_{\bigcdot k}}^{\T} \bm{R} \bm{Z} & \bm{C}_{2_{\bigcdot k}}^{\T} \bm{R} \bm{C}_{2_{\bigcdot k}}
    \end{pmatrix}^{-1}.
\end{align*}
The numerators of $\hat{z}_g$ and $z_g$ are the $[t+1]st$ elements of
\begin{align*}
    \hat{\bm{M}}\begin{pmatrix} \bm{Z}^{\T} \bm{R} \bm{y} \\ \hat{\bm{C}}_{2_{\bigcdot k}}^{\T} \bm{R} \bm{y} \end{pmatrix}, \quad \bm{M}\begin{pmatrix} \bm{Z}^{\T} \bm{R} \bm{y} \\ \bm{C}_{2_{\bigcdot k}}^{\T} \bm{R} \bm{y} \end{pmatrix}
\end{align*}
and the denominators are
\begin{align*}
    \left[  \hat{\bm{M}}_{(t+1)(t+1)} \left\lbrace \left( n_{r=1}-t-1 \right)^{-1} \left( \bm{R} \bm{y}\right)^{\T} P_{R\left( Z,\, \hat{C}_{2_{\bigcdot k}}\right)}^{\perp} \left( \bm{R} \bm{y}\right) \right\rbrace\right]^{1/2},\\
    \left[ \bm{M}_{(t+1)(t+1)} \left\lbrace \left( n_{r=1}-t-1 \right)^{-1} \left( \bm{R} \bm{y}\right)^{\T} P_{R\left( Z,\, C_{2_{\bigcdot k}}\right)}^{\perp} \left( \bm{R} \bm{y}\right) \right\rbrace\right]^{1/2},
\end{align*}
where $n_{r=1} = \sum\limits_{i=1}^n r_i$. Note that for $\hat{\tilde{v}}_n$ defined in \eqref{equation:MetabMiss:vtilde},
\begin{align*}
  \hat{\tilde{v}}_n &= \left( n_{r=1}-t-1 \right)^{-1}\left( \bm{R} \bm{y}\right)^{\T} P_{R\left( Z,\, \left[C_2 \right]_{\bigcdot k}\right)}^{\perp} \left( \bm{R} \bm{y}\right).
\end{align*}
\indent First, if $ \hat{\bm{C}}_{2_{\bigcdot k}}$ is independent of $\bm{y}$, then it suffices to assume that $\bm{\Xi}$ is independent of both $\bm{y}$ and $\bm{R}$. Therefore, by Lemma \ref{lemma:MetabMiss:OLS_MNAR_observed}, we simply have to show the following to prove the theorem:
\begin{subequations}
\begin{align}
    \label{equation:MetabMiss:Midff}
    &\norm{n\hat{\bm{M}} - n\bm{M}}_2 = o_P\left( n^{-1/2}\right)\\
    \label{equation:MetabMiss:ldiff}
    &n^{-1}\hat{\bm{C}}_{2_{\bigcdot k}}^{\T} \bm{R} \bm{y} = n^{-1}\bm{C}_{2_{\bigcdot k}}^{\T} \bm{R} \bm{y} + o_P\left( n^{-1/2}\right)\\
    \label{equation:MetabMiss:vdiff}
    & \left( n_{r=1}-t-1 \right)^{-1} \left( \bm{R} \bm{y}\right)^{\T} P_{R\left( Z,\, \hat{C}_{2_{\bigcdot k}}\right)}^{\perp} \left( \bm{R} \bm{y}\right) =  \hat{\tilde{v}}_n + o_P(1).
\end{align}
\end{subequations}
\indent We start by showing \eqref{equation:MetabMiss:Midff}. Define $\hat{\bm{c}} = \hat{\bm{C}}_{2_{\bigcdot k}}$ and $\bm{c} = \bm{C}_{2_{\bigcdot k}}$. Then by Lemma \ref{lemma:MetabMiss:PCA} and \eqref{equation:MetabMiss:C2hat},
\begin{align*}
    n^{-1}\hat{\bm{c}}^{\T} \bm{R} \hat{\bm{c}} =& n^{-1}\hat{\bm{v}}_{\bigcdot k}^{\T} \bm{C}_2^{\T} \bm{R}\bm{C}_2 \hat{\bm{v}}_{\bigcdot k} + 2n^{-1/2}\hat{\bm{v}}_{\bigcdot k}^{\T} \bm{C}_2^{\T} \bm{R}\bm{Q}_{C_2} \hat{\bm{w}}_{\bigcdot k}\\
    & + \hat{\bm{w}}_{\bigcdot k}^{\T} \bm{Q}_{C_2}^{\T} \bm{R}\bm{Q}_{C_2} \hat{\bm{w}}_{\bigcdot k}
\end{align*}
By \eqref{equation:MetabMiss:wvResults:bound}, the third term is $o_P\left( n^{-1/2}\right)$. Similarly, since $\norm{n^{-1/2} \bm{R}\bm{C}_2}_2 \leq \norm{n^{-1/2}\bm{C}_2}_2 = 1$, we can use \eqref{equation:MetabMiss:wvResults:bound} to get
\begin{align*}
    n^{-1}\hat{\bm{c}}^{\T} \bm{R} \hat{\bm{c}} = n^{-1}\bm{c}^{\T} \bm{R} \bm{c} + 2n^{-1/2}\bm{c}^{\T} \bm{R}\bm{Q}_{C_2} \hat{\bm{w}}_{\bigcdot k} + o_P\left(n^{-1/2} \right).
\end{align*}
We then use \eqref{equation:MetabMiss:wvResults:w} to show that the second term in the above expression is $o_P\left(n^{-1/2} \right)$. The proof of this follows from the fact that $\bm{c}^{\T} \bm{R}$ is independent of $\bm{E}_{\Observed}$. The details are nearly identically to those used to prove Lemma \ref{lemma:MetabMiss:PCA} and are ommitted. An identical technique can also be used to show that $\norm{n^{-1}\bm{Z}^{\T} \bm{R}\hat{\bm{c}} - n^{-1}\bm{Z}^{\T} \bm{R}\bm{c}}_2 = o_P\left( n^{-1/2}\right)$, which proves \eqref{equation:MetabMiss:Midff}.\par
\indent To show \eqref{equation:MetabMiss:ldiff}, we again use \eqref{equation:MetabMiss:C2hat}, which shows that
\begin{align*}
    n^{-1}\hat{\bm{c}}^{\T} \bm{R}\bm{y} = n^{-1}\hat{\bm{v}}_{\bigcdot k}^{\T} \bm{C}_2^{\T} \bm{R}\bm{y} + n^{-1/2} \hat{\bm{w}}_{\bigcdot k}^{\T} \bm{Q}_{C_2}^{\T} \bm{R}\bm{y} =& n^{-1}\bm{c}^{\T} \bm{R}y + n^{-1/2} \hat{\bm{w}}_{\bigcdot k}^{\T} \bm{Q}_{C_2}^{\T} \bm{R}\bm{y}\\
    & + o_P\left( n^{-1/2}\right),
\end{align*}
where the second equality follows from \eqref{equation:MetabMiss:wvResults:bound}. Again, the proof that $n^{-1/2} \hat{\bm{w}}_{\bigcdot k}^{\T} \bm{Q}_{C_2}^{\T} \bm{R}\bm{y} = o_P\left( n^{-1/2}\right)$ follows from \eqref{equation:MetabMiss:wvResults:w} the fact that $\bm{R}\bm{y}$ is independent of $\bm{E}_{\Observed}$, and is omitted.\par 
\indent To show \eqref{equation:MetabMiss:vdiff}, let $\tilde{n} = n_{r=1} - t - 1$. Then
\begin{align*}
    \hat{x}_n &= \tilde{n}^{-1}\left( \bm{R} \bm{y}\right)^{\T} P_{R\left( Z,\, \hat{C}_{2_{\bigcdot k}}\right)}^{\perp} \left( \bm{R} \bm{y}\right)\\
    & =\tilde{n}^{-1}\bm{y}^{\T}\bm{R}\bm{y} -  n^{-1}\left( \tilde{n}^{-1/2}\bm{R}\bm{y}\right)^{\T} \left(\bm{Z},\, \hat{\bm{c}} \right)\left(n \hat{\bm{M}} \right)\left(\bm{Z},\, \hat{\bm{c}} \right)^{\T} \left( \tilde{n}^{-1/2}\bm{R}\bm{y}\right).
\end{align*}
First, $\norm{\tilde{n}^{-1/2}\bm{R}\bm{y}}_2 \leq \norm{\tilde{n}^{-1/2}\bm{y}}_2 = O_P(1)$. Next, because $\norm{n\hat{\bm{M}} - n\bm{M}} = o_P\left( n^{-1/2}\right)$, $\norm{n^{-1/2}\bm{Z}}_2 = O(1)$ and $\norm{n^{-1/2}\hat{\bm{c}}}_2 = 1$,
\begin{align*}
    \hat{x}_n =\tilde{n}^{-1}\bm{y}^{\T}\bm{R}\bm{y}- n^{-1}\left( \tilde{n}^{-1/2}\bm{R}\bm{y}\right)^{\T} \left(\bm{Z},\, \hat{\bm{c}} \right)\left(n \bm{M} \right)\left(\bm{Z},\, \hat{\bm{c}} \right)^{\T} \left( \tilde{n}^{-1/2}\bm{R}\bm{y}\right) + o_P\left( n^{-1/2}\right).
\end{align*}
Lastly, \eqref{equation:MetabMiss:C2hat} and \eqref{equation:MetabMiss:wvResults:bound} imply
\begin{align*}
 \hat{x}_n =\tilde{n}^{-1}\bm{y}^{\T}\bm{R}\bm{y}- n^{-1}\left( \tilde{n}^{-1/2}\bm{R}\bm{y}\right)^{\T} \left(\bm{Z},\, \bm{c} \right)\left(n \bm{M} \right)\left(\bm{Z},\, \bm{c} \right)^{\T} \left( \tilde{n}^{-1/2}\bm{R}\bm{y}\right) + o_P\left( 1\right) = \hat{\tilde{v}}_n + o_P\left( 1\right),
\end{align*}
which completes the proof.
\end{proof}

\begin{theorem}
\label{theorem:Supp:OrderEffects}
Fix a $g \in \Missing$, suppose Assumption \ref{assumption:Basic} holds with $\bm{Z}=\bm{1}_n$ and let $\sigma$ be a permutation of $\left[K_{\miss}\right]$ such that $\abs{\Corr\left\lbrace y_{g1},\bm{\Xi}_{1\sigma(1)} \mid r_{g1}=1\right\rbrace} \geq \cdots \geq \abs{\Corr\left\lbrace y_{g1},\bm{\Xi}_{1\sigma\left(K_{\miss}\right)} \mid r_{g1}=1\right\rbrace}$. For $p_{g,k}$ defined in Step \ref{item:InstrumentSelection:Pvalue} of Algorithm \ref{algorithm:SelectU}, suppose the corresponding q-value from Step \ref{item:InstrumentSelection:Qvalue}, $q_{g,k}$, is defined to be 
\begin{align*}
    q_{g,k} = \frac{p_{g,k}\hat{\pi}_{0,k}}{\hat{F}_k\left(p_{g,k}\right)}, \quad g\in\Missing;k \in \left[K_{\miss}\right],
\end{align*}
where $\hat{\pi}_{0,k},\hat{F}_k(x) \in [0,1]$ are estimates for
\begin{align*}
    \pi_{0,k} &= p^{-1}\sum\limits_{g\in\Missing} I\left( \hat{\bm{C}}_{\miss_{\bigcdot k}} \indep \bm{y}_g \right), \quad k \in \left[K_{\miss}\right]\\
    F_k(x) &= p^{-1}\sum\limits_{g\in\Missing} \Prob\left( p_{g,k} \leq x \right)\quad k \in \left[K_{\miss}\right]; x \in [0,1].
\end{align*}
Assume the following hold:
\begin{enumerate}[label=(\roman*)]
    \item There exists $s_1,s_2 > 0$ such that $\left(p^{s_1}\hat{\pi}_{0,k}\right)^{-1}, \left\lbrace p^{s_1}\hat{F}_{k}\left(\mathop{\min}\limits_{g\in\Missing}p_{g,k}\right)\right\rbrace^{-1} = o_P(1)$ as $n \to \infty$ for all $k \in \left[ 0,K_{\miss}\right]$.\label{item:Supp:qvalue}
    \item $\abs{\Corr\left\lbrace y_{g1},\bm{\Xi}_{1\sigma(2)} \mid r_{g1}=1\right\rbrace}> \abs{\Corr\left\lbrace y_{g1},\bm{\Xi}_{1\sigma(3)} \mid r_{g1}=1\right\rbrace}$.\label{item:Supp:Corr}
\end{enumerate}
Then for $g_1,g_2$ defined in Step \ref{item:InstrumentSelection:Indices} of Algorithm \ref{algorithm:SelectU},
\begin{align*}
    \lim_{n \to \infty}\Prob\left[ g_1,g_2 \in \left\lbrace \sigma(1),\sigma(2) \right\rbrace\right] = 1.
\end{align*}
\end{theorem}

\begin{remark}
\label{remark:qvalue}
Item \ref{item:Supp:qvalue} is a weak condition, since $\hat{F}_{k}\left(\mathop{\min}\limits_{g\in\Missing}p_{g,k}\right)$ is typically $\abs{\Missing}^{-1}$ \citep{qvalue}.
\end{remark}

\begin{remark}
\label{remark:IndicesZ}
Item \ref{item:Supp:Corr} has an analogue when $\bm{Z} \neq \bm{1}_n$, although it is not as intuitive as when $\bm{Z} = \bm{1}_n$.
\end{remark}

\begin{proof}
Without loss of generality, assume $\sigma$ is the identity. For $\hat{\bm{C}}_2$ defined in \eqref{equation:MetabMiss:C2hat}, let $\hat{\bm{c}}_k = \hat{\bm{C}}_{2_{\bigcdot k}}$ for all $k \in \left[K_{\miss}\right]$. Then for $\bm{R}_g = \diag\left(r_{g1},\ldots,r_{gn}\right)$ and $\tilde{n}=\sum\limits_{i=1}^n r_{gi} - 2$, define $z_{g,k}^2$ for each $k \in \left[K_{\miss}\right]$ to be
\begin{align*}
    x_{g,k} &= \tilde{n}^{-1}\bm{y}_g^{\T}\bm{R}_g P_{R_g1_n}^{\perp}\bm{R}_g\hat{\bm{c}}_k \left(\hat{\bm{c}}_k^{\T}\bm{R}_gP_{R_g1_n}^{\perp}\bm{R}_g \hat{\bm{c}}_k\right)^{-1} \hat{\bm{c}}_k^{\T}\bm{R}_gP_{R_g1_n}^{\perp}\bm{R}_g\bm{y}_g\\
    \tilde{n}^{-1}z_{g,k}^2 &= \frac{x_{g,k}}{\tilde{n}^{-1}\left\lbrace \bm{y}_g^{\T}\bm{R}_gP_{R_g1_n}^{\perp}\bm{R}_g\bm{y}_g - \tilde{n}x_{g,k} \right\rbrace}.
\end{align*}
Then the \textit{P} values defined in Step \ref{item:InstrumentSelection:Pvalue} of Algorithm \ref{algorithm:SelectU} are $p_{g,k} = 2\Phi\left(-\abs{z_{g,k}}^{1/2}\right)$, where $\Phi$ is the probit function. We first note that $n^{-1}\tilde{n} = \Prob\left(r_{g1}=1\right) + o_{a.s.}(1)$ where $\Prob\left(r_{g1}=1\right) > 0$ by Lemma \ref{lemma:MetabMiss:Pz} and
\begin{align*}
    \tilde{n}^{-1}\bm{\Xi}_{\bigcdot k}^{\T} \bm{R}_g P_{R_g1_n}^{\perp} \bm{R}_g \bm{\Xi}_{\bigcdot k} &= \tilde{n}^{-1}\sum\limits_{i=1}^n r_{gi}\bm{\Xi}_{ik}^2 - \left( \tilde{n}^{-1}\sum\limits_{i=1}^n r_{gi}\bm{\Xi}_{ik} \right)^2 = \V\left(\bm{\Xi}_{1k} \mid r_{g1}=1 \right) + o_{a.s.}(1)\\
    \tilde{n}^{-1}\bm{y}_g^{\T}\bm{R}_g P_{R_g1_n}^{\perp}\bm{R}_g\hat{\bm{c}}_k &= \tilde{n}^{-1}\sum\limits_{i=1}^n r_{gi}y_{gi}\bm{\Xi}_{ik} - \left(\tilde{n}^{-1}\sum\limits_{i=1}^n r_{gi}y_{gi}\right)\left(\tilde{n}^{-1}\sum\limits_{i=1}^n r_{gi}\bm{\Xi}_{ik}\right)\\
    & = \C\left(\bm{\Xi}_{i1},y_{g1}\mid r_{g1}=1\right) + o_{a.s.}\left(1\right)
\end{align*}
as $n \to \infty$. Next, by Lemma \ref{lemma:MetabMiss:PCA} and Assumption \ref{assumption:Basic}, it is easy to show that $\norm{\hat{\bm{c}}_k - \bm{\Xi}_{\bigcdot k}}_2 = o_P\left(n^{1/2}\right)$. Therefore, since $\V\left(\bm{\Xi}_{1k} \mid r_{g1}=1 \right) > 0$,
\begin{align*}
    x_{g,k} = \Corr\left(y_{g1},\bm{\Xi}_{1k} \mid r_{g1}=1\right)^2 \V\left(y_{g1}\mid r_{g1}=1\right) + o_P(1), \quad k\in\left[K_{\miss}\right]
\end{align*}
as $n \to \infty$. Therefore, for $k \in [2]$ and $t \in \left\lbrace 3,\ldots,K_{\miss} \right\rbrace$
\begin{align*}
    \abs{z_{g_k}} - \abs{z_{g_t}} \geq n^{1/2}\left\lbrace\abs{\Corr\left( y_{g1},\bm{\Xi}_{1k} \mid r_{g1}=1 \right)} - \abs{\Corr\left( y_{g1},\bm{\Xi}_{1t}  \mid r_{g1}=1 \right)}\right\rbrace \left\lbrace 1+o_P(1) \right\rbrace.
\end{align*}
The theorem then holds because for all $s > 0$,
\begin{align*}
    \frac{\Phi\left(-2\abs{z_{g_k}}\right)}{\Phi\left(-2\abs{z_{g_t}}\right)}p^s = o_P(1), \quad k \in [2]; t \in \left\lbrace 3,\ldots,K_{\miss} \right\rbrace
\end{align*}
as $n \to \infty$.
\end{proof}

\section{Approximating the sample moments with independent normal distributions}
\label{section:Supp:NormalProof}
In this section, we justify approximating the distribution of $\bar{\bm{h}}\left(\alpha_g,\delta_g \right) \mid \left(\alpha_g,\delta_g \right)$ with a normal distribution. We also show that for $\left\lbrace g_1,\ldots,g_r \right\rbrace \subseteq \Missing$ a set of finite cardinality, $\bar{\bm{h}}\left(\alpha_{g_1},\delta_{g_1} \right),\ldots,\bar{\bm{h}}\left(\alpha_{g_r},\delta_{g_r} \right)$ are asymptotically independent as $n,p \to \infty$, conditional on $\left(\alpha_{g_1},\delta_{g_1} \right),\ldots,\left(\alpha_{g_r},\delta_{g_r} \right)$. This helps to justify our hierarchical Bayesian generalized method of moments (HB-GMM) procedure in Section \ref{subsection:HBGMM}. Our main result is Theorem \ref{theorem:MetabMiss:Normal}. We first place an assumption on the smoothness of $\Psi(x)$, which is used throughout this and the next section.

\begin{assumption}
\label{assumption:MetabMiss:MissMech}
$\Psi(x)$ is twice continuously differentiable with bounded first and second derivatives. Further, for some large constants $M_1,M_2 > 0$,
\begin{enumerate}[label=(\roman*)]
    \item Either $a\abs{x}^k \Psi(x)=1+R(x)$ or $a\exp\left(k\abs{x}\right)\Psi(x)=1+R(x)$ for all $x \in \left(-\infty,-M_1\right)$, where $\mathop{\lim}\limits_{x \to -\infty}R(x) = 0$ and $\abs{dR(x)/dx }, \abs{d^2R(x)/dx^2 } \leq M_1$ for some $a,k > 0$.
    \item $\E\left(\left[ \Psi\left\lbrace \left(\alpha_{g}+M_1^{-1}\right)\left( y_{g1} - \delta_{g}\right) \right\rbrace \right]^{-\left(3+M_2^{-1}\right)}\right) < \infty$ for all $g \in \Missing$.
\end{enumerate}
\end{assumption}

\begin{remark}
\label{remark:MetabMiss:Bound}
Under Assumption \ref{assumption:Basic}, one can show Assumption \ref{assumption:MetabMiss:MissMech} holds for the following values of $\Psi(x)$:
\begin{enumerate}[label=(\roman*)]
    \item If $\Psi(x)=\exp(x)/\left\lbrace 1+\exp(x) \right\rbrace$, Assumption \ref{assumption:MetabMiss:MissMech} holds if the entries of $\bm{\Xi}$ and $\bm{e}_g$ have a moment generating function that is defined on all of $\mathbb{R}$.
    \item If $\Psi(x)=F_{\nu}(x)$, Assumption \ref{assumption:MetabMiss:MissMech} holds if $\E\left(\abs{ \bm{\Xi}_{1k} }^{3\nu + \epsilon}\right), \E\left(\abs{ e_{g1} }^{3\nu + \epsilon}\right) < \infty$ for some $\epsilon > 0$ for all $k \in [K]$.
\end{enumerate}
\end{remark}

\begin{lemma}
\label{lemma:MetabMissPsi}
Fix a $g \in \Missing$ and let $M_1,M_2$ be as defined in Assumption \ref{assumption:MetabMiss:MissMech}. Under Assumptions \ref{assumption:Basic} and \ref{assumption:MetabMiss:MissMech},
\begin{align*}
    \limsup_{n \to \infty}\left[\max_{i \in [n]}  \left\lbrace \E\left(\left[ \Psi\left\lbrace \left(\alpha_{g}+M_1^{-1}\right)\left( y_{gi} - \delta_{g}\right) \right\rbrace \right]^{-\left(3+M_2^{-1}\right)}\right) \right\rbrace \right] < \infty.
\end{align*}
\end{lemma}

\begin{proof}
Let $\mu_i = \bm{Z}_{i\bigcdot}^{\T} \left(\bm{\xi}_{g} + \bm{A}\bm{\ell}_{g}\right)$ and $\tilde{e}_{gi}=y_{gi} - \mu_i$ for each $i \in [n]$. Then $\tilde{e}_{g1},\ldots,\tilde{e}_{gn}$ are identically distributed and $\mathop{\liminf}\limits_{n \to \infty}\mathop{\min}\limits_{i\in[n]}\mu_i \geq \mu > -\infty$ because the entries of $\bm{Z}$ are uniformly bounded. Therefore, for any $i \in [n]$,
\begin{align*}
    \Psi\left\lbrace \left(\alpha_{g}+M_1^{-1}\right)\left( y_{gi} - \delta_{g}\right) \right\rbrace &\geq \Psi\left[ \left(\alpha_{g}+M_1^{-1}\right)\left\lbrace \tilde{e}_{gi}+\mu_1 + \left(\mu-\mu_1\right) - \delta_{g}\right\rbrace \right]\\
    & \edist \Psi\left[ \left(\alpha_{g}+M_1^{-1}\right)\left\lbrace y_{g1} + \left(\mu-\mu_1\right) - \delta_{g}\right\rbrace \right].
\end{align*}
The result then follows because $\mu-\mu_1$ is finite.
\end{proof}

\begin{lemma}
\label{lemma:MetabMiss:NormalKnown}
Suppose $t < 3$, fix a $g \in \Missing$ and let $g_1,\ldots,g_{3-t} \in \left[K_{\miss}\right]$. For $\bm{C}_2$ defined in \eqref{equation:MetabMiss:Ctilde} and $w_{gi}\left( \alpha,\delta\right) = r_{gi}/\Psi\left\lbrace \alpha\left(y_{gi}-\delta\right) \right\rbrace$, define
\begin{align*}
    &\bm{u}_{gi} = \left( \bm{Z}_{i\bigcdot}^{\T} ,\, \bm{C}_{2_{ig_1}}, \cdots, \bm{C}_{2_{ig_{3-t}}} \right)^{\T} \in \mathbb{R}^3, \quad i \in [n]\\
    &\tilde{\bm{h}}_g\left( \alpha,\delta \right) = n^{-1}\sum\limits_{i=1}^n \bm{u}_{gi}\left\lbrace 1-w_{gi}\left( \alpha,\delta\right) \right\rbrace\\
    &\tilde{\bm{\Sigma}}_g\left( \alpha,\delta \right) = n^{-1}\sum\limits_{i=1}^n \left[ \bm{u}_{gi}\left\lbrace 1-w_{gi}\left( \alpha,\delta\right) \right\rbrace - \tilde{\bm{h}}_g\left( \alpha,\delta \right) \right] \left[ \bm{u}_{gi}\left\lbrace 1-w_{gi}\left( \alpha,\delta\right) \right\rbrace - \tilde{\bm{h}}_g\left( \alpha,\delta \right) \right]^{\T}.
\end{align*}
Then if Assumptions \ref{assumption:Basic} and \ref{assumption:MetabMiss:MissMech} hold,
\begin{align*}
    & \norm{n\V\left\lbrace \tilde{\bm{h}}_g\left( \alpha_g,\delta_g \right)\right\rbrace - \tilde{\bm{\Sigma}}_g\left( \alpha_g,\delta_g \right)}_2 = o_{P}(1), \quad \norm{n\V\left\lbrace \tilde{\bm{h}}_g\left( \alpha_g,\delta_g \right)\right\rbrace}_2,\, \norm{\left[ n\V\left\lbrace \tilde{\bm{h}}_g\left( \alpha_g,\delta_g \right)\right\rbrace \right]^{-1}}_2 \leq c\\
    &n^{1/2}\left\lbrace \tilde{\bm{\Sigma}}_g\left( \alpha_g,\delta_g \right) \right\rbrace^{-1/2} \tilde{\bm{h}}_g\left( \alpha_g,\delta_g \right) \tdist N_3\left( \bm{0}_3, I_3\right)
\end{align*}
as $n \to \infty$, where $c > 0$ is a constant that does not depend on $n$. Further, if $\left\lbrace g_1,\ldots,g_r \right\rbrace \subseteq \Missing$ is a set of at most finite cardinality, $\tilde{\bm{h}}_{g_1}\left( \alpha_{g_1},\delta_{g_1} \right),\ldots,\tilde{\bm{h}}_{g_r}\left( \alpha_{g_r},\delta_{g_r} \right)$ are asymptotically independent as $n \to \infty$.
\end{lemma}

\begin{proof}
Since $K_{\miss}$ is at most finite, it suffices to assume $\bm{u}_{gi} = \left(\bm{Z}_{i \bigcdot}^{\T}, \bm{C}_{2_{i1}},\ldots,\bm{C}_{2_{iK_{\miss}}} \right)^{\T}$ to prove the lemma. Let $\bm{U}_{\miss} = \left(\bm{U}_{\bigcdot 1} \cdots \bm{U}_{\bigcdot K_{\miss}}\right)$, where $\bm{U}$ is as defined in the statement of Lemma \ref{lemma:MetabMiss:PCA}. Then as defined,
\begin{align*}
    \bm{u}_{gi} = \hat{\bm{M}}^{\T} \begin{pmatrix}\bm{Z}_{i \bigcdot}\\
   \bm{\Xi}_{i \bigcdot}\end{pmatrix}, \quad \hat{\bm{M}} = \begin{pmatrix}
   I_t & -\left(\bm{Z}^{\T}\bm{Z} \right)^{-1}\bm{Z}^{\T} \bm{\Xi}\left(\bm{\Xi}^{\T} P_Z^{\perp}\bm{\Xi}\right)^{-1/2}\bm{U}_{\miss}\\
   \bm{0} & \left(\bm{\Xi}^{\T} P_Z^{\perp}\bm{\Xi}\right)^{-1/2}\bm{U}_{\miss}
   \end{pmatrix}.
\end{align*}
And since 
\begin{align*}
    \norm{\hat{\bm{M}}^{\T} - \begin{pmatrix}I_t & \bm{0} & \bm{0}\\\bm{0} & \bm{I}_{K_{\miss}} & \bm{0} \end{pmatrix} }_2 = o_P(1)
\end{align*}
as $n \to \infty$, it suffices to further simplify the problem and assume $\bm{u}_{gi} = \left(\bm{Z}_{i \bigcdot}^{\T}, \bm{\Xi}_{i1},\ldots, \bm{\Xi}_{iK} \right)^{\T}$, meaning $\tilde{\bm{h}}_g\left( \alpha_g,\delta_g \right)$ is an average of independent random variables. Further, for $\Data = \left\lbrace \bm{Y},\bm{C},\bm{Z} \right\rbrace$ and any $g \neq s \in \Missing$,
\begin{align*}
    \C\left\lbrace \tilde{\bm{h}}\left(\alpha_g,\delta_g\right), \tilde{\bm{h}}\left(\alpha_s,\delta_s\right) \right\rbrace =& \C\left[ \E\left\lbrace \tilde{\bm{h}}\left(\alpha_g,\delta_g\right) \mid \Data \right\rbrace, \E\left\lbrace \tilde{\bm{h}}\left(\alpha_s,\delta_s\right) \mid \Data \right\rbrace \right]\\
    &+ \E\left[ \C\left\lbrace \tilde{\bm{h}}\left(\alpha_g,\delta_g\right), \tilde{\bm{h}}\left(\alpha_s,\delta_s\right) \mid \Data \right\rbrace \right] = \bm{0}
\end{align*}
because $\E\left\lbrace \tilde{\bm{h}}\left(\alpha_g,\delta_g\right) \mid \Data \right\rbrace = \bm{0}$ and $\tilde{\bm{h}}\left(\alpha_g,\delta_g\right) \indep \tilde{\bm{h}}\left(\alpha_s,\delta_s\right) \mid \Data$. Therefore, to prove the lemma, we need only check that the Lindeberg condition holds and that $\norm{n\V\left\lbrace \tilde{\bm{h}}_g\left(\alpha_g, \delta_g\right)\right\rbrace - \tilde{\bm{\Sigma}}\left(\alpha_g, \delta_g\right)}_2= o_P(1)$.\par
\indent Let $\bm{v} = \left(\bm{v}_1^{\T}, \, \bm{v}_2^{\T} \right)^{\T} \in \mathbb{R}^{t+K}$ be a unit vector, where $\bm{v}_1 \in \mathbb{R}^t$ and $\bm{v}_2 \in \mathbb{R}^{K_{\miss}}$. First,
\begin{align*}
    n\V\left\lbrace \bm{v}^{\T}\tilde{\bm{h}}_g\left(\alpha_g, \delta_g\right)\right\rbrace &= n^{-1}\sum\limits_{i=1}^n \E\left[ \left\lbrace 1-w_{gi}\left(\alpha_g, \delta_g \right) \right\rbrace^2 \left( \bm{v}^{\T} \bm{u}_{gi}\right)^2 \right]\\
    & \leq n^{-1}\sum\limits_{i=1}^n\left(\E\left[ \left\lbrace 1-w_{gi}\left(\alpha_g, \delta_g \right) \right\rbrace^4 \right] \right)^{1/2}\left[\E\left\lbrace \left( \bm{v}^{\T} \bm{u}_{gi}\right)^4 \right\rbrace \right]^{1/2}.
\end{align*}
We see that
\begin{align*}
   \E\left[ \left\lbrace w_{gi}\left(\alpha_g, \delta_g \right) \right\rbrace^4 \right] = \E\left( \left[\Psi\left\lbrace \alpha_g\left(y_{gi}-\delta_{gi}\right) \right\rbrace\right]^{-3} \right), \, \E\left\lbrace \left( \bm{v}^{\T} \bm{u}_{gi}\right)^4 \right\rbrace \leq c
\end{align*}
for some constant $c$ that does not depend on $i$ or $n$ by Lemma \ref{lemma:MetabMissPsi} and Assumption \ref{assumption:Basic}, meaning $n\V\left\lbrace \bm{v}^{\T}\tilde{\bm{h}}_g\left(\alpha_g, \delta_g\right)\right\rbrace$ exists and is bounded from above. Next, let $M > 0$ be a large constant. Then for $\tilde{e}_{gi}$ as defined in Lemma \ref{lemma:MetabMissPsi}, and because $\E\left( y_{gi} \right)$ is uniformly bounded from below, 
\begin{align*}
        n\V\left\lbrace \bm{v}^{\T}\tilde{\bm{h}}_g\left(\alpha_g, \delta_g\right)\right\rbrace &= n^{-1}\sum\limits_{i=1}^n \E\left[ \left\lbrace 1-w_{gi}\left(\alpha_g, \delta_g \right) \right\rbrace^2 \left( \bm{v}^{\T} \bm{u}_{gi}\right)^2 \right]\\
        &\geq n^{-1}\sum\limits_{i=1}^n \E\left[  \left\lbrace 1-w_{gi}\left(\alpha_g, \delta_g \right) \right\rbrace^2 \left( \bm{v}^{\T} \bm{u}_{gi}\right)^2 I\left( \tilde{e}_{gi}\geq-M \right)\right]\\
        &\geq \eta_M \E\left\lbrace \left( \bm{v}^{\T} \bm{u}_{g1}\right)^2 I\left( \tilde{e}_{g1}\geq-M \right) \right\rbrace
\end{align*}
where $\eta_M > 0$ for all $M$. And since $\E\left\lbrace \left( \bm{v}^{\T} \bm{u}_{g1}\right)^2 I\left( \tilde{e}_{g1}\geq-M \right) \right\rbrace \geq c_2$ for some constant $c_2 > 0$ for all $M$ large enough, $n\V\left\lbrace \bm{v}^{\T}\tilde{\bm{h}}_g\left(\alpha_g, \delta_g\right)\right\rbrace \geq c_2\eta_M$. This proves that the eigenvalues of $n\V\left\lbrace \tilde{\bm{h}}_g\left(\alpha_g, \delta_g\right)\right\rbrace$ are uniformly bounded above 0 and below infinity.\par
\indent We next prove that for $g \in \Missing$,
\begin{align*}
    n^{1/2}\left[ n\V\left\lbrace \tilde{\bm{h}}_g\left(\alpha_g, \delta_g\right) \right\rbrace \right]^{-1/2}\tilde{\bm{h}}_g\left(\alpha_g, \delta_g\right) \tdist N_{t+K_{\miss}}\left( \bm{0}_{t+K_{\miss}},I_{t+K_{\miss}}\right).
\end{align*}
The proof that $\tilde{\bm{h}}_{g_1}\left(\alpha_{g_1}, \delta_{g_1}\right),\ldots,\tilde{\bm{h}}_{g_r}\left(\alpha_{g_r}, \delta_{g_r}\right)$ are asymptotically independent and jointly normal is a simple extension and is omitted. To do this, we need only prove that the Lindeberg condition holds. We note that $\bm{v}^{\T}\bm{u}_{gi} = \bm{v}_1^{\T}\bm{Z}_{i \bigcdot} + \bm{v}_2^{\T} \bm{\Xi}_{i \bigcdot}$, where for $\norm{\bm{Z}_{i \bigcdot}}_2 \leq c_z$,
\begin{align*}
    \left( \bm{v}^{\T}\bm{u}_{gi}\right)^2 \leq \left(c_z + \bm{v}_2^{\T} \bm{\Xi}_{i \bigcdot}\right)^2 I\left( \bm{v}_2^{\T} \bm{\Xi}_{i \bigcdot} \geq 0\right) + \left(c_z - \bm{v}_2^{\T} \bm{\Xi}_{i \bigcdot}\right)^2 I\left( \bm{v}_2^{\T} \bm{\Xi}_{i \bigcdot} < 0\right). 
\end{align*}
For the remainder of the proof, we let $\tilde{e}_{gi}=e_{gi}+\bm{\Xi}_{i\bigcdot}^{\T}\bm{\ell}_{g}$, $\mu_i$ and $\mu$ be as defined in the proof of Lemma \ref{lemma:MetabMissPsi}, and let $\tilde{\mu} = \mathop{\limsup}\limits_{n \to \infty}\left( \mathop{\max}\limits_{i\in[n]}\mu_i \right)$. For each $i \in [n]$, we define
\begin{align*}
    X_i = \left\lbrace 1-r_{gi}w_{gi}\left(\alpha_g, \delta_g \right) \right\rbrace^2 \left( \bm{v}^{\T} \bm{u}_{gi}\right)^2 = \left(1-r_{gi}\left[\Psi\left\lbrace \alpha_{g}\left(\mu_i+\tilde{e}_{gi} -\delta_{g}\right) \right\rbrace\right]^{-1} \right)^2 \left( \bm{v}^{\T} \bm{u}_{gi}\right)^2.
\end{align*}
Next, define
\begin{align*}
    r_{gi}^{(m)}  \begin{cases}
    = 1 & \text{if $r_{gi} = 1$}\\
    \sim \text{Ber}\left[ \frac{\Psi\left\lbrace \alpha_{g}\left(\tilde{\mu}+\tilde{e}_{gi} -\delta_{g}\right) \right\rbrace - \Psi\left\lbrace \alpha_{g}\left(\mu_i+\tilde{e}_{gi} -\delta_{g}\right) \right\rbrace}{1-\Psi\left\lbrace \alpha_{g}\left(\mu_i+\tilde{e}_{gi} -\delta_{g}\right) \right\rbrace} \right] & \text{if $r_{gi}=0$}
    \end{cases},
\end{align*}
where $r_{gi} \leq r_{gi}^{(m)}$ and conditional on $e_{gi}$ and $\bm{\Xi}_{i \bigcdot}$, $r_{gi}^{(m)} \sim \text{Ber}\left[\Psi\left\lbrace \alpha_{g}\left(\tilde{\mu}+\tilde{e}_{gi} -\delta_{g}\right) \right\rbrace \right]$. Lastly, define
\begin{align*}
    X_i^{(m)} =& \left(1-r_{gi}^{(m)}\left[\Psi\left\lbrace \alpha_{g}\left(\mu+\tilde{e}_{gi} -\delta_{g}\right) \right\rbrace\right]^{-1} \right)^2 \left\lbrace \left(c_z + \bm{v}_2^{\T} \bm{\Xi}_{i \bigcdot}\right)^2 I\left( \bm{v}_2^{\T} \bm{\Xi}_{i \bigcdot} \geq 0\right)\right.\\
    &\left. + \left(c_z - \bm{v}_2^{\T}\bm{\Xi}_{i \bigcdot}\right)^2 I\left( \bm{v}_2^{\T} \bm{\Xi}_{i \bigcdot} < 0\right) \right\rbrace.
\end{align*}
Clearly, $X_1^{(m)},\ldots,X_n^{(m)}$ are independent and identically distributed and $X_i \leq X_i^{(m)}$ for all $i \in [n]$. We also see that
\begin{align*}
    &\E\left\lbrace \left( r_{g1}^{(m)}\left[\Psi\left\lbrace \alpha_{g}\left(\mu+\tilde{e}_{g1} -\delta_{g}\right) \right\rbrace\right]^{-1} \right)^4 \right\rbrace\\
     = &\E\left( \frac{\Psi\left\lbrace  \alpha_{g}\left(\tilde{\mu}+\tilde{e}_{g1} -\delta_{g}\right) \right\rbrace}{\Psi\left\lbrace  \alpha_{g}\left(\mu+\tilde{e}_{g1} -\delta_{g}\right) \right\rbrace} \left[\Psi\left\lbrace \alpha_{g}\left(\mu+\tilde{e}_{g1} -\delta_{g}\right) \right\rbrace\right]^{-3}\right) < \infty
\end{align*}
because $\frac{\Psi\left\lbrace  \alpha_{g}\left(\tilde{\mu}+\tilde{e}_{g1} -\delta_{g}\right) \right\rbrace}{\Psi\left\lbrace  \alpha_{g}\left(\mu+\tilde{e}_{g1} -\delta_{g}\right) \right\rbrace}$ is bounded from above by Assumption \ref{assumption:MetabMiss:MissMech}. This then shows that $E\left\lbrace X_1^{(m)} \right\rbrace < \infty$. Therefore, for any $\eta > 0$,
\begin{align*}
    n^{-1}\sum\limits_{i=1}^n \E\left\lbrace X_i I\left( X_i \geq \eta n\right) \right\rbrace &\leq n^{-1}\sum\limits_{i=1}^n \E\left[ X_i^{(m)} I\left\lbrace X_i^{(m)} \geq \eta n\right\rbrace \right]\\
    & = \E\left[ X_1^{(m)} I\left\lbrace X_1^{(m)} \geq \eta n\right\rbrace \right] \text{$\to 0$ as $n \to \infty$}
\end{align*}
by the dominated convergence theorem. This proves that
\begin{align*}
    n^{1/2}\left[ n\V\left\lbrace \tilde{\bm{h}}_g\left(\alpha_g, \delta_g\right) \right\rbrace \right]^{-1/2}\tilde{\bm{h}}_g\left(\alpha_g, \delta_g\right) \tdist N_{t+\hat{K}}\left( \bm{0}_{t+\hat{K}},I_{t+\hat{K}}\right)
\end{align*}
as $n \to \infty$.\par
\indent We use a standard truncation argument to show that $\norm{\tilde{\bm{\Sigma}}_g\left(\alpha_g, \delta_g\right) - n\V\left\lbrace \tilde{\bm{h}}_g\left(\alpha_g, \delta_g\right) \right\rbrace}_2 = o_P(1)$. Let $\bm{v} = \left(\bm{v}_1^{\T}, \, \bm{v}_2^{\T} \right)^{\T} \in \mathbb{R}^{t+K_{\miss}}$ be a unit vector, where $\bm{v}_1 \in \mathbb{R}^t$ and $\bm{v}_2 \in \mathbb{R}^{K_{\miss}}$, and let $X_i$, $r_{gi}^{(m)}$ and $X_i^{(m)}$ be as defined above. We also define
\begin{align*}
    Y_i = \left\lbrace X_i - \E\left( X_i\right) \right\rbrace I\left\lbrace X_i^{(m)} \leq i\right\rbrace \quad (i=1,\ldots,n).
\end{align*}
Since $X_1^{(m)},X_2^{(m)},\ldots$ are identically distributed and $\E\left\lbrace X_1^{(m)} \right\rbrace < \infty$,
\begin{align}
\label{equation:MetabMiss:io}
    \Prob\left[ \bigcap_{N \geq 1}\bigcup_{n \geq N} \left\lbrace X_n^{(m)} > n \right\rbrace \right] = 0
\end{align}
by Lemma 4.31 of \cite{LalleysNotes}. We also have
\begin{align*}
    \abs{n^{-1}\sum\limits_{i=1}^n \left\lbrace X_i - \E\left( X_i\right) \right\rbrace} \leq & \abs{n^{-1}\sum\limits_{i=1}^n I\left\lbrace X_i^{(m)} > i \right\rbrace} + \abs{n^{-1}\sum\limits_{i=1}^n \E\left( Y_i\right)}\\
    & + \abs{n^{-1}\sum\limits_{i=1}^n \left\lbrace Y_i - \E\left( Y_i\right) \right\rbrace}.
\end{align*}
By \eqref{equation:MetabMiss:io}, the first term is $o_{a.s.}(1)$ as $n \to \infty$. For the second term, we may assume $\mu_1 \leq \cdots \leq \mu_n$ without loss of generality. Define $r_{g1}^{(i)}$ inductively as
\begin{align*}
    r_{g1}^{(1)} &= r_{g1}\\
    r_{g1}^{(i)} &\begin{cases}
    = 1 & \text{if $r_{g1}^{(i-1)} = 1$}\\
    \sim \text{Ber}\left[ \frac{\Psi\left\lbrace \alpha_{g}\left(\mu_i+\tilde{e}_{g1} -\delta_{g}\right) \right\rbrace - \Psi\left\lbrace \alpha_{g}\left(\mu_{i-1}+\tilde{e}_{g1} -\delta_{g}\right) \right\rbrace}{1-\Psi\left\lbrace \alpha_{g}\left(\mu_{i-1}+\tilde{e}_{g1} -\delta_{g}\right) \right\rbrace} \right] & \text{if $r_{g1}^{(i-1)}=0$}
    \end{cases} \quad (i=2,\ldots,n).
\end{align*}
and let
\begin{align*}
    \tilde{X}_i = \left(1-r_{g1}^{(i)}\left[\Psi\left\lbrace \alpha_{g}\left(\mu_i+\tilde{e}_{g1} -\delta_{g}\right) \right\rbrace\right]^{-1} \right)^2 \left( \bm{v}_1^{\T}\bm{Z}_{i\bigcdot}+\bm{v}_2^{\T} \bm{\Xi}_{1 \bigcdot}\right)^2.
\end{align*}
And for $\norm{\bm{Z}_{i\bigcdot}}_2 \leq c_z$, define
\begin{align*}
    r_{g1}^{(m)} =& \begin{cases}
    = 1 & \text{if $r_{g1}^{(n)} = 1$}\\
    \sim \text{Ber}\left[ \frac{\Psi\left\lbrace \alpha_{g}\left(\tilde{\mu}+\tilde{e}_{g1} -\delta_{g}\right) \right\rbrace - \Psi\left\lbrace \alpha_{g}\left(\mu_n+\tilde{e}_{g1} -\delta_{g}\right) \right\rbrace}{1-\Psi\left\lbrace \alpha_{g}\left(\mu_n+\tilde{e}_{g1} -\delta_{g}\right) \right\rbrace} \right] & \text{if $r_{g1}^{(n)} =0$}
    \end{cases}\\
    \tilde{X}_1^{(m)} =& \left(1-r_{g1}^{(m)}\left[\Psi\left\lbrace \alpha_{g}\left(\mu+\tilde{e}_{g1} -\delta_{g}\right) \right\rbrace\right]^{-1} \right)^2 \left\lbrace \left(c_z + \bm{v}_2^{\T}\bm{\Xi}_{1 \bigcdot}\right)^2 I\left( \bm{v}_2^{\T}\bm{\Xi}_{1 \bigcdot} \geq 0\right)\right.\\
    &\left. + \left(c_z - \bm{v}_2^{\T}\bm{\Xi}_{1 \bigcdot}\right)^2 I\left( \bm{v}_2^{\T}\bm{\Xi}_{1 \bigcdot} < 0\right) \right\rbrace.
\end{align*}
Note that
\begin{align*}
    X_i I\left\lbrace X_i^{(m)} \leq i \right\rbrace \edist \tilde{X}_i I\left\lbrace X_1^{(m)} \leq i \right\rbrace \leq X_1^{(m)} \quad (i=1,\ldots,n).
\end{align*}
Since $X_1^{(m)}$ is integrable and $\E\left( X_i\right)$ is uniformly bounded from above, $\abs{n^{-1}\sum\limits_{i=1}^n \E\left( Y_i\right)} \to 0$ as $n \to \infty$ by the dominated convergence theorem. We lastly show $\abs{n^{-1}\sum\limits_{i=1}^n \left\lbrace Y_i - \E\left( Y_i\right) \right\rbrace} = o_{a.s.}(1)$ as $n \to \infty$ to complete the proof. By Kronecker’s Lemma and the Khintchine-Kolmogorov theorem, it suffices to show $\sum\limits_{n=1}^{\infty} n^{-2}\E\left(Y_n^2 \right) < \infty$. And because $\E\left(X_i \right)$ is uniformly bounded and $X_i \leq X_i^{(m)}$, we need only show that
\begin{align*}
    \sum\limits_{n=1}^{\infty} n^{-2}\E\left[\left\lbrace X_n^{(m)} \right\rbrace^2 I\left\lbrace X_n^{(m)} \leq n \right\rbrace \right] < \infty.
\end{align*}
However, this follows from the proof of Theorem 4.30 in \cite{LalleysNotes}.
\end{proof}

\begin{theorem}
\label{theorem:MetabMiss:Normal}
Fix a $g \in \Missing$ and suppose Assumptions \ref{assumption:Basic} and \ref{assumption:MetabMiss:MissMech} hold for $t < 3$ and $\epsilon_{\miss}=0$. Let $g_1,\ldots,g_{3-t} \in \left[K_{\miss} \right]$. For $w_{gi}\left( \alpha,\delta\right)$ defined in the statement of Lemma \ref{lemma:MetabMiss:NormalKnown} and $\hat{\bm{C}}_2$ defined in \eqref{equation:MetabMiss:C2hat}, let
\begin{align*}
    &\hat{\bm{u}}_{gi} = \left( \bm{Z}_i^{\T} ,\, \hat{\bm{C}}_{2_{ig_1}}, \cdots,  \hat{\bm{C}}_{2_{ig_{3-t}}} \right)^{\T} \in \mathbb{R}^3, \quad , \quad g \in \Missing;i \in [n]\\
    &\bar{\bm{h}}_g\left( \alpha, \delta \right) = n^{-1}\sum\limits_{i=1}^n \hat{\bm{u}}_{gi}\left\lbrace 1-w_{gi}\left( \alpha, \delta\right)  \right\rbrace, \quad g \in \Missing\\
    &\hat{\bm{\Sigma}}_g\left( \alpha, \delta \right) = n^{-1}\sum\limits_{i=1}^n \left[ \hat{\bm{u}}_{gi}\left\lbrace 1-w_{gi}\left( \alpha, \delta\right) \right\rbrace - \bar{\bm{h}}_g\left( \alpha, \delta \right) \right] \left[ \hat{\bm{u}}_{gi}\left\lbrace 1-w_{gi}\left( \alpha, \delta\right) \right\rbrace - \bar{\bm{h}}_g\left( \alpha, \delta \right) \right]^{\T}, \quad g \in \Missing.
\end{align*}
Then
\begin{align*}
    n^{1/2}\left\lbrace \hat{\bm{\Sigma}}_g\left( \alpha_g,\delta_g \right) \right\rbrace^{-1/2} \bar{\bm{h}}_g\left( \alpha_g,\delta_g \right) \tdist N_3\left( \bm{0}_3, I_3\right), \quad g \in \Missing
\end{align*}
as $n,p \to \infty$. Further, if $\left\lbrace g_1,\ldots,g_r \right\rbrace \subseteq \Missing$ is a set of at most finite cardinality, then $\bar{\bm{h}}_{g_1}\left( \alpha_{g_1}, \delta_{g_1} \right),\ldots,$ $\bar{\bm{h}}_{g_r}\left( \alpha_{g_r}, \delta_{g_r} \right)$ are asymptotically independent as $n,p \to \infty$.
\end{theorem}

\begin{proof}
As we did in Lemma \ref{lemma:MetabMiss:NormalKnown}, it suffices to re-define $\hat{\bm{u}}_{gi}=\left( \bm{Z}_i^{\T} , \hat{\bm{C}}_{2_{i1}},\ldots,\hat{\bm{C}}_{2_{iK_{\miss}}}\right)^{\T}$. Let
\begin{align*}
    \bm{D} &= \diag\left\lbrace 1-w_{g1}\left(\alpha_g, \delta_g \right), \ldots, 1-w_{gn}\left(\alpha_g, \delta_g \right) \right\rbrace \in \mathbb{R}^{n \times n}\\
    \bm{d} &= \left(1-w_{g1}\left(\alpha_g, \delta_g \right),\, \ldots,\, 1-w_{gn}\left(\alpha_g, \delta_g \right)\right)^{\T} \in \mathbb{R}^{n}.
\end{align*}
By Lemma \ref{lemma:MetabMiss:NormalKnown}, it suffices to show that
\begin{subequations}
\begin{align}
\label{equation:MetabMiss:dshow}
    &\norm{n^{-1/2}\bm{d}^{\T}\left( \hat{\bm{C}}_{2_k} - \bm{C}_{2_k} \right)}_2 = o_P(1), \quad k\in\left[K_{\miss}\right]\\
\label{equation:MetabMiss:Dshow:C}
    &\norm{n^{-1}\hat{\bm{C}}_{2_r}^{\T}\bm{D}^2\hat{\bm{C}}_{2_s} - n^{-1}\bm{C}_{2_r}^{\T}\bm{D}^2\bm{C}_{2_s} }_2 = o_P(1), \quad r,s\in\left[K_{\miss}\right]\\
 \label{equation:MetabMiss:Dshow:Z}
    &\norm{n^{-1}\bm{Z}^{\T}\bm{D}^2\hat{\bm{C}}_{2_k} - n^{-1}\bm{Z}^{\T}\bm{D}^2\bm{C}_{2_k} }_2 = o_P(1), \quad k\in\left[K_{\miss}\right]
\end{align}
\end{subequations}
to prove the theorem. By Assumption \ref{assumption:MetabMiss:MissMech} and Lemma \ref{lemma:MetabMissPsi}, $\norm{\bm{d}}_2 = O_P\left( n^{1/2}\right)$ and $\norm{\bm{D}^2}_2 = o_P\left( n^{1/2}\right)$. The latter follows from the fact under the assumptions on the left hand tail of $\Psi(x)$ in Assumption \ref{assumption:MetabMiss:MissMech},
\begin{align*}
    \E\left[\left\lbrace w_{gi}\left(\alpha_g, \delta_g\right)\right\rbrace^{4+\eta} \right] \leq c, \quad i \in [n]
\end{align*}
for $\eta>0$ small enough and $c > 0$ large enough.\par
\indent We start by showing \eqref{equation:MetabMiss:dshow}. Let $\bm{a}_k \in \mathbb{R}^{K}$ be the $k$th standard basis vector. By \eqref{equation:MetabMiss:C2hat},
\begin{align*}
    \norm{n^{-1/2}\bm{d}^{\T}\left( \hat{\bm{C}}_{2_k} - \bm{C}_{2_k} \right)}_2 \leq  \norm{\bm{d}}_2\norm{\hat{v}_{\bigcdot k} - \bm{a}_k}_2 + \norm{\bm{d}^{\T}\bm{Q}_{P_Z^{\perp}C}\hat{\bm{w}}_{\bigcdot k}}_2, \quad k\in\left[ K_{\miss}\right].
\end{align*}
The first term is $o_P(1)$ by \eqref{equation:MetabMiss:wvResults:bound}. And since $\bm{d}$ is independent of $\bm{E}_{\Observed}$, the second term is also $o_P(1)$ by \eqref{equation:MetabMiss:wvResults:bound}.\par
\indent For \eqref{equation:MetabMiss:Dshow:C} and $r,s \in \left[K_{\miss}\right]$,
\begin{align*}
    \norm{n^{-1}\hat{\bm{C}}_{2_r}^{\T}\bm{D}^2\hat{\bm{C}}_{2_s} - n^{-1}\bm{C}_{2_r}^{\T}\bm{D}^2\bm{C}_{2_s} }_2 \leq& \norm{n^{-1}\hat{\bm{v}}_{\bigcdot r}^{\T}\bm{C}_2^{\T}\bm{D}^2\bm{C}_2 \hat{\bm{v}}_{\bigcdot s} - n^{-1}\bm{C}_{2_{\bigcdot r}}^{\T}\bm{D}^2\bm{C}_{2_{\bigcdot s}} }_2\\
    & + \norm{n^{-1/2}\bm{C}_{2_{\bigcdot r}}^{\T}\bm{D}^2 \bm{Q}_{P_{Z}^{\perp}C}\hat{\bm{w}}_{\bigcdot s}}_2 + \norm{n^{-1/2}\bm{C}_{2_{\bigcdot s}}^{\T}\bm{D}^2 \bm{Q}_{P_{Z}^{\perp}C}\hat{\bm{w}}_{\bigcdot r}}_2\\
    &+\norm{ \hat{\bm{w}}_{\bigcdot r}^{\T}\bm{Q}_{P_{Z}^{\perp}C}^{\T}\bm{D}^2 \bm{Q}_{P_{Z}^{\perp}C}\hat{\bm{w}}_{\bigcdot s} }_2.
\end{align*}
The first and fourth terms are clearly $o_P(1)$ by \eqref{equation:MetabMiss:wvResults:bound}. And since $\norm{\bm{D}^2}_2 = o_P\left(n^{1/2}\right)$, the second and third terms are also $o_P(1)$. Identical techniques can be used to show \eqref{equation:MetabMiss:Dshow:Z}, which completes the proof.
\end{proof}

\section{The asymptotic distribution of the generalized method of moments estimator}
\label{section:Supp:EstGMMProof}
Here we prove that under mild assumptions, the two-step generalized method of moments estimators $\hat{\alpha}_g^{\GMM}$ and $\hat{\delta}_g^{\GMM}$, defined in \eqref{equation:TwoStepGMM}, are consistent and asymptotically normal. Our results are analogous to those in \cite{GMM_MNAR}, which assumes the instruments $\hat{\bm{U}}_g$ are observed. Our results are also easier to interpret, since the assumptions we make only involve the moments of $\bm{y}_g$ and the properties of the function $\Psi(x)$. We also show that the generalized method of moments estimators for different metabolites are asymptotically independent, which justifies estimating the prior in Section \ref{subsection:HBGMM} using the product likelihood. We first make a standard assumption regarding the identifiability of $\alpha_g$ and $\delta_g$.

\begin{assumption}
\label{assumption:MetabMiss:Gamma}
Define $\bm{u}_{gi} = \left( \bm{Z}_{i \bigcdot}, \Xi_{ig_1},\ldots, \Xi_{ig_{3-t}}\right)^{\T}$, where the non-random indices $g_1,\ldots,g_{3-t} \in \left[K_{\miss} \right]$ depend on $g \in \Missing$, and
\begin{align*}
    \bm{M}_{g}\left(\alpha,\delta\right) = -\nabla_{\left(\alpha,\delta\right)}\left( n^{-1}\sum\limits_{i=1}^n \E\left[ \bm{u}_{gi}r_{gi}/\Psi\left\lbrace \alpha\left(y_{gi}-\delta\right) \right\rbrace \right]\right), \quad g \in \Missing.
\end{align*}
Then $\bm{M}_g\left(\alpha_g,\delta_g\right)^{\T} \bm{M}_g\left(\alpha_g,\delta_g\right) \succeq \gamma_g I_2$ for some constant $\gamma_g > 0$ that may depend on $g$ but does not depend on $n$ or $p$.
\end{assumption}

\begin{remark}
\label{remark:GammaIdentify}
We prove $\bm{M}_g\left(\alpha_g,\delta_g\right)$ exists in Lemma \ref{lemma:MetabMiss:GUWLLN}. This assumption on the gradient of the population moment is a standard assumption in the generalized method of moment literature \citep{Hansen_2step,GMM_MNAR} and helps to guarantee that $\alpha_g$ and $\delta_g$ are locally identifiable.
\end{remark}

Let $\dot{\Psi}(x)$ and $\ddot{\Psi}(x)$ be the first and second derivatives of $\Psi(x)$ and define $\bm{\theta}_g =  \left(\alpha_g,-\alpha_g\delta_g\right)^{\T}$. For the remainder of the supplement, we define
\begin{subequations}
\label{equation:Supp:Moments}
\begin{align}
\hat{\bm{u}}_{gi} =& \left( \bm{Z}_{i \bigcdot}, \hat{\bm{C}}_{2_{ig_1}},\ldots, \hat{\bm{C}}_{2_{ig_{3-t}}}\right)^{\T},\quad g\in \Missing; i \in [n]\\
\bar{\bm{h}}_g\left(\bm{\theta}\right) =& n^{-1}\sum\limits_{i=1}^n \hat{\bm{u}}_{gi} \left[ 1-r_{gi}\left\lbrace \Psi\left(\bm{\theta}_1 y_{gi} + \bm{\theta}_2\right) \right\rbrace^{-1} \right], \quad g \in \Missing\\
\bm{\Gamma}_g\left( \bm{\theta}\right) =& \nabla_{\bm{\theta}}\bar{\bm{h}}_g\left(\bm{\theta}\right) = n^{-1}\sum\limits_{i=1}^n r_{gi} \frac{\dot{\Psi}\left(\bm{\theta}_1 y_{gi} +  \bm{\theta}_2\right)}{\Psi\left(\bm{\theta}_1 y_{gi} +  \bm{\theta}_2\right)^2}\hat{\bm{u}}_{gi} \left(y_{gi},1\right), \quad g \in \Missing\\
\hat{\bm{\Sigma}}_g\left(\bm{\theta}\right) =& n^{-1}\sum\limits_{i=1}^n \left[ 1-r_{gi}\left\lbrace \Psi\left(\bm{\theta}_1 y_{gi} + \bm{\theta}_2\right) \right\rbrace^{-1} \right]^2 \hat{\bm{u}}_{gi} \hat{\bm{u}}_{gi}^{\T}, \quad g \in \Missing,
\end{align}
\end{subequations}
where $\hat{\bm{C}}_2$ is as defined in \eqref{equation:MetabMiss:C2hat}. Note that $\hat{\bm{\Sigma}}_g\left(\bm{\theta}\right)$ and that defined in \eqref{equation:hbarAsy} differ by a factor of $\bar{\bm{h}}_g\left(\bm{\theta}\right)\bar{\bm{h}}_g\left(\bm{\theta}\right)^{\T}$. Since we are only interested in the behavior of $\hat{\bm{\Sigma}}_g\left(\bm{\theta}\right)$ around $\bm{\theta}=\bm{\theta}_g$, this difference is asymptotically negligible.\par 
\indent For any weight matrix $\bm{W}_g$, the generalized method of moments estimate $\hat{\bm{\theta}}_g^{\GMM}$ satisfies
\begin{align}
\label{equation:MetabMiss:GMMTaylor}
\bm{0} =& \bm{\Gamma}_g\left\lbrace \hat{\bm{\theta}}_g^{\GMM} \right\rbrace^{\T} \bm{W}_g \bar{\bm{h}}_g\left\lbrace \hat{\bm{\theta}}_g^{\GMM} \right\rbrace = \bm{\Gamma}_g\left\lbrace \hat{\bm{\theta}}_g^{\GMM} \right\rbrace^{\T}\bm{W}_g\bar{\bm{h}}_g\left(\bm{\theta}_g\right)\nonumber\\
&+ \bm{\Gamma}_g\left\lbrace \hat{\bm{\theta}}_g^{\GMM} \right\rbrace^{\T}\bm{W}_g \bm{\Gamma}_g\left( \tilde{\bm{\theta}}_g\right)\left\lbrace \hat{\bm{\theta}}_g^{\GMM} - \bm{\theta}_g\right\rbrace
\end{align}
where $\tilde{\bm{\theta}}_g = b \bm{\theta}_g + (1-b)\hat{\bm{\theta}}_g^{\GMM}$ for some $b \in [0,1]$. Since we have already proven that $\bar{\bm{h}}_g\left(\bm{\theta}_g \right)$ is asymptotically normal and that $\bar{\bm{h}}_{g_1}\left(\bm{\theta}_{g_1} \right),\ldots,\bar{\bm{h}}_{g_r}\left(\bm{\theta}_{g_r} \right)$ are asymptotically independent for a distinct, finite set of elements $\left\lbrace g_1,\ldots,g_r \right\rbrace \subseteq \Missing$ in Theorem \ref{theorem:MetabMiss:Normal}, proving the asymptotic normality of $\hat{\bm{\theta}}_g^{\GMM}$ and asymptotic independence of $\hat{\bm{\theta}}_{g_1}^{\GMM},\ldots,\hat{\bm{\theta}}_{g_r}^{\GMM}$ only requires understanding the convergence of $\bm{\Gamma}\left\lbrace \hat{\bm{\theta}}_g^{\GMM}\right\rbrace$ and $\bm{W}_g$ for a fixed $g \in \Missing$.

\begin{lemma}
\label{lemma:MetabMiss:PsiDotPrelim}
Under Assumption \ref{assumption:MetabMiss:MissMech}, there exists a constant $M>0$ such that $\abs{\dot{\Psi}(x)/\Psi(x)}$, $\abs{\ddot{\Psi}(x)/\Psi(x)} \leq M$ for all $x \in \mathbb{R}$.
\end{lemma}

\begin{proof}
Since $\abs{\dot{\Psi}(x)},\abs{\ddot{\Psi}(x)}$ are uniformly bounded, we need only consider the case when $x \to -\infty$. When $\Psi(x) = \abs{x}^{-k}\left\lbrace a+R(x) \right\rbrace$,
\begin{align*}
    \dot{\Psi}(x) &= k\abs{x}^{-(k+1)}\left\lbrace a+R(x) \right\rbrace + \abs{x}^{-k}\frac{dR(x)}{dx}\\
    \ddot{\Psi}(x) &= k(k+1)\abs{x}^{-(k+2)}\left\lbrace a+R(x) \right\rbrace +  2k\abs{x}^{-(k+1)}\frac{dR(x)}{dx} + \abs{x}^{-k}\frac{d^2R(x)}{dx^2}
\end{align*}
and when $\Psi(x) = \exp\left(-k\abs{x}\right)\left\lbrace a+R(x) \right\rbrace$,
\begin{align*}
    \dot{\Psi}(x) &= k\exp\left(-k\abs{x}\right)\left\lbrace a+R(x) \right\rbrace + \exp\left(-k\abs{x}\right)\frac{dR(x)}{dx}\\
    \ddot{\Psi}(x) &= k^2 \exp\left(-k\abs{x}\right)\left\lbrace a+R(x) \right\rbrace + 2k\exp\left(-k\abs{x}\right) \frac{dR(x)}{dx} + \exp\left(-k\abs{x}\right)\frac{d^2R(x)}{dx^2}.
\end{align*}
The result then follows by the assumptions on $R(x)$.
\end{proof}

\begin{lemma}
\label{lemma:MetabMiss:PsiDot}
Fix a $g \in \Missing$, let $B\left( \eta; \bm{x}\right) = \left\lbrace \bm{x}_0 : \norm{\bm{x}-\bm{x}_0}_2 < \eta \right\rbrace$ and suppose Assumptions \ref{assumption:Basic} and \ref{assumption:MetabMiss:MissMech} hold. Let $\bm{u}_{gi} = \left( \bm{Z}_{i \bigcdot}, \bm{\Xi}_{ig_1},\ldots,\bm{\Xi}_{ig_{3-t}}\right)$ for $g_1,\ldots,g_{3-t} \in \left[K_{\miss}\right]$ and define
\begin{align}
\label{equation:MetabMiss:ParamsKnown}
\begin{aligned}
    &\tilde{\bm{h}}_g\left( \bm{\theta}\right) = n^{-1}\sum\limits_{i=1}^n \bm{u}_{gi} \left\lbrace 1- r_{gi}/\Psi\left(\bm{\theta}_1 y_{gi} + \bm{\theta}_2\right) \right\rbrace, \quad \tilde{\bm{\Gamma}}_g\left( \bm{\theta}\right) = \nabla_{\bm{\theta}} \tilde{\bm{h}}_g\left( \bm{\theta}\right)\\
    &\tilde{\bm{\Sigma}}_g\left(\bm{\theta}\right) = n^{-1}\sum\limits_{i=1}^n \left\lbrace 1- r_{gi}/\Psi\left(\bm{\theta}_1 y_{gi} + \bm{\theta}_2\right) \right\rbrace^2 \bm{u}_{gi} \bm{u}_{gi}^{\T}.
\end{aligned}
\end{align}
Then there exists constants $\gamma_*,\eta_* > 0$ such that for all $\eta \in \left(0,\eta_*\right)$,
\begin{subequations}
\label{equation:MetabMiss:SupGGamma}
\begin{align}
    \label{equation:MetabMiss:SupG}
    &\E\left\lbrace \sup_{\bm{\theta} \in B\left(\eta; \bm{\theta}_g\right)} \norm{ \tilde{\bm{h}}_g\left( \bm{\theta}\right) - \tilde{\bm{h}}_g\left( \bm{\theta}_g\right) }_2 \right\rbrace \leq \gamma_*\eta\\
    \label{equation:MetabMiss:SupGamma}
    &\E\left\lbrace \sup_{\bm{\theta} \in B\left(\eta; \bm{\theta}_g\right)} \norm{ \tilde{\bm{\Gamma}}_g\left( \bm{\theta}\right) - \tilde{\bm{\Gamma}}_g\left( \bm{\theta}_g\right) }_2 \right\rbrace \leq \gamma_* \eta\\
    \label{equation:MetabMiss:SupSigma}
    &\E\left\lbrace \sup_{\bm{\theta} \in B\left(\eta; \bm{\theta}_g\right)} \norm{ \tilde{\bm{\Sigma}}_g\left( \bm{\theta}\right) - \tilde{\bm{\Sigma}}_g\left( \bm{\theta}_g\right) }_2 \right\rbrace \leq \gamma_* \eta
\end{align}
\end{subequations}
for all $n$ large enough.
\end{lemma}

\begin{proof}
Fix $\eta  > 0$, let $\bm{\theta} \in B\left(\eta; \bm{\theta}_g\right)$ and let $\bm{v} \in \mathbb{R}^3$ be any unit vector. Then for some set of $\tilde{\bm{\theta}}_i = b_i\bm{\theta} + (1-b_i)\bm{\theta}_g$, $b_i\in [0,1]$, and $\left(\epsilon_1,\epsilon_2\right)^{\T} = \bm{\theta}-\bm{\theta}_g$,
\begin{align*}
    \abs{\bm{v}^{\T}\left\lbrace \tilde{\bm{h}}_g\left( \bm{\theta}\right) - \tilde{\bm{h}}_g\left( \bm{\theta}_g\right) \right\rbrace} = & \abs{n^{-1}\sum\limits_{i=1}^n \left(\epsilon_1 y_{gi} + \epsilon_2\right) \frac{\dot{\Psi}\left(\tilde{\bm{\theta}}_{i_1} y_{gi} +  \tilde{\bm{\theta}}_{i_2}\right)}{\Psi\left(\tilde{\bm{\theta}}_{i_1} y_{gi} +  \tilde{\bm{\theta}}_{i_2}\right)}\frac{r_{gi} }{\Psi\left(\tilde{\bm{\theta}}_{i_1} y_{gi} +  \tilde{\bm{\theta}}_{i_2}\right)}\bm{v}^{\T}\bm{u}_{gi}}\\
    \leq & M \abs{\epsilon_1} n^{-1}\sum\limits_{i=1}^n \frac{r_{gi} }{\Psi\left(\tilde{\bm{\theta}}_{i_1} y_{gi} +  \tilde{\bm{\theta}}_{i_2}\right)}\abs{\bm{v}^{\T}\bm{u}_{gi}}\\
    & + M\abs{\epsilon_2}n^{-1}\sum\limits_{i=1}^n \frac{r_{gi} }{\Psi\left(\tilde{\bm{\theta}}_{i_1} y_{gi} +  \tilde{\bm{\theta}}_{i_2}\right)}\abs{y_{gi}\bm{v}^{\T}\bm{u}_{gi}}
\end{align*}
where $M > 0$ is defined in Lemma \ref{lemma:MetabMiss:PsiDotPrelim}. \eqref{equation:MetabMiss:SupG} then follows easily by Assumptions \ref{assumption:Basic} and \ref{assumption:MetabMiss:MissMech}.\par
\indent For \eqref{equation:MetabMiss:SupGamma},
\begin{align*}
    \tilde{\bm{\Gamma}}_g\left(\bm{\theta}\right) - \tilde{\bm{\Gamma}}_g\left(\bm{\theta}_g\right) =& n^{-1}\sum\limits_{i=1}^n r_{gi} \left(\epsilon_1 y_{gi} + \epsilon_2\right)\left\lbrace -2\frac{\dot{\Psi}\left( \tilde{\bm{\theta}}_{i_1} y_{gi} + \tilde{\bm{\theta}}_{i_2} \right)^2}{\Psi\left( \tilde{\bm{\theta}}_{i_1} y_{gi} + \tilde{\bm{\theta}}_{i_2} \right)^3}\right.\\
    &\left. + \frac{\ddot{\Psi}\left( \tilde{\bm{\theta}}_{i_1} y_{gi} + \tilde{\bm{\theta}}_{i_2} \right)}{\Psi\left( \tilde{\bm{\theta}}_{i_1} y_{gi} + \tilde{\bm{\theta}}_{i_2} \right)^2}\right\rbrace \bm{u}_{gi}\left( y_{gi}, 1\right)
\end{align*}
where $\epsilon_1,\epsilon_2$ are defined above and $\tilde{\bm{\theta}}_i = b_i\bm{\theta} + (1-b_i)\bm{\theta}_g$ for some $b_i \in [0,1]$. To prove \eqref{equation:MetabMiss:SupGamma}, it suffices to show that
\begin{align*}
    n^{-1}\sum\limits_{i=1}^n y_{gi}^2 \frac{r_{gi}}{\Psi\left( \tilde{\bm{\theta}}_{i_1} y_{gi} + \tilde{\bm{\theta}}_{i_2} \right)}\bm{u}_{gi}
\end{align*}
has at most finite expectation by Lemma \ref{lemma:MetabMiss:PsiDotPrelim}. However, this follows because the entries of $\bm{u}_i$ have uniformly bounded sixth moment.\par
\indent Using the same notation as above, we can express \eqref{equation:MetabMiss:SupSigma} as
\begin{align*}
    \tilde{\bm{\Sigma}}_g\left( \bm{\theta}\right) - \tilde{\bm{\Sigma}}_g\left( \bm{\theta}_g\right) =& 2 n^{-1}\sum\limits_{i=1}^n \left(\epsilon_1 y_{gi} + \epsilon_2\right)r_{gi} \left\lbrace \frac{\dot{\Psi}\left( \tilde{\bm{\theta}}_{i_1} y_{gi} + \tilde{\bm{\theta}}_{i_2} \right)}{\Psi\left( \tilde{\bm{\theta}}_{i_1} y_{gi} + \tilde{\bm{\theta}}_{i_2} \right)^2}\right.\\
    &\left. - \frac{\dot{\Psi}\left( \tilde{\bm{\theta}}_{i_1} y_{gi} + \tilde{\bm{\theta}}_{i_2} \right)}{\Psi\left( \tilde{\bm{\theta}}_{i_1} y_{gi} + \tilde{\bm{\theta}}_{i_2} \right)^3} \right\rbrace \bm{u}_{gi} \bm{u}_{gi}^{\T}.
\end{align*}
Again, by Lemma \ref{lemma:MetabMiss:PsiDotPrelim}, it suffices to show that
\begin{align*}
    n^{-1}\sum\limits_{i=1}^n \frac{r_{gi}}{\Psi\left( \tilde{\bm{\theta}}_{i_1} y_{gi} + \tilde{\bm{\theta}}_{i_2} \right)^2}\abs{y_{gi}}\bm{u}_{gi} \bm{u}_{gi}^{\T}
\end{align*}
has bounded expectation. However, this follows by the bounded sixth moment assumption on the entries of $\bm{u}_{gi}$ and Assumption \ref{assumption:MetabMiss:MissMech}.
\end{proof}

\begin{lemma}
\label{lemma:MetabMiss:GUest}
Let $\bar{\bm{h}}_g\left(\bm{\theta}\right),\bm{\Gamma}_g\left(\bm{\theta}\right),\hat{\bm{\Sigma}}\left(\bm{\theta}\right)$ and $\tilde{\bm{h}}_g\left(\bm{\theta}\right),\tilde{\bm{\Gamma}}_g\left(\bm{\theta}\right),\tilde{\bm{\Sigma}}\left(\bm{\theta}\right)$ be as defined in \eqref{equation:Supp:Moments} and \eqref{equation:MetabMiss:ParamsKnown}, respectively. Suppose the assumptions of Lemma \ref{lemma:MetabMiss:PsiDot} hold. Then for $\eta > 0$ small enough,
\begin{align*}
    \sup_{\bm{\theta} \in B\left( \eta; \bm{\theta}_g\right)} \norm{ \bar{\bm{h}}_g\left(\bm{\theta}\right) - \tilde{\bm{h}}_g\left(\bm{\theta}\right) }_2, \, \sup_{\bm{\theta} \in B\left( \eta; \bm{\theta}_g\right)} \norm{ \bm{\Gamma}_g\left(\bm{\theta}\right) - \tilde{\bm{\Gamma}}_g\left(\bm{\theta}\right) }_2,\,  \sup_{\bm{\theta} \in B\left( \eta; \bm{\theta}_g\right)} \norm{ \hat{\bm{\Sigma}}_g\left(\bm{\theta}\right) - \tilde{\bm{\Sigma}}_g\left(\bm{\theta}\right) }_2 = o_P(1)
\end{align*}
as $n,p \to \infty$.
\end{lemma}

\begin{proof}
Define
\begin{align*}
    &\bm{d}_1\left( \bm{\theta} \right) = \left( 1-\frac{r_{g1}}{\Psi\left(\bm{\theta}_1 y_{g1} + \bm{\theta}_2\right)} ,\ldots, 1-\frac{r_{gn}}{\Psi\left(\bm{\theta}_1 y_{gn} + \bm{\theta}_2\right)} \right)^{\T} \in \mathbb{R}^{n}\\
    &\bm{d}_2\left( \bm{\theta} \right) = \left( \left\lbrace 1- \frac{r_{g1} \dot{\Psi}\left(\bm{\theta}_1 y_{g1} + \bm{\theta}_2\right)}{\Psi\left(\bm{\theta}_1 y_{g1} + \bm{\theta}_2\right)^2} \right\rbrace\left(1,y_{g1}\right)^{\T}, \ldots, \left\lbrace 1- \frac{r_{gn} \dot{\Psi}\left(\bm{\theta}_1 y_{gn} + \bm{\theta}_2\right)}{\Psi\left(\bm{\theta}_1 y_{gn} + \bm{\theta}_2\right)^2} \right\rbrace\left(1,y_{gn}\right)^{\T} \right)^{\T}\in \mathbb{R}^{n \times 2}\\
    &\bm{D}\left(\bm{\theta}\right) = \diag\left[ \left\lbrace 1-\frac{r_{g1}}{\Psi\left(\bm{\theta}_1 y_{g1} + \bm{\theta}_2\right)}\right\rbrace^2,\ldots, \left\lbrace 1-\frac{r_{gn}}{\Psi\left(\bm{\theta}_1 y_{gn} + \bm{\theta}_2\right)}\right\rbrace^2 \right] \in \mathbb{R}^{n \times n}.
\end{align*}
\indent For $\bm{a}_k \in \mathbb{R}^K$ the $k$th standard basis vector, $\bm{U}$ defined in the statement of Lemma \ref{lemma:MetabMiss:PCA} and $\hat{\bm{v}},\hat{\bm{w}}$ defined in \eqref{equation:MetabMiss:wvResults},
\begin{align}
\label{equation:Supp:CXiDelta}
    \hat{\bm{C}}_{2_{\bigcdot k}} - \bm{\Xi}_{\bigcdot k} = P_{Z}^{\perp}\bm{\Xi}\left\lbrace \left(n^{-1}\bm{\Xi}^{\T}P_{Z}^{\perp}\bm{\Xi}\right)^{-1/2}\bm{U}\hat{\bm{v}}_{\bigcdot k} - \bm{a}_k \right\rbrace + n^{1/2}\bm{Q}_{P_Z^{\perp}\Xi}\hat{\bm{w}}_{\bigcdot k} + P_{Z}\bm{\Xi}_{\bigcdot k}, \quad k \in \left[K\right].
\end{align}
By Lemma \ref{lemma:MetabMiss:PCA} and the fact that the upper $K_{\miss} \times K_{\miss}$ block of $\bm{U}$ is $I_{K_{\miss}} + O_P\left(n^{-1/2}\right)$ under Assumption \ref{assumption:Basic},
\begin{align}
\label{equation:Supp:Delta}
    \bm{\Delta}_k = \left(n^{-1}\bm{\Xi}^{\T}P_{Z}^{\perp}\bm{\Xi}\right)^{-1/2}\bm{U}\hat{\bm{v}}_{\bigcdot k} - \bm{a}_k = O_P\left(n^{-1/2}\right), \quad k \in \left[ K_{\miss}\right]
\end{align}
as $n \to \infty$. If $j \in [t]$, we see that $ \left\lbrace\bar{\bm{h}}_g\left(\bm{\theta} \right) - \tilde{\bm{h}}_g\left(\bm{\theta} \right)\right\rbrace_{j}$ and $ \left\lbrace\bar{\bm{\Gamma}}_g\left(\bm{\theta} \right) - \tilde{\bm{\Gamma}}_g\left(\bm{\theta} \right)\right\rbrace_{j\bigcdot}$ are 0. Otherwise,
\begin{align*}
    \left\lbrace\bar{\bm{h}}_g\left(\bm{\theta} \right) - \tilde{\bm{h}}_g\left(\bm{\theta} \right)\right\rbrace_{j} =& n^{-1}\bm{\Delta}_{g_j}^{\T}\bm{\Xi}^{\T}\bm{d}_1\left(\bm{\theta}\right) + n^{-1/2}\hat{\bm{w}}_{\bigcdot g_j}^{\T}\bm{Q}_{P_{Z}^{\perp}\Xi}^{\T}\bm{d}_1\left(\bm{\theta}\right)\\
    &+ n^{-1}\left(\bm{a}_{g_j} - \bm{\Delta}_{g_j}\right)^{\T}\bm{\Xi}^{\T}P_{Z}\bm{d}_1\left(\bm{\theta}\right), \quad j \in \left[t+1,3-t\right]\\
    \left\lbrace \bm{\Gamma}_g\left(\bm{\theta} \right) - \tilde{\bm{\Gamma}}_g\left(\bm{\theta} \right)\right\rbrace_{j \bigcdot} =& n^{-1}\bm{\Delta}_{g_j}^{\T}\bm{\Xi}^{\T}\bm{d}_2\left(\bm{\theta}\right) + n^{-1/2}\hat{\bm{w}}_{\bigcdot g_j}^{\T}\bm{Q}_{P_{Z}^{\perp}\Xi}^{\T}\bm{d}_2\left(\bm{\theta}\right)\\
    &+ n^{-1}\left(\bm{a}_{g_j} - \bm{\Delta}_{g_j}\right)^{\T}\bm{\Xi}^{\T}P_{Z}\bm{d}_2\left(\bm{\theta}\right), \quad j \in \left[t+1,3-t\right].
\end{align*}
We first see that
\begin{align*}
   \sup\limits_{\bm{\theta} \in B\left(\eta; \bm{\theta}_g\right)} \norm{ \bm{d}_1\left(\bm{\theta}\right)}_1, \sup\limits_{\bm{\theta} \in B\left(\eta; \bm{\theta}_g\right)}\norm{ \bm{d}_2\left(\bm{\theta}\right)}_2 = O_P\left(n^{1/2}\right)
\end{align*}
by Assumptions \ref{assumption:Basic} and \ref{assumption:MetabMiss:MissMech} for $\eta > 0$ small enough. Therefore, for all $\eta > 0$ small enough and $i=1,2$,
\begin{align*}
    \sup_{\bm{\theta} \in B\left(\eta,\bm{\theta}_g\right)}\norm{ n^{-1/2}\hat{\bm{w}}_{\bigcdot g_j}^{\T}\bm{Q}_{P_{Z}^{\perp}\Xi}^{\T}\bm{d}_i\left(\bm{\theta}\right) + n^{-1}\left(\bm{a}_{g_j} - \bm{\Delta}_{g_j}\right)^{\T}\bm{\Xi}^{\T}P_{Z}\bm{d}_i\left(\bm{\theta}\right)}_2 = o_P(1), \quad j \in [t+1,3-t]
\end{align*}
as $n,p \to \infty$ by Lemma \ref{lemma:MetabMiss:PCA}. Lastly, by Assumption \ref{assumption:MetabMiss:MissMech} and since $y_{g1},\ldots,y_{gn}$ have uniformly bounded fourth moments under Assumption \ref{assumption:Basic},
\begin{align*}
    \sup_{\bm{\theta} \in B\left(\eta,\bm{\theta}_g\right)} \norm{n^{-1}\bm{\Xi}^{\T}\bm{d}_i\left(\bm{\theta}\right)}_2 = O_P\left(1\right), \quad i \in [2]
\end{align*}
as $n \to \infty$ for all $\eta > 0$ small enough. This shows
\begin{align*}
     \sup_{\bm{\theta} \in B\left(\eta,\bm{\theta}_g\right)} \norm{\bar{\bm{h}}_g\left(\bm{\theta}\right) - \tilde{\bm{h}}_g\left(\bm{\theta}\right)}_2, \, \sup_{\bm{\theta} \in B\left(\eta,\bm{\theta}_g\right)} \norm{ \bm{\Gamma}_g\left(\bm{\theta}\right) - \tilde{\bm{\Gamma}}_g\left(\bm{\theta}\right)}_2 = o_P(1)
\end{align*}
for all $\eta > 0$ small enough as $n,p \to \infty$.\par 
\indent For the last relation, let $\bm{U}_g = \left(\bm{u}_{g1}\cdots\bm{u}_{gn}\right)^{\T}$ and $\hat{\bm{U}}_g = \left(\hat{\bm{u}}_{g1}\cdots \hat{\bm{u}}_{gn}\right)^{\T}$, where $\bm{u}_{gi}$ and $\hat{\bm{u}}_{gi}$ defined in the statement of Lemma \ref{lemma:MetabMiss:PsiDot} and \eqref{equation:Supp:Moments}, respectively. Then
\begin{align*}
    \hat{\bm{\Sigma}}_g\left(\bm{\theta}\right) - \tilde{\bm{\Sigma}}_g\left(\bm{\theta}\right) =& n^{-1}\left(\hat{\bm{U}}_g - \bm{U}_g\right)^{\T}\bm{D}\left(\bm{\theta}\right)\bm{U}_g + n^{-1}\left\lbrace \left(\hat{\bm{U}}_g - \bm{U}_g\right)^{\T}\bm{D}\left(\bm{\theta}\right)\bm{U}_g \right\rbrace^{\T}\\
    & + n^{-1}\left(\hat{\bm{U}}_g - \bm{U}_g\right)^{\T}\bm{D}\left(\bm{\theta}\right)\left(\hat{\bm{U}}_g - \bm{U}_g\right).
\end{align*}
First, by Assumption \ref{assumption:MetabMiss:MissMech} and Lemma \ref{lemma:MetabMissPsi}, there exists constants $\gamma,c > 0$ such that
\begin{align*}
    \E\left[ \frac{r_{gi}}{\Psi\left\lbrace\left(\bm{\theta}_{g_1}+\eta\right) y_{gi} + \left(\bm{\theta}_{g_2}-\eta\right)\right\rbrace^{4+\gamma}} \right] \leq c, \quad i \in [n].
\end{align*}
Therefore, $\sup\limits_{\bm{\theta} \in B\left(\eta; \bm{\theta}_g\right)}\norm{ \bm{D}\left(\bm{\theta}\right)}_2 = o_P\left(n^{1/2}\right)$ for $\eta > 0$ small enough. By \eqref{equation:Supp:CXiDelta} and \eqref{equation:Supp:Delta} and because $\norm{n^{1/2}\hat{\bm{w}}_{\bigcdot k}}_2 = O_P(1)$ by Lemma \ref{lemma:MetabMiss:PCA}, $\norm{\hat{\bm{U}}_g - \bm{U}_g}_2 = O_P(1)$ as $n,p \to \infty$. Since $\norm{\bm{U}_g}_2 = O_P\left(n^{1/2}\right)$, this completes the proof.
\end{proof}

\begin{lemma}
\label{lemma:MetabMiss:GUWLLN}
Suppose the assumptions of Lemma \ref{lemma:MetabMiss:GUest} hold and let $\bm{M}_g\left(\bm{\theta}\right)$ be as defined in Assumption \ref{assumption:MetabMiss:Gamma}. Then the following hold for $\bar{\bm{h}}_g\left(\bm{\theta}\right),\tilde{\bm{h}}_g\left(\bm{\theta}\right)$ and $\bm{\Gamma}_g\left(\bm{\theta}\right)$ defined in Lemma \ref{lemma:MetabMiss:GUest}:
\begin{enumerate}[label=(\roman*)]
    \item There exists a constant $\eta > 0$ such that $\E\left\lbrace \tilde{\bm{h}}_g\left(\bm{\theta} \right) \right\rbrace$ and $\bm{M}_g\left(\bm{\theta}\right)$ exist and are continuous for all $\bm{\theta} \in B\left(\eta; \bm{\theta}_g\right)$.\label{item:MetabMiss:GMexist}
    \item There exists a constant $\eta > 0$ such that
    \begin{align*}
        \sup_{\bm{\theta} \in B\left(\eta; \bm{\theta}_g\right)}\norm{ \bar{\bm{h}}_g\left(\bm{\theta}\right) - \E\left\lbrace \tilde{\bm{h}}_g\left(\bm{\theta} \right) \right\rbrace }_2, \, \sup_{\bm{\theta} \in B\left(\eta; \bm{\theta}_g\right)}\norm{ \bm{\Gamma}_g\left(\bm{\theta}\right) - \bm{M}_g\left(\bm{\theta} \right) }_2 = o_P(1)
    \end{align*}
    as $n,p \to \infty$.\label{item:MetabMiss:GMUWLLN}
\end{enumerate}
\end{lemma}

\begin{proof}
The existence and continuity of $\E\left\lbrace \tilde{\bm{h}}_g\left(\bm{\theta} \right) \right\rbrace$ is a direct consequence of \eqref{equation:MetabMiss:SupG}. To show the existence of $\bm{M}_g\left(\bm{\theta}\right)$, let $\left\lbrace \gamma_m \right\rbrace_{m\geq 1}$ be such that $\mathop{\lim}\limits_{m \to \infty} \gamma_m = 0$ and $\gamma_m \neq 0$ for all $m \geq 1$. Then for any $\bm{v} \in \mathbb{R}^2$,
\begin{align*}
    \gamma_m^{-1}\left\lbrace \tilde{\bm{h}}_g\left( \bm{\theta}+\gamma_m\bm{v}\right) - \tilde{\bm{h}}_g\left( \bm{\theta}\right) \right\rbrace = & n^{-1}\sum\limits_{i=1}^n \left(\bm{v}_1 y_{gi} + \bm{v}_2\right) \frac{\dot{\Psi}\left(\tilde{\bm{\theta}}_{i_1} y_{gi} +  \tilde{\bm{\theta}}_{i_2}\right)}{\Psi\left(\tilde{\bm{\theta}}_{i_1} y_{gi} +  \tilde{\bm{\theta}}_{i_2}\right)}\frac{r_{gi} }{\Psi\left(\tilde{\bm{\theta}}_{i_1} y_{gi} +  \tilde{\bm{\theta}}_{i_2}\right)}\bm{u}_{gi},
\end{align*}
where $\tilde{\bm{\theta}}_i = \bm{\theta} + b_i \bm{v}$ for $\abs{b_i }\in \left[ 0,\abs{\gamma_m}\right]$ for all $i \in [n]$. By assumption \ref{assumption:MetabMiss:MissMech}, $\abs{\frac{\dot{\Psi}\left(\tilde{\bm{\theta}}_{i_1} y_{gi} +  \tilde{\bm{\theta}}_{i_2}\right)}{\Psi\left(\tilde{\bm{\theta}}_{i_1} y_{gi} +  \tilde{\bm{\theta}}_{i_2}\right)}} \leq M$ for some constant $M > 0$. Since $\Psi(x)$ is an increasing function,
\begin{align*}
    0 \leq \frac{r_{gi} }{\Psi\left(\tilde{\bm{\theta}}_{i_1} y_{gi} +  \tilde{\bm{\theta}}_{i_2}\right)}I\left(y_{gi} \leq 0\right) \leq \frac{r_{gi} }{\Psi\left\lbrace\left(\bm{\theta}_1-\gamma\right) y_{gi} +  \left(\bm{\theta}_2-\gamma\right)\right\rbrace}I\left(y_{gi} \leq 0\right), \quad i \in [n]
\end{align*}
for some small $\gamma > 0$. An application of the dominated convergence theorem proves $\bm{M}_g\left(\bm{\theta}\right)$ exists for all $\bm{\theta} \in B\left(\eta,\bm{\theta}_g\right)$ for some $\eta > 0$. The continuity of $\bm{M}_g\left(\bm{\theta}\right)$ then follows directly from \eqref{equation:MetabMiss:SupGamma}.\par
\indent To prove \ref{item:MetabMiss:GMUWLLN},
\begin{align*}
    \norm{\bar{\bm{h}}_g\left(\bm{\theta}\right) - \E\left\lbrace \tilde{\bm{h}}_g\left(\bm{\theta}\right) \right\rbrace}_2 &\leq \norm{\bar{\bm{h}}_g\left(\bm{\theta}\right) - \tilde{\bm{h}}_g\left(\bm{\theta}\right)}_2 + \norm{\tilde{\bm{h}}_g\left(\bm{\theta}\right) - \E\left\lbrace \tilde{\bm{h}}_g\left(\bm{\theta}\right) \right\rbrace}_2\\
    \norm{\bar{\bm{h}}_g\left(\bm{\theta}\right) - \bm{M}_g\left(\bm{\theta}\right)}_2 &\leq \norm{\bm{\Gamma}_g\left(\bm{\theta}\right) - \tilde{\bm{\Gamma}}_g\left(\bm{\theta}\right)}_2 + \norm{\tilde{\bm{\Gamma}}_g\left(\bm{\theta}\right) -  \bm{M}_g\left(\bm{\theta}\right) }_2.
\end{align*}
First, $\mathop{\sup}\limits_{\bm{\theta}\in B\left(\eta,\bm{\theta}_g\right)}\norm{\bar{\bm{h}}_g\left(\bm{\theta}\right) - \tilde{\bm{h}}_g\left(\bm{\theta}\right)}_2, \mathop{\sup}\limits_{\bm{\theta}\in B\left(\eta,\bm{\theta}_g\right)}\norm{\bm{\Gamma}_g\left(\bm{\theta}\right) - \tilde{\bm{\Gamma}}_g\left(\bm{\theta}\right)}_2 = o_P(1)$ as $n,p\to \infty$ for $\eta > 0$ small enough by Lemma \ref{lemma:MetabMiss:GUest}. Next, it is easy to use Assumptions \ref{assumption:Basic} and \ref{assumption:MetabMiss:MissMech} to show that for all $\bm{\theta}\in B\left(\eta,\bm{\theta}_g\right)$,
\begin{align*}
    \V\left\lbrace \tilde{\bm{h}}_g\left(\bm{\theta}\right) \right\rbrace, \, \V\left[ \vecM\left\lbrace \tilde{\bm{\Gamma}}_g\left(\bm{\theta}\right) \right\rbrace \right] = o(1)
\end{align*}
as $n \to \infty$. And since $\tilde{\bm{h}}_g\left(\bm{\theta}\right) - \E\left\lbrace \tilde{\bm{h}}_g\left(\bm{\theta}\right) \right\rbrace$ and $\tilde{\bm{\Gamma}}_g\left(\bm{\theta}\right) -  \bm{M}_g\left(\bm{\theta}\right)$ are stochastically equicontinuous when restricted to a small enough compact parameter space that contains $\bm{\theta}_g$ by Lemma \ref{lemma:MetabMiss:PsiDot}, the result follows.
\end{proof}

\begin{theorem}
\label{theorem:MetabMiss:GMMAsy}
Fix a $g \in \Missing$ and suppose Assumptions \ref{assumption:Basic}, \ref{assumption:MetabMiss:MissMech} and \ref{assumption:MetabMiss:Gamma} hold. Define
\begin{align*}
    f_g\left(\bm{\theta}\right) = \bar{\bm{h}}_g\left( \bm{\theta}\right)^{\T} \left( n^{-1}\sum\limits_{i=1}^n \hat{\bm{u}}_{gi}\hat{\bm{u}}_{gi}^{\T} \right)^{-1}\bar{\bm{h}}_g\left( \bm{\theta}\right),
\end{align*}
let $\left\lbrace \hat{\bm{\theta}}_{g_n}^{(1)} \right\rbrace_{n \geq 1}$ be a sequence of minima of $f_g\left(\bm{\theta}\right)$ and define $\bm{W}_{g}^{-1} = \hat{\bm{\Sigma}}_g\left\lbrace \hat{\bm{\theta}}^{(1)}_{g_n} \right\rbrace$. Then there exists a sequence $\left\lbrace \hat{\bm{\theta}}_{g_n}^{\GMM} \right\rbrace_{n \geq 1}$ of minima of $\bar{\bm{h}}_g\left( \bm{\theta}\right)^{\T} \bm{W}_{g}\bar{\bm{h}}_g\left( \bm{\theta}\right)$ such that for
\begin{align*}
    \hat{\bm{V}}_{g_n} = \left[ \bm{\Gamma}_g\left\lbrace \hat{\bm{\theta}}_{g_n}^{\GMM} \right\rbrace^{\T} \bm{W}_g \bm{\Gamma}_g\left\lbrace \hat{\bm{\theta}}_{g_n}^{\GMM} \right\rbrace \right]^{-1} \text{ and } \bm{V}_{g_n} = \left( \bm{M}_g\left( \bm{\theta}_g \right)^{\T} \left[\E\left\lbrace \tilde{\bm{\Sigma}}_g\left(\bm{\theta}_g\right) \right\rbrace\right]^{-1} \bm{M}_g\left( \bm{\theta}_g \right) \right)^{-1},
\end{align*}
\begin{align*}
    n^{1/2}\hat{\bm{V}}_{g_n}^{-1/2}\left\lbrace \hat{\bm{\theta}}_{g_n}^{\GMM} - \bm{\theta}_g \right\rbrace \tdist N_2\left(\bm{0}_2,I_2\right), \quad \norm{\hat{\bm{V}}_{g_n}\bm{V}_{g_n}^{-1} - I_2}_2 = o_P(1)
\end{align*}
as $n \to \infty$.
\end{theorem}

\begin{proof}
By Assumption \ref{assumption:MetabMiss:Gamma} and \ref{item:MetabMiss:GMexist} of Lemma \ref{lemma:MetabMiss:GUWLLN},
\begin{align*}
    \norm{\E\left\lbrace \tilde{\bm{h}}_g\left(\bm{\theta} \right) \right\rbrace}_2 \geq c \norm{\bm{\theta} - \bm{\theta}_g}_2
\end{align*}
for some constant $c > 0$ and $\bm{\theta} \in B\left(\eta; \bm{\theta}_g\right)$ for some $\eta > 0$. By \ref{item:MetabMiss:GMUWLLN} of Lemma \ref{lemma:MetabMiss:GUWLLN}, this implies there exists a minimizer $\hat{\bm{\theta}}_{g_n}^{(1)}$ of $f_g\left(\bm{\theta}\right)$ defined above such that $\norm{\hat{\bm{\theta}}_{g_n}^{(1)} - \bm{\theta}_g}_2 = o_P(1)$ as $n \to \infty$. Next,
\begin{align*}
    \norm{\hat{\bm{\Sigma}}_g\left\lbrace \hat{\bm{\theta}}^{(1)}_{g_n} \right\rbrace - \E\left\lbrace \tilde{\bm{\Sigma}}\left(\bm{\theta}_g\right) \right\rbrace}_2 \leq & \norm{ \tilde{\bm{\Sigma}}_g\left\lbrace \hat{\bm{\theta}}^{(1)}_{g_n} \right\rbrace - \tilde{\bm{\Sigma}}_g\left(\bm{\theta}_g\right) }_2 + \norm{ \hat{\bm{\Sigma}}_g\left\lbrace \hat{\bm{\theta}}^{(1)}_{g_n} \right\rbrace - \tilde{\bm{\Sigma}}_g\left\lbrace \hat{\bm{\theta}}^{(1)}_{g_n} \right\rbrace }_2\\
    &+\norm{\tilde{\bm{\Sigma}}_g\left(\bm{\theta}_g\right) - \E\left\lbrace \tilde{\bm{\Sigma}}\left(\bm{\theta}_g\right) \right\rbrace}_2.
\end{align*}
Each of the above three terms are $o_P(1)$ as $n \to \infty$ by  Lemma \ref{lemma:MetabMiss:PsiDot}, Lemma \ref{lemma:MetabMiss:GUest} and the proof of Lemma \ref{lemma:MetabMiss:NormalKnown}, respectively. Therefore, there exists a minimizer $\hat{\bm{\theta}}_{g_n}^{\GMM}$ of $\bar{\bm{h}}_g\left(\bm{\theta}\right)^{\T} \bm{W}_g \bar{\bm{h}}_g\left(\bm{\theta}\right)$ such that $\norm{ \hat{\bm{\theta}}_{g_n}^{\GMM} - \bm{\theta}_g }_2 = o_P(1)$ as $n \to \infty$. Additionally, for $\tilde{\bm{\theta}} \in b\bm{\theta}_g + (1-b)\hat{\bm{\theta}}^{\GMM}_{g_n}$ for any $b\in[0,1]$,
\begin{align*}
    \norm{\bm{\Gamma}_g\left( \tilde{\bm{\theta}} \right) - \bm{M}_g\left(\bm{\theta}_g\right)}_2 \leq & \norm{ \tilde{\bm{\Gamma}}_g\left( \tilde{\bm{\theta}} \right) - \tilde{\bm{\Gamma}}_g\left(\bm{\theta}_g\right) }_2 + \norm{ \bm{\Gamma}_g\left( \tilde{\bm{\theta}} \right) - \tilde{\bm{\Gamma}}_g\left( \tilde{\bm{\theta}} \right) }_2\\
    &+\norm{\tilde{\bm{\Gamma}}_g\left(\bm{\theta}_g\right) - \bm{M}_g\left(\bm{\theta}_g\right)}_2,
\end{align*}
where the above three terms are $o_P(1)$ as $n \to \infty$ by Lemmas \ref{lemma:MetabMiss:PsiDot}, \ref{lemma:MetabMiss:GUest} and \ref{lemma:MetabMiss:GUWLLN}, respectively. The result then follows by Theorem \ref{theorem:MetabMiss:Normal} and the Taylor expansion in \eqref{equation:MetabMiss:GMMTaylor}.
\end{proof}

\begin{corollary}
\label{corollary:MetabMiss:Jtest}
Fix a $g \in \Missing$ and suppose the assumptions of Theorem \ref{theorem:MetabMiss:GMMAsy} hold. Then for $\bm{W}_g$ and $\hat{\bm{\theta}}_{g_n}^{\GMM}$ defined in the statement of Theorem \ref{theorem:MetabMiss:GMMAsy},
\begin{align*}
    n\bar{\bm{h}}_g\left\lbrace\hat{\bm{\theta}}_{g_n}^{\GMM}\right\rbrace^{\T}\bm{W}_g \bar{\bm{h}}_g\left\lbrace\hat{\bm{\theta}}_{g_n}^{\GMM}\right\rbrace \tdist \chi^2_1
\end{align*}
as $n\to \infty$.
\end{corollary}

\begin{proof}
Let $\tilde{\bm{W}}_g = \E\left\lbrace \tilde{\bm{\Sigma}}_g\left( \bm{\theta}_g\right) \right\rbrace$. Then for
\begin{align*}
    \hat{\bm{A}}_g &= P_{W_n^{1/2} \Gamma_g\left\lbrace\hat{\theta}_{g_n}^{\GMM}\right\rbrace}^{\perp} = I_3 - \bm{W}_g^{1/2}\bm{\Gamma}_g\left\lbrace\hat{\bm{\theta}}_{g_n}^{\GMM}\right\rbrace \left[ \bm{\Gamma}_g\left\lbrace\hat{\bm{\theta}}_{g_n}^{\GMM}\right\rbrace^{\T} \bm{W}_g \bm{\Gamma}_g\left\lbrace\hat{\bm{\theta}}_{g_n}^{\GMM}\right\rbrace \right]^{-1}\bm{\Gamma}_g\left\lbrace\hat{\bm{\theta}}_{g_n}^{\GMM}\right\rbrace^{\T} \bm{W}_g^{1/2}\\
    \tilde{\bm{A}}_g &= P_{\tilde{W}_g^{1/2} M_g\left(\theta_g\right)}^{\perp} = I_3 - \tilde{\bm{W}}_g^{1/2}\bm{M}_g\left( \bm{\theta}_g\right) \left\lbrace \bm{M}_g\left( \bm{\theta}_g\right)^{\T} \tilde{\bm{W}}_g \bm{M}_g\left( \bm{\theta}_g\right) \right\rbrace^{-1}\bm{M}_g\left( \bm{\theta}_g\right)^{\T} \tilde{\bm{W}}_g^{1/2},
\end{align*}
$\norm{\hat{\bm{A}}_g - \tilde{\bm{A}}_g}_2 = o_P(1)$ by the proof of Theorem \ref{theorem:MetabMiss:GMMAsy}. Further, $\tilde{\bm{A}}_g$ is a non-random, rank 1 matrix for all $n$ large enough by Assumption \ref{assumption:MetabMiss:Gamma} and Lemma \ref{lemma:MetabMiss:GUWLLN}. We then get that
\begin{align*}
    n\bar{\bm{h}}_g\left\lbrace\hat{\bm{\theta}}_{g_n}^{\GMM}\right\rbrace^{\T} \bm{W}_g \bar{\bm{h}}_g\left\lbrace\hat{\bm{\theta}}_{g_n}^{\GMM}\right\rbrace =&  n\bar{\bm{h}}_g\left\lbrace\hat{\bm{\theta}}_{g_n}^{\GMM}\right\rbrace^{\T} \bm{W}_n^{1/2} P_{W_g^{1/2} \Gamma\left\lbrace\hat{\theta}_{g_n}^{\GMM}\right\rbrace} \bm{W}_g^{1/2}\bar{\bm{h}}_g\left\lbrace\hat{\bm{\theta}}_{g_n}^{\GMM}\right\rbrace\\
    &+n\bar{\bm{h}}_g\left\lbrace\hat{\bm{\theta}}_{g_n}^{\GMM}\right\rbrace^{\T} \bm{W}_g^{1/2} \hat{\bm{A}}_g \bm{W}_g^{1/2}\bar{\bm{h}}_g\left\lbrace\hat{\bm{\theta}}_{g_n}^{\GMM}\right\rbrace,
\end{align*}
where the first term is zero because $\hat{\bm{\theta}}_{g_n}^{\GMM}$ satisfies $\bm{\Gamma}\left\lbrace\hat{\bm{\theta}}_{g_n}^{\GMM}\right\rbrace^{\T}\bm{W}_g \bar{\bm{h}}_g\left\lbrace\hat{\bm{\theta}}_{g_n}^{\GMM}\right\rbrace = \bm{0}_2$. For the second term,
\begin{align*}
    n^{1/2}\hat{\bm{A}}_g \bm{W}_g^{1/2}\bar{\bm{h}}_g\left\lbrace\hat{\bm{\theta}}_{g_n}^{\GMM}\right\rbrace = n^{1/2}\hat{\bm{A}}_g \bm{W}_g^{1/2}\bar{\bm{h}}_g\left( \bm{\theta}_g\right) + n^{1/2}\hat{\bm{A}}_g \bm{W}_g^{1/2} \bm{\Gamma}\left(\tilde{\bm{\theta}}\right) \left\lbrace\hat{\bm{\theta}}_{g_n}^{\GMM} - \bm{\theta}_g\right\rbrace
\end{align*}
for some $\tilde{\bm{\theta}} = b\bm{\theta}_g + (1-b)\hat{\bm{\theta}}_{g_n}^{\GMM}$, $b\in [0,1]$. The result follows by Theorems \ref{theorem:MetabMiss:OLS_MNAR} and \ref{theorem:MetabMiss:GMMAsy} because
\begin{align*}
    n^{1/2}\bm{W}_g^{1/2}\bar{\bm{h}}_g\left( \bm{\theta}_g\right) \tdist N_3\left(\bm{0}_3,I_3\right),\, n^{1/2}\norm{\hat{\bm{\theta}}_{g_n}^{\GMM} - \bm{\theta}_g}_2 = O_P(1),\, \norm{\bm{\Gamma}_g\left(\tilde{\bm{\theta}}\right) - \bm{\Gamma}\left\lbrace\hat{\bm{\theta}}_{g_n}^{\GMM}\right\rbrace}_2 = o_P(1)
\end{align*}
as $n \to \infty$ and $\hat{\bm{A}}_g \bm{W}_g^{1/2}\bm{\Gamma}_g\left\lbrace\hat{\bm{\theta}}_{g_n}^{\GMM}\right\rbrace = \bm{0}$.
\end{proof}

\end{changemargin}

\end{document}